\pdfoutput=1
\documentclass[
    affil-it, 
    auth-lg, 
    british, 
    compression, 
    contents, 
    custom-packages={multirow,booktabs}, 
    references={bigone}, 
]{lib/preprint}

\tikzset{
    aux/.style={dashed},
    dottedline/.style={dotted},
    gluon/.style={
        decorate, 
        draw=black,
        decoration={
            snake,
            post=lineto,
            post length=0pt,
            segment length=4,
            amplitude=0.9
        }
    }
}

\makecommand{sf}{\mathsf}{{Diff,Aut,Weak,Strong,Jac,Poiss}}
\makecommand{fr}{\mathfrak}{{Kin,Scal,Mod,Col,Gam,KIN,SCAL}}

\let\oldblacksquare\blacksquare
\newcommand{\BBox}{{\textcolor{gray}{\mathop\oldblacksquare\nolimits}}} 
\newcommand{\BVbox}[1][]{%
    \ifthenelse{\equal{#1}{}}{%
        \mathrm{BV}^\BBox%
    }{%
        \mathrm{BV}^{\BBox,\,\mathrm{#1}}
    }%
}

\begin{document}
    
    \hypersetup{
        pdftitle = {Double copy from tensor products of metric BV◼-algebras},
        pdfauthor = {Leron Borsten, Branislav Jurco, Hyungrok Kim, Tommaso Macrelli, Christian Saemann, Martin Wolf},
        pdfkeywords = {double copy, color-kinematics duality, colour-kinematics duality, BV◼-algebras, Batalin-Vilkovisky algebras, Hopf algebras, kinematic L∞-algebra, kinematic Lie algebra, syngamy},
    }
    
    \date{\today}
    
    \email{l.borsten@herts.ac.uk,h.kim2@herts.ac.uk,branislav.jurco@gmail.com,tmacrelli @phys.ethz.ch,c.saemann@hw.ac.uk,m.wolf@surrey.ac.uk}
    
    \preprint{EMPG--23--14,DMUS--MP--23/11}
    
    \title{Double Copy from Tensor Products of Metric $\BVbox$-algebras} 
    
    \author[a]{Leron~Borsten\,\orcidlink{0000-0001-9008-7725}\,}
    \author[b]{Branislav~Jur{\v c}o\,\orcidlink{0000-0001-7782-2326}\,}
    \author[a]{Hyungrok~Kim\,\orcidlink{0000-0001-7909-4510}\,}
    \author[d]{Tommaso~Macrelli\,\orcidlink{0000-0002-4207-0846}\,}
    \author[c]{Christian~Saemann\,\orcidlink{0000-0002-5273-3359}\,}
    \author[e]{Martin~Wolf\,\orcidlink{0009-0002-8192-3124}\,}
    
    \affil[a]{Department of Physics, Astronomy, and Mathematics,\\University of Hertfordshire, Hatfield AL10 9AB, United Kingdom}
    \affil[b]{Mathematical Institute, Faculty of Mathematics and Physics,\\ Charles University, Prague 186 75, Czech Republic}
    \affil[c]{Maxwell Institute for Mathematical Sciences, Department of Mathematics,\\ Heriot--Watt University, Edinburgh EH14 4AS, United Kingdom}
    \affil[d]{Institute for Theoretical Physics, ETH Zurich, 8093 Zurich, Switzerland}
    \affil[e]{School of Mathematics and Physics,\\ University of Surrey, Guildford GU2 7XH, United Kingdom}

    \abstract{Field theories with kinematic Lie algebras, such as field theories featuring colour--kinematics duality, possess an underlying algebraic structure known as $\BVbox$-algebra. If, additionally, matter fields are present, this structure is supplemented by a module for the $\BVbox$-algebra. We explain this perspective, expanding on our previous work and providing many additional mathematical details. We also show how the tensor product of two metric $\BVbox$-algebras yields the action of a new syngamy field theory, a construction which comprises the familiar double copy construction. As examples, we discuss various scalar field theories, Chern--Simons theory, self-dual Yang--Mills theory, and the pure spinor formulations of both M2-brane models and supersymmetric Yang--Mills theory. The latter leads to a new cubic pure spinor action for ten-dimensional supergravity. We also give a homotopy-algebraic perspective on colour--flavour-stripping, obtain a new restricted tensor product over a wide class of bialgebras, and we show that any field theory (even one without colour--kinematics duality) comes with a kinematic $L_\infty$-algebra.}
    
    \acknowledgements{We are particularly grateful to Martin Cederwall and Pietro Antonio Grassi for detailed discussion on pure spinor formulations of gauge theories and supergravity. We also benefited from discussion with Maor Ben-Shahar, Silvia Nagy, Anton Zeitlin, and an anonymous user on MathOverflow.}

    \declarations{
        \textbf{Funding.}
        H.K. and C.S.~were supported by the Leverhulme Research Project Grant RPG-2018-329. B.J.~was supported by the GA\v{C}R Grant EXPRO 19-28628X.\\[5pt]
        \textbf{Conflict of interest.}
        The authors have no relevant financial or non-financial interests to disclose.\\[5pt]
        \textbf{Data statement.}
        No additional research data beyond the data presented and cited in this work are needed to validate the research findings in this work.\\[5pt]
        \textbf{Licence statement.}
        For the purpose of open access, the authors have applied a Creative Commons Attribution (CC-BY) license to any author-accepted manuscript version arising.
    }
    
    \begin{body}
        \section{Introduction and results}\label{sec:intro}
        
        \paragraph{Background.}
        The space of observables of a classical field theory is a rather complicated object. In order to obtain it, one needs to quotient the classical field space by gauge transformations and then divide the ring of functions on this quotient space by the ideal generated by the equations of motion. At the classical level, the Batalin--Vilkovisky (BV) formalism~\cite{Batalin:1977pb,Batalin:1981jr,Batalin:1984jr,Batalin:1984ss,Batalin:1985qj,Schwarz:1992nx} turns this space into a differential complex, called the BV complex, in which the observables are encoded in the cohomology of the classical BV differential. We work at the purely classical level throughout.
        
        The BV complex forms, in fact, a differential graded commutative algebra, which is the Chevalley--Eilenberg algebra, or the dual description, of an $L_\infty$-algebra, see e.g.~\cite{Jurco:2018sby} for a detailed review as well as~\cite{Hohm:2017pnh} for the discussion of equations of motion. Such an $L_\infty$-algebra is a generalisation of a differential graded Lie algebra, in which the Jacobi identity holds only up to homotopy. Moreover, the anti-bracket on the BV complex encodes a metric on the $L_\infty$-algebra. Altogether, this leads to the \uline{homotopy algebraic perspective} on perturbative quantum field theory, which implies a dictionary between physical concepts and algorithms and mathematical notions and constructions; we list some elements of this dictionary in \cref{tab:1}.
        
        \begin{table}[ht]
            \vspace*{10pt}
            \begin{center}
                \resizebox{\columnwidth}{!}{
                    \begin{tabular}{ll}
                        \toprule
                        Perturbative quantum field theory & Homotopy algebraic perspective
                        \\
                        \midrule
                        classical action $S$ & metric $L_\infty$-algebra $\frL_S$
                        \\
                        tree-level scattering amplitude for $S$ & minimal model for $\frL_S$
                        \\
                        choice of gauge fixing & embedding of the minimal model into $\frL_S$
                        \\
                        integrating out fields & homotopy transfer from $\frL_S$ to smaller $L_\infty$-algebra
                        \\
                        semi-classical equivalence $S\sim\tilde S$ & quasi-isomorphism $\frL_S\cong\frL_{\tilde S}$
                        \\
                        Feynman diagram expansion & homological perturbation lemma
                        \\
                        Berends--Giele recursion relation & geometric series via homological perturbation lemma
                        \\
                        colour-stripping of amplitudes & factorisation $\frL_S\cong\frg\otimes\frC$ with $\frC$ a $C_\infty$-algebra
                        \\
                        special properties of amplitudes & homotopy algebraic refinement of $L_\infty$-algebra $\frL_S$
                        \\
                        colour--kinematics duality & $\frL_S\cong\frg\otimes\frB$ with $\frB$ a homotopy $\BVbox$-algebra
                        \\
                        manifest colour--kinematics duality & $\frL_S\cong\frg\otimes\frB$ with $\frB$ a $\BVbox$-algebra
                        \\
                        \midrule
                        loop level considerations & extend the above to loop homotopy algebras
                        \\
                        \bottomrule
                    \end{tabular}
                }
            \end{center}
            \caption{Some entries in the dictionary describing the homotopy algebraic perspective on perturbative quantum field theory.}
            \label{tab:1}
        \end{table}
        
        Particularly noteworthy is the fact that the homotopy algebraic perspective on quantum field theory puts action principles and scattering amplitudes on equal footing: both are particular forms of $L_\infty$-algebras~\cite{Jurco:2018sby,Jurco:2019bvp,Arvanitakis:2019ald,Macrelli:2019afx,Jurco:2019yfd,Jurco:2020yyu}. Closely related to this perspective is also the work by Costello~\cite{Costello:2011aa} and Costello and Gwilliam~\cite{Costello:2016vjw,Costello:2021jvx}.
        
        In this paper, our goal is to explain the connection between colour--kinematics duality in much more detail and to add the following further line to \cref{tab:1}:

        \begin{table}[ht]
            \vspace*{5pt}
            \begin{center}
                \resizebox{\columnwidth}{!}{
                \begin{tabular}{ll}
                    \toprule
                    Perturbative quantum field theory & Homotopy algebraic perspective
                    \\
                    \midrule
                    double copy & kinematic Lie algebra in tensor product of metric $\BVbox$-algebras
                    \\ 
                    \bottomrule
                \end{tabular}
                }
            \end{center}
            \vspace*{-10pt}
        \end{table}
        
        Recall that colour--kinematics (CK) duality~\cite{Bern:2008qj,Bern:2010ue,Bern:2010yg} is a surprising and non-evident feature of perturbative quantum field theories, first observed in tree-level scattering amplitudes of Yang--Mills theories. Concretely, the scattering amplitudes of a CK-dual field theory can be decomposed into sums of cubic graphs with each diagram having a contribution $\frac{1}{p_\ell^2}$ from the propagator along each internal line $\ell$, a colour contribution, and a remaining kinematic contribution. CK duality is now the statement that the algebraic properties of the colour contributions induced by the anti-symmetry and Jacobi identity of the Lie bracket are precisely mirrored in the kinematic contributions. 
        
        It is natural to assume, and indeed is the case in many examples, that the  interaction  vertices are cubic and  decompose into products of the structure constants of a colour Lie algebra and the structure constants of a second Lie algebra, usually called the kinematic Lie algebra~\cite{Monteiro:2011pc,Bjerrum-Bohr:2012kaa,Monteiro:2013rya}. It is further natural to assume that the cubic graphs exhibiting CK duality are indeed the Feynman diagrams of the tree-level perturbative expansion of a field theory given by an action principle. In this case, the kinematic Lie algebra is manifested in the action itself, and a number of action-based approaches to CK duality and the double copy have been presented in the literature~\cite{Bern:2010yg,Hohm:2011dz,Tolotti:2013caa,Cardoso:2016ngt,Luna:2016hge, Borsten:2020xbt,Borsten:2020zgj,Borsten:2021hua,Borsten:2021gyl,Ben-Shahar:2021doh,Ben-Shahar:2021zww, Escudero:2022zdz,Borsten:2022vtg,Ben-Shahar:2022ixa, Lipstein:2023pih}.
        
        Interestingly, the homotopy algebraic perspective has an elegant description of this situation. Since there are only cubic vertices, the $L_\infty$-algebra $\frL$ encoding the action is simply a differential graded Lie algebra. The fact that we have a kinematic Lie algebra amounts to a factorisation $\frL\cong\frg\otimes\frB$, where $\frB$ is a differential graded commutative algebra refined to a $\BVbox$-algebra\footnote{in most cases; in the body of the paper, we will explain that a kinematic Lie algebra merely implies a pseudo-$\BVbox$-algebra structure}. This fact was first noted by Reiterer~\cite{Reiterer:2019dys} in the context of Yang--Mills theory in a first order formulation. In this picture, the kinematic Lie algebra appears in a degree-shifted form as the Gerstenhaber bracket that each $\BVbox$-algebra naturally possesses. This homotopy algebraic perspective on CK duality allowed us to produce a number of new and interesting results with comparatively little effort, cf.~\cite{Borsten:2022vtg,Borsten:2023reb}.
        
        CK duality has many implications and applications; see~\cite{Carrasco:2015iwa,Borsten:2020bgv, Bern:2019prr,Adamo:2022dcm,Bern:2022wqg} for reviews. For this paper, it is important to recall that CK duality is the key ingredient to the famous double copy prescription~\cite{Bern:2008qj,Bern:2010ue,Bern:2010yg} summarised by the slogan that `gravity is the square of Yang--Mills theory'. More precisely, the kinematic contribution to the CK-dual parametrisation of the Yang--Mills scattering amplitudes can be used to replace the colour contribution, leading to the scattering amplitudes of $\caN=0$ supergravity. The latter theory is a string-theoretically natural extension of Einstein--Hilbert gravity by a scalar dilaton field and a Kalb--Ramond 2-form field. 
        
        To arrive at a homotopy algebraic perspective on the double copy, it is natural to start from the $\BVbox$-algebras $\frB$ encoding the kinematic Lie algebra of Yang--Mills theory and to consider the tensor product with itself, $\hat \frB\coloneqq\frB\otimes\frB$. Recall that the tensor product of differential graded commutative algebras is again a differential graded commutative algebra, and this tensor product extends to $\BVbox$-algebras. 
        
        The field content of Yang--Mills theory is contained in $\frB_1$, the linear subspace of $\frB$ containing the homogeneous elements of degree~$1$. Correspondingly, the double-copied field content sits in $\frB_1\otimes\frB_1\subseteq\frB_2$. We expect the double copy to be described by a differential graded Lie algebra with the double-copied fields in degree~$1$, so it is evident that we will have to degree-shift $\frB$. There is now an evident candidate for this Lie algebra, namely the grade-shifted kinematic Lie algebra contained in the $\BVbox$-algebra $\hat\frB$ in the form of a Gerstenhaber bracket. 
        
        This suggestive answer has to be corrected in two ways. First of all, the domain of all fields in $\hat\frB$ is formed by two copies of the original space-time, somewhat akin to what happens in double field theory. This can be taken into account by introducing a cocommutative Hopf algebra $\frH$ whose elements correspond to the momenta labels of the field theory and act on $\frB$ and, thus, naturally on $\hat\frB$. We can then restrict to the invariants under this action, leading to fields taking values on the original space-time.\footnote{Another possibility is to take a double field theory-like approach and to impose a section condition, as done in~\cite{Bonezzi:2022bse}. A third possibility, suggested by~\cite{Anastasiou:2014qba, Anastasiou:2018rdx, Monteiro:2021ztt, Luna:2022dxo}, is to replace the pointwise  product with a convolution, as described in \cref{app:compactification-setting}.} 
        
        Secondly, the BV field space turns out to be twice the expected size of the usual BV field space for the double-copied field content. This can be corrected by restricting to the kernel of a naturally defined operator on $\hat \frB$. This kernel is closely related to level-matching in string theory and was also used for the double copy in~\cite{Bonezzi:2022bse,Bonezzi:2023pox,Bonezzi:2023ciu}. The result is indeed the differential graded Lie algebra of the double-copied field theory. 
        
        To demonstrate our mathematical constructions in detail, we consider a number of explicit examples in \cref{sec:examples}. In particular, we discuss our formalism for both CK duality and the double copy for the biadjoint scalar field theory (as well as the instructive extension to a biadjoint-bifundamental scalar field theory) and pure Chern--Simons theory. In the latter case, the double copy produces the complete BV triangle for an interesting biform field theory, whose physical part was previously derived in~\cite{Ben-Shahar:2021zww}. We also sketch our description of CK duality of~\cite{Borsten:2022vtg} and explain the relation to the recent work of~\cite{Bonezzi:2023pox}. Our most important examples are the pure spinor descriptions of Yang--Mills theory and M2-brane models. We review our description of full tree-level CK duality from~\cite{Borsten:2023reb}, but then also develop the corresponding picture for the double copy. In the case of Yang--Mills theory, we obtain the first cubic pure spinor action for ten-dimensional supergravity, which may also shed some light on questions in previous pure spinor actions for supergravity. In the case of M2-brane models, we obtain the again a cubic biform action which is an extension of the one obtained for Chern--Simons theory. This action is a candidate for either a supergravity or a Born--Infeld like action. We also consider the interesting example of a sesquiadjoint scalar field theory, a deformation of a biadjoint scalar field theory in which one of the two Lie algebras is replaced by a more general algebraic structure. In this case, the kinematic Lie algebra is lifted to a kinematic $L_\infty$-algebra, an object that any classical field theory possesses. 
        
        \paragraph{Results.}
        Altogether, our results can be summarised as follows. We show that any field theory that exhibits a kinematic Lie algebra has an underlying pseudo-$\BVbox$-algebra, a mild generalisation of a $\BVbox$-algebra. In these pseudo-$\BVbox$-algebras, the kinematic Lie algebra appears in a grade-shifted form, and the Lie bracket is given by a derived bracket\footnote{Such constructions are common in homotopical algebra.}. If $\BBox=\wave$, the Minkowski d'Alembertian, we have the usual form of CK duality.
        We also show that this kinematic Lie algebra is a special case of a more general kinematic $L_\infty$-algebra that any classical field theory possesses, but does not, however, imply CK duality (see \cref{ssec:kinematic_L_infty_algebras}). We then give a construction of the action of a syngamy field theory of two field theories with metric $\BVbox$-algebras. The familiar double copy is a special case of this construction, and using pure spinors, we find a new cubic action for ten-dimensional supergravity. Finally, this contribution provides a mathematically complete foundation for elements of our previous work, where the detailed mathematical tools were only sketched and developed as far as absolutely required.
        
        Byproducts of our constructions include the homotopy algebraic perspective on colour--flavour-stripping, see \cref{ssec:colour-flavour-stripping}, as well as a restricted tensor product of modules over a wide class of bialgebras, see \cref{app:restricted_tensor_product}, which appears to be a new mathematical construction.
        
        \paragraph{Literature overview.} 
        There have been a number of important developments in recent years closely related to this work, some in quick succession and happening in parallel, so it may be useful to give a brief contextual overview of the literature that uses an action-based approach to CK duality and the double copy, particularly from the homotopy algebraic perspective. 
        
        The idea that CK duality and the double copy can be approached from the perspective of the action is rather old and dates back to~\cite{Bern:2010yg}; see~\cite{Tolotti:2013caa,Cardoso:2016ngt,Luna:2016hge,Borsten:2020xbt,Borsten:2020zgj,Borsten:2021hua,Borsten:2021gyl,Ben-Shahar:2021doh,Ben-Shahar:2021zww,Escudero:2022zdz,Borsten:2022vtg,Ben-Shahar:2022ixa,Lipstein:2023pih} for work along the same lines. In the context of the double copy, homotopy algebras were first used\footnote{There is earlier work by Zeitlin~\cite{Zeitlin:2014xma}, in which a particular set of $\caN=0$ supergravity equations are reproduced within a tensor product of the homotopy commutative algebras underlying Yang--Mills theories, at least for Hermitian manifolds and at least to first order in homotopification. This paper does not link this observation to the double copy or the KLT relations.} in~\cite{Borsten:2020zgj,Borsten:2021hua}, where the double copy construction was given by a twisted tensor product;  recent applications of this technology include  homotopy double copies for Navier--Stokes equations~\cite{Escudero:2022zdz} and  non-commutative gauge theories~\cite{Szabo:2023cmv}. In this work, and in particular in~\cite{Borsten:2021gyl}, we demonstrated that CK duality could be realised at the level of the complete off-shell BV action up to counterterms that may be required to ensure manifest unitarity. In particular, we provided an algorithm to construct the CK duality manifesting BV action to any order in perturbation theory. This picture involved adding a tower of higher-order interaction terms to the BV action while preserving the S-matrix, building on the results of~\cite{Bern:2010yg,Tolotti:2013caa} by including ghost, longitudinal and off-shell states. The latter may induce counterterms required for unitarity that could break CK duality at the loop-level post-regularisation and prior to renormalisation  unitarity may not be manifest, as explained in [30].
        
        In~\cite{Reiterer:2019dys}, Reiterer made a seminal contribution to our understanding of CK duality.  In particular, it was shown that Zeitlin's differential graded commutative algebra of the (colour-stripped) first-order formulation of pure four-dimensional Yang--Mills theory~\cite{Zeitlin:2008cc} carries a homotopy $\BVbox_\infty$-algebra structure (also defined in~\cite{Reiterer:2019dys}). The central and immediate corollary is that the corresponding Feynman diagram expansion of the S-matrix satisfies CK duality up to homotopies given by the $\BVbox_\infty$-algebra. As for all homotopy algebras, there is a corresponding strict form of the $\BVbox_\infty$-algebra. Indeed, Reiterer provided a strictification (or rectification) relating the $\BVbox_\infty$-algebra to a $\BVbox$-algebra, making CK duality of the tree-level S-matrix exact and manifest. An interesting precursor to~\cite{Reiterer:2019dys} is found in the work of Zeitlin. In~\cite{Zeitlin:2008cc}, he speculates that there is a homotopy Gerstenhaber algebra in Yang--Mills theory, anticipating parts of a $\BVbox_\infty$-algebra structure. He also linked the homotopy commutative algebra arising in Yang--Mills theory in a particular limit to a homotopy commutative algebra arising for the Courant algebroid~\cite{Zeitlin:2009hc,Zeitlin:2009tj}, for which there exists a sketch of an argument that this algebra extends to a $\text{BV}_\infty$-algebra\footnote{We thank Anton Zeitlin for pointing this out.}.
            
        In~\cite{Borsten:2021gyl,Kim:2021amp, Borsten:2022srni,Borsten:2022young,Macrelli:2022bay,Macrelli2022sm,Macrelli:2022ich} it was explained that the  higher-order interaction terms, introduced in~\cite{Borsten:2020zgj, Borsten:2021hua, Borsten:2021gyl} to render the  BV action CK-dual, correspond (after colour-stripping) precisely to the higher products of a $\BVbox_\infty$-algebra\footnote{The non-trivial higher-products of the $\BVbox_\infty$-algebra roughly split into three classes corresponding to interactions generated by Tolotti--Weinzierl-type terms, gauge-fixing and field  redefinitions. With hind-sight, the algorithms  of~\cite{Tolotti:2013caa, Borsten:2020zgj, Borsten:2021hua, Borsten:2021gyl} can be understood as uncovering fragments of a $\BVbox_\infty$-algebra.}. By introducing auxiliary fields, the tower of higher-order interactions can be made cubic and we arrive at a strict $\BVbox$-algebra with manifest CK duality~\cite{Borsten:2021gyl}. The conclusion (roughly) is that any theory with a CK duality manifesting BV action  has an $L_\infty$-algebra carrying a $\BVbox$-algebra structure~\cite{Borsten:2021gyl, Borsten:2022vtg}. This gives rise to the penultimate entry in~\cref{tab:1}. Implicit in this statement, is a cyclic structure for the $\BVbox_\infty$-algebra, inherited from the anti-bracket, answering one of the open problems identified in~\cite{Reiterer:2019dys}. We make this precise in the present contribution.
        
        In light of these developments, CK duality is a (possibly anomalous\footnote{In the sense described above; CK duality violating counter-terms may be required to ensure manifest unitarity~\cite{Borsten:2021gyl}.}) symmetry of the action itself; as such, it is natural to expect that there is an underlying organising principle manifesting this symmetry. In~\cite{Ben-Shahar:2021doh}, the authors realised that pure spinor space can provide such a principle, and using it, they could establish CK duality for the tree-level currents of ten dimensional supersymmetric Yang--Mills theory. In~\cite{Borsten:2022vtg}, we then identified twistor spaces as a second, and closely related, organising principle. This should come as no surprise; besides the even simpler biadjoint scalar field theory~\cite{Hodges:2011wm,Vaman:2010ez,Cachazo:2013iea,Monteiro:2013rya,Cachazo:2014xea,Monteiro:2014cda,Chiodaroli:2014xia,Naculich:2014naa,Luna:2015paa,Naculich:2015zha,Chiodaroli:2015rdg,Luna:2016due,White:2016jzc,Cheung:2016prv,Chiodaroli:2017ngp,Brown:2018wss}, Chern--Simons theory is a prime example of a CK-dual field theory, cf.~\cite{Ben-Shahar:2021zww}, and both pure spinors and twistor space allow for a reformulation of Yang--Mills theories as Chern--Simons-type theories.
        
        Using twistor space, it is possible to concretely identify the kinematic Lie algebras of self-dual and full supersymmetric Yang--Mills theories. In the case of self-dual Yang--Mills theory, the resulting kinematic Lie algebra comes in a form that implies conventional CK duality even at the loop level. Having become aware of the work~\cite{Ben-Shahar:2021doh}, we also studied pure spinor space actions of ten dimensional supersymmetric Yang--Mills theory in~\cite{Borsten:2023reb}, where by using a different choice of gauge, we could lift the result of~\cite{Ben-Shahar:2021doh} to the tree-level amplitudes. This implied a new proof of tree-level CK duality for Yang--Mills theories in arbitrary dimensions $d\leq 10$ with an arbitrary amount of supersymmetry, which is simpler than existing ones in that it uses directly the action and does not rely on any concrete computations. In the same paper, we also extended Reiterer's perspective on CK duality to gauge--matter theories, which come with additional $\BVbox$-modules from the homotopy algebraic perspective. This, together with the pure spinor actions for M2-brane models of~\cite{Cederwall:2008vd,Cederwall:2008xu}, allowed us to give the first proof of full, tree-level CK duality for M2-brane models.
        
        Given Reiterer's interpretation of CK duality as a $\BVbox$-algebra, it is natural to look for an interpretation of the double copy in the tensor product of two $\BVbox$-algebras, as originally suggested in~\cite{Reiterer:2019dys}. We presented initial ideas for such a construction in~\cite{Borsten:2021gyl,Macrelli:2022bay,Macrelli:2022ich}. Independently, a double-field-theory-inspired version of this interpretation was then given in~\cite{Bonezzi:2022bse}, drawing on ideas in earlier work~\cite{Hohm:2011dz, Diaz-Jaramillo:2021wtl,Bonezzi:2022yuh} relating the double copy to double field theory; see also~\cite{Bonezzi:2023ced,Bonezzi:2023pox,Bonezzi:2023ciu} for recent work building on this, for example constructing weakly constrained double field theory to quartic order and elucidating the case of self-dual gravity, as well as~\cite{Lee:2018gxc,Cho:2019ype,Kim:2019jwm,Angus:2021zhy,Cho:2021nim} for further double-field-theory-inspired work on the double copy. Our present contribution mostly agrees with the constructions of~\cite{Bonezzi:2022bse}\footnote{which, in turn, have some similarity to those of~\cite{Zeitlin:2014xma}}, except that we use a Hopf algebra\footnote{This is in line with Reiterer's original construction, and very helpful for the homotopification of this picture to be presented in~\cite{Borsten:2022aa}.} to control momentum dependence, while~\cite{Bonezzi:2022bse} employs a double-field-theory-like section condition. However, we would like to stress that our constructions go beyond those of~\cite{Bonezzi:2022bse} in a number of ways. First of all, all our construction applies to metric\footnote{Homotopy algebraists may prefer the term `cyclic'.} $\BVbox$-algebras, and we give an explicit prescription for double copying the field-space metric. This is important for considering amplitudes and action principles; in particular, a $\BVbox$-algebra implies CK-duality on currents, but not on amplitudes, as explained in~\cref{ssec:kinematic_Lie_algebras}. Secondly, we discuss gauge--matter theories by allowing for modules over $\BVbox$-algebras. Thirdly, since we focus on $\BVbox$-algebras, and all our constructions are exact; in~\cite{Bonezzi:2022bse}, the authors use $\BVbox_\infty$-algebras, for which the precise definition of tensor product is unclear, forcing one to work order by order in the double copy.\footnote{We note that in the conclusions of~\cite{Bonezzi:2022bse}, the authors identify a complete form of the $\BVbox_\infty$-algebra of Yang--Mills theory as the most important outstanding problem. Our twistor space descriptions of self-dual and full Yang--Mills theories~\cite{Borsten:2022vtg} provide such a complete form. To turn it into a plain space-time expression, all one has to do is  perform a mode expansion and integrate over the auxiliary spectral parameters in twistor space. A similar construction exists for the pure spinor actions. From our perspective, an order-by-order computation is possible (as explained already in~\cite{Borsten:2020zgj,Borsten:2021hua}), but we believe that just as for supersymmetry, using an auxiliary space providing an organising principle is much more useful.}
        
        Most recently, $\BVbox$-algebras were also used in~\cite{Bonezzi:2023pox} to study self-dual Yang--Mills theory, but contrary to~\cite{Borsten:2022vtg}, where the exact $\BVbox$-algebra was given using an auxiliary space, a gauge-invariant formulation of self-dual Yang--Mills theory on space-time was studied directly, leading to a $\BVbox_\infty$-algebra deduced up to cubic order; we comment in detail on the relation between this work and our perspective in~\cref{sec:twistorSpaceSDYM}.
        
        \section{Basics of colour--kinematics duality}\label{ssec:CK duality}
        
        \subsection{Colour--kinematics duality and the double copy}
        
        We begin with a concise review of colour--kinematics (CK) duality. For general reviews on CK duality and the double copy, see~\cite{Carrasco:2015iwa,Borsten:2020bgv,Bern:2019prr,Adamo:2022dcm,Bern:2022wqg}.
        
        \paragraph{Colour--kinematics duality.}
        A gauge field theory is said to possess colour--kinematics duality if its scattering amplitude integrands can be parametrised in terms of cubic graphs (i.e.~diagrams with vertices that all have degree~$3$) such that at vertices and connected pairs of vertices, the gauge Lie algebra contribution to these diagrams has the same algebraic properties as the kinematic contribution. More specifically, the $n$-point, $L$-loop scattering amplitude integrands $\scA_{n,L}$ can be parametrised as 
        \begin{equation}\label{eq:CK_amplitudes_parameterization}
            \scA_{n,L}\ \sim\ \sum_{\gamma\in\Gamma_{n,L}}\frac{\sfc_\gamma\sfn_\gamma}{|\sfAut(\gamma)|d_\gamma}~,
        \end{equation}
        where $\Gamma_{n,L}$ is the set of $n$-point, $L$-loop cubic diagrams; $\sfc_\gamma$ is the \uline{colour numerator}, that is, the contribution to the diagram $\gamma$ due to the metric and the structure constants of the gauge Lie algebra; $d_\gamma$ is the product of the denominators of the propagators (without colour component) for $\gamma$, usually $\frac{1}{p_\ell^2}$ for each propagator line $\ell\in \gamma$; $|\sfAut(\gamma)|$ is the symmetry factor of the diagram $\gamma$, i.e.~the order of its automorphism group; and $\sfn_\gamma$ is the \uline{kinematic numerator} containing the remaining contributions of $\gamma$ to $\scA_{n,L}$. The anti-symmetry of the Lie algebra structure constants and the Jacobi identity induce certain sums of colour numerators to vanish, i.e.
        \begin{equation}\label{eq:CKcolor}
            \sfc_{\gamma_{a1}}+\sfc_{\gamma_{a2}}\ =\ 0
            \eand
            \sfc_{\gamma_{J1}}+\sfc_{\gamma_{J2}}+\sfc_{\gamma_{J3}}\ =\ 0
        \end{equation}
        for certain pairs $(\gamma_{a1},\gamma_{a2})$ and triples $(\gamma_{J1},\gamma_{J2},\gamma_{J3})$. A theory is said to be colour--kinematics (CK) dual if the same relations hold for the corresponding kinematic numerators:
        \begin{equation}\label{eq:CKkinematic}
            \sfn_{\gamma_{a1}}+\sfn_{\gamma_{a2}}\ =\ 0
            \eand
            \sfn_{\gamma_{J1}}+\sfn_{\gamma_{J2}}+\sfn_{\gamma_{J3}}\ =\ 0~.
        \end{equation}
        
        Full CK duality has been established for very few field theories; in particular, it is found for the archetypal cases of biadjoint scalar field theory and Chern--Simons theory\footnote{As Chern--Simons theory is trivial on Minkowski space, one considers `scattering amplitudes' of harmonic differential forms.}~\cite{Ben-Shahar:2021zww}. For Yang--Mills theory and supersymmetric generalisations, CK duality has been established at the tree level using a variety of approaches~\cite{BjerrumBohr:2010hn,BjerrumBohr:2009rd,Stieberger:2009hq,Mafra:2011kj,Broedel:2013tta,Du:2016tbc,Mizera:2019blq,Reiterer:2019dys,Borsten:2023reb}. It is known, however, that loop--level CK duality for pure Yang--Mills theory is not possible if one assumes that the kinematic numerators could have been derived from the Feynman diagrams of a local action with manifest unitarity~\cite{Bern:2015ooa}. This conclusion is also confirmed by observations regarding possible CK-dual action principles in~\cite{Borsten:2021gyl,Borsten:2022vtg}. A lift up to anomalies, however, does exist~\cite{Borsten:2021gyl}.
        
        \paragraph{Colour--kinematics duality for currents.}
        Note that we can also study CK duality for currents as e.g.~the famous Berends--Giele gluon currents~\cite{Berends:1987me}. These are essentially amplitudes, but with one external leg kept off-shell and a propagator attached to this leg. They can be computed recursively, and sometimes possess a more evident form of CK duality, cf.~e.g.~\cite{Monteiro:2011pc,Ben-Shahar:2021doh}. Explicitly, we have a similar parametrisation to~\eqref{eq:CK_amplitudes_parameterization}, namely
        \begin{equation}\label{eq:CK_currents_parameterization}
            \scC_{n,L}\ \sim\ \sum_{\gamma\in\Gamma_{n,L}}\frac{\sfc_\gamma\sfn_\gamma}{|\sfAut(\gamma)|d_\gamma}
        \end{equation}
        such that~\eqref{eq:CKcolor} implies~\eqref{eq:CKkinematic} in the evident fashion, but $d_\gamma$ here contains an additional factor arising from the propagator on the single external leg with propagator, and the $\sfn_\gamma$ now also may involve off-shell polarisations.        
        
        \paragraph{Double copy.}
        CK duality is the crucial ingredient in the double copy construction: the kinematic numerators $\sfn_\gamma$ of a CK-dual field theory can be doubled to construct consistent scattering amplitude integrands of a new field theory,
        \begin{equation}
            \tilde\scA_{n,L}\ \sim\ \sum_{\gamma\in\Gamma_{n,L}}\frac{\sfn_\gamma\sfn_\gamma}{|\sfAut(\gamma)|d_\gamma}~.
        \end{equation}
        It has been shown that starting from tree-level pure Yang--Mills scattering amplitudes, the double copy construction yields the tree-level scattering amplitudes of $\caN=0$ supergravity~\cite{Bern:2008qj,Bern:2010ue,Bern:2010yg}, and this generalises to supersymmetric gauge and gravity theories, see again~\cite{Bern:2008qj,Bern:2010ue,Bern:2010yg}.
        
        More generally, one can take the kinematic numerators $\sfn^{(1)}_\gamma$ and $\sfn^{(2)}_\gamma$ of two CK-dual field theories and form their \uline{syngamy}\footnote{We follow again our nomenclature of~\cite{Borsten:2021gyl}.} theory, i.e.  
        \begin{equation}
            \tilde\scA_{n,L}\ \sim\ \sum_{\gamma\in \Gamma_{n,L}}\frac{\sfn^{(1)}_\gamma\sfn^{(2)}_\gamma}{|\sfAut(\gamma)|d_\gamma}~.
        \end{equation}
        
        In this paper, we shall focus on the Lagrangian perspective on CK duality and the double copy~\cite{Bern:2010yg,Tolotti:2013caa,Borsten:2020zgj,Borsten:2021hua,Borsten:2021gyl,Borsten:2022vtg}. Our aim is then to explain the relevant mathematical structures underlying the double copy prescription from this perspective.
        
        \paragraph{Gauge--matter colour--kinematics duality.}
        The above form of colour--kinematics duality can be extended from  gauge theories to gauge--matter theories~\cite{Johansson:2014zca,Johansson:2015oia}. See~\cite{Chiodaroli:2015wal,Anastasiou:2016csv,Chiodaroli:2017ehv, Anastasiou:2017nsz,Chiodaroli:2018dbu,Ben-Shahar:2018uie, Johansson:2019dnu} for a variety of gauge--matter  colour--kinematics duality and double copy examples. By  gauge theory, we mean any theory where all the fields are valued in the adjoint representation of the gauge Lie algebra $\frg$, such as Yang--Mills and maximally supersymmetric Yang--Mills theories.\footnote{Thus, theories without gauge symmetry such as the biadjoint scalar or the non-linear sigma model on a principal homogeneous space are nevertheless `gauge theories' in our sense.} Gauge--matter theories, on the  other hand, include (possibly integer spin) `matter' fields carrying some other representation $R$ of $\frg$.  The colour--stripped amplitudes are constructed in the same manner as the case of purely adjoint fields, although the colour decomposition may be more involved~\cite{Johansson:2015oia}, essentially due to the particular representation theoretic properties of the matter. See \cref{ssec:colour-flavour-stripping} for the details relevant to our discussion. 
        
        CK duality proceeds much as before. The only structural difference from the case of gauge theories is that now~\eqref{eq:CKcolor} can hold either due to the Jacobi identity of the gauge Lie algebra, as before, due to the commutation relations in the (not necessarily irreducible) representation $R$~\cite{Johansson:2014zca,Johansson:2015oia}, or due to some combination of the two. Correspondingly, the sum over cubic Feynman diagrams~\eqref{eq:CK_amplitudes_parameterization} is enlarged to include all possible decorations of the edges by matter field representations $R$:
        \begin{equation}\label{eq:CK_amplitudes_parameterization_matter}
            \scA_{n,L}\ \sim\ \sum_{\gamma\in \Gamma^{R}_{n,L}}\frac{\sfc_\gamma\sfn_\gamma}{|\sfAut(\gamma)|d_\gamma}~.
        \end{equation}
        Here, $\Gamma^{R}_{n,L}$ denotes the set of $n$-point, $L$-loop cubic graphs with all consistent decorations of the edges by $R$, including the subset $\Gamma_{n,L}\subseteq\Gamma^{R}_{n,L}$ without decorations (the pure adjoint graphs). Note that $R$ may include several copies of the same irreducible representation of the gauge Lie algebra to incorporate flavours.      
        
        \paragraph{Double copy with gauge--matter theories.}
        The double copy is usually generalised to $\sfc^{(1)}_{\gamma^{(1)}} \sfn^{(1)}_{\gamma^{(1)}} \rightarrow \sfn^{(2)}_{\gamma^{(2)}} \sfn^{(1)}_{\gamma^{(1)}}$, where  $\gamma^{(1)}$ and ${\gamma^{(2)}}$ either both belong to $\Gamma_{n,L}$ or belong to $\Gamma^{R^{(1)}}_{n,L}\setminus \Gamma_{n,L}$ and $\Gamma^{R^{(2)}}_{n,L}\setminus\Gamma_{n,L}$, respectively\footnote{Note, this is  in the spirit of~\cite{Bern:2019prr} and more general than the working rule~4 adopted in~\cite{Anastasiou:2016csv}. It is consistent nonetheless, at least when there is an underlying action.}. This restriction reflects the fact that only field couplings corresponding to $R\times R\rightarrow\frg$ and, dually, $\frg\times R\rightarrow R$ do not require any properties of the representations beyond the universal Jacobi identities, commutation relations, and  existence of conjugates\footnote{We are implicitly assuming here that $R$ contains all required conjugate representations.}. While more elaborate coupling are in principle possible, we explicitly restrict to these cases, as described in \cref{ssec:colour-flavour-stripping}. This is mathematically natural, see \cref{sec:syngamiesMatterTheories}, and appears to be physically necessary. Allowing, say, $\gamma\in\Gamma_{n,L}$ and $\gamma'\in\Gamma^{R^{(2)}}_{n,L}\setminus\Gamma_{n,L}$ could be used to produce arbitrary numbers of gravitini, which would be inconsistent with the accompanying local supersymmetry~\cite{Anastasiou:2016csv}. 
        
        \subsection{Field theories and homotopy algebras}
        
        Our discussion will be based on the \uline{homotopy algebraic perspective} on classical field theories, cf.~e.g.~\cite{Jurco:2018sby,Borsten:2021hua} or~\cite{Hohm:2017pnh}. 
        
        \paragraph{Metric differential graded Lie algebras.}
        The classical Batalin--Vilkovisky (BV) action\footnote{Note that the BV algebras and $\BVbox$-algebras that form an essential ingredient in our picture are not obtained from a BV formulation of the theories we consider.} of a field theory with cubic vertices is dual to a metric \uline{differential graded (dg) Lie algebra} $(\frL,\mu_1,\mu_2)$ with the underlying graded vector space $\frL\cong\bigoplus_{i\in\IZ}\frL_i$ and cochain complex
        \begin{equation}
            \sfCh(\frL)\ \coloneqq\ \big(
            \begin{tikzcd}
                \cdots\arrow[r,"\mu_1"] & \frL_0\arrow[r,"\mu_1"] & \frL_1\arrow[r,"\mu_1"] & \frL_2\arrow[r,"\mu_1"] & \frL_3\arrow[r,"\mu_1"] & \cdots
            \end{tikzcd}
            \big)\,.
        \end{equation}
        Here, $\frL_0$ contains the ghosts, $\frL_1$ the fields, $\frL_2$ the anti-fields, and $\frL_3$ the anti-fields of the ghosts,\footnote{not to be confused with the anti-ghost fields} respectively. Hence, the degree~$|\phi|$ of a field $\phi\in\frL$ is given by 
        \begin{equation}
            |\phi|\ \coloneqq\ 1-|\phi|_\text{gh}~,
        \end{equation}
        where $|\phi|_\text{gh}$ is the ghost degree of $\phi$. Correspondingly, in a gauge-fixed BV formulation of an ordinary gauge theory, $\frL_1$ will also contain the Nakanishi--Lautrup field and the anti-field of the anti-ghost and $\frL_2$ will also contain the anti-field of the Nakanishi--Lautrup field and the anti-ghost. The differential $\mu_1$ encodes all linear features of the theory, such as kinematic terms, linearised gauge transformations, and their duals. Interactions, non-linear parts of gauge transformations, and their duals are encoded in a \uline{graded Lie bracket}
        \begin{equation}
            \mu_2\,:\,\frL\times\frL\ \rightarrow\ \frL~,
        \end{equation}
        which is of degree~$0$, bilinear, graded anti-symmetric, compatible with the differential, and satisfies the graded Jacobi identity. The \uline{metric} (or \uline{cyclic structure})
        \begin{equation}
            \inner{-}{-}\,:\,\frL\times\frL\ \rightarrow\ \IR
        \end{equation}
        is a non-degenerate, bilinear, and graded symmetric map of a fixed degree, which is compatible with the differential $\mu_1$ and the Lie bracket $\mu_2$ in the sense that
        \begin{equation}
            \begin{aligned}
                \inner{\mu_1(\phi_1)}{\phi_2}+(-1)^{|\phi_1|}\inner{\phi_1}{\mu_1(\phi_2)}\ &=\ 0~,
                \\
                \inner{\mu_2(\phi_1,\phi_2)}{\phi_3}+(-1)^{|\phi_1|\,|\phi_2|}\inner{\phi_2}{\mu_2(\phi_1,\phi_3)}\ &=\ 0
            \end{aligned}
        \end{equation}
        for all $\phi_{1,2,3}\in\frL$. If the metric is of degree~$-3$, we can use it to write down an action principle
        \begin{equation}
            S\ \coloneqq\ \tfrac12\inner{\phi}{\mu_1(\phi)}+\tfrac{1}{3!}\inner{\phi}{\mu_2(\phi,\phi)}
        \end{equation}
        for the fields $\phi\in\frL_1$. In this way, any action with exclusively cubic interaction vertices can be encoded in a metric dg Lie algebra.
        
        \paragraph{Homotopy transfer.}
        We can obtain an equivalent field theory by `integrating out' parts of the field content. This is done by an appropriate tree-level Feynman diagram expansion, and mathematically, this corresponds to a \uline{homotopy transfer} from the cochain complex $(\frL,\mu_1)$ to a quasi-isomorphic cochain complex $(\tilde\frL,\tilde\mu_1)$ consisting of the modes that have not been integrated out, cf.~\cite{Doubek:2017naz}.\footnote{The fact that homotopy transfer amounts to integrating out fields is a general folklore in BV quantisation; see also~\cite{Arvanitakis:2020rrk} and~\cite{Farahani:2023ptb} for recent applications.} In particular, we have the diagram
        \begin{subequations}
            \begin{equation}\label{eq:deformation_retract}
                \begin{tikzcd}
                    \ar[loop,out=160,in=200,distance=20,"\sfh" left] (\frL,\mu_1)\arrow[r,shift left]{}{\sfp} & (\tilde \frL,\tilde \mu_1) \arrow[l,shift left]{}{\sfe}~,
                \end{tikzcd}
            \end{equation}
            where $\sfp$ and $\sfe$ are cochain maps, denoting a projection and an embedding, such that
            \begin{equation}
                \sfp\circ\sfe\ =\ \sfid_{\tilde\frL}~,
            \end{equation}
            which implies that 
            \begin{equation}\label{eq:projector_on_onshell_modes}
                \Pi\ \coloneqq\ \sfe\circ\sfp
            \end{equation}
            is a projector. There is usually some ambiguity in choosing $\sfe$, which involves a choice of gauge. The \uline{contracting homotopy} $\sfh\colon\frL\rightarrow\frL$ is a map of degree~$-1$ satisfying 
            \begin{equation}\label{eq:contracting_homotopy}
                \id_\frL-\Pi\ =\ \mu_1\circ\sfh+\sfh\circ\mu_1
            \end{equation}        
            as well as the \uline{annihilation} or \uline{side conditions}
            \begin{equation}\label{eq:HT_side_conditions}
                \sfp\circ\sfh\ =\ 0~,~~~
                \sfh\circ\sfe\ =\ 0~,~~~
                \sfh\circ\sfh\ =\ 0~.
            \end{equation}
            Even if the side conditions do not hold, one can redefine $\sfh$ such that they do, cf.~\cite{Markl:0002130}.
        \end{subequations}        
        Note that equation~\eqref{eq:contracting_homotopy} implies that $\sfh$ is the inverse of $\mu_1$ on the modes that are being integrated out. 
        
        In other words, $\sfh$ can be regarded as a \uline{propagator}, and the homotopy transfer indeed reproduces the usual tree-level Feynman diagram expansion with propagator $\sfh$. The result of this homotopy transfer generically contains $n$-point vertices, which are encoded in algebraic operations with $n-1$ inputs and one output. Therefore, the result of the homotopy transfer is no longer a dg Lie algebra but a generalisation known as an \uline{$L_\infty$-algebra}. The notion of a dg Lie algebra is equivalent to that of a \uline{strict} $L_\infty$-algebra. Further details are again found, e.g., in~\cite{Jurco:2018sby,Borsten:2021hua}, but they will be irrelevant to our discussion.
        
        The smallest permissible cochain complex $(\tilde\frL,\tilde \mu_1)$ yields the \uline{minimal model} $(\frL^\circ,0)$, and it is given by the cohomology $\frL^\circ\coloneqq H^\bullet_{\mu_1}(\frL)$ of $(\frL,\mu_1)$. The minimal model is unique up to (strict) isomorphisms, and its $L_\infty$-algebra structure encodes the tree-level scattering amplitudes of the theory~\cite{Kajiura:2001ng,Kajiura:2003ax,Nutzi:2018vkl,Macrelli:2019afx,Arvanitakis:2019ald,Jurco:2019yfd}. Indeed, physical fields in the cohomology satisfy the free or linearised equations of motion, and linear gauge transformations have been quotiented out. We thus see that the physical fields in the cohomology correspond to the asymptotically free fields, labelling the open legs of scattering amplitudes. Altogether, there is now a dictionary between physical features and operations with scattering amplitudes and amputated correlators as well as (homotopy) algebraic operations, as indicated in \cref{tab:1}.
        
        \paragraph{Factorisation.}
        For example, we can factor out the colour or gauge Lie algebra $(\frg,[-,-]_\frg)$ by writing 
        \begin{equation}\label{eq:factorisationLGB}
            \frL\ \cong\ \frg\otimes\frB~,
        \end{equation}
        where $(\frB,\sfd,\sfm_2)$ is the \uline{differential graded (dg) commutative algebra} with
        \begin{equation}
            \begin{aligned}
                \mu_1(\tau_1\otimes\phi_1)\ &=\ \tau_1\otimes\sfd\phi_1~,
                \\
                \mu_2(\tau_1\otimes\phi_1,\tau_2\otimes\phi_2)\ &=\ [\tau_1,\tau_2]_\frg\otimes\sfm_2(\phi_1,\phi_2)
            \end{aligned}
        \end{equation}
        for all $\tau_{1,2}\in\frg$ and $\phi_{1,2}\in\frB$. This is the mathematical formulation of what physicists would call \uline{colour-stripping}, cf.~\cite{Zeitlin:2008cc,Borsten:2021hua}.
        
        In this paper, we will always regard a field theory as a metric dg Lie algebra, and we collect many examples in~\cref{sec:examples}.
        
        \subsection{Colour--flavour-stripping}\label{ssec:colour-flavour-stripping}
        
        We saw above that, mathematically, colour-stripping a cubic field theory amounts to a factorisation of the theory's dg Lie algebra into a colour Lie algebra and a dg commutative algebra. We are not aware of a discussion of the extension to colour--flavour-stripping in the literature, so we give a more detailed account here. This will become important when discussing CK duality of gauge--matter theories.
        
        \paragraph{Factorisation and Lie algebra representations.}
        Consider a gauge field theory with only cubic interaction vertices and gauge Lie algebra $\frg$. Then, the space of fields $\frF$ decomposes into irreducible representations of $\frg$ as
        \begin{equation}\label{eq:restricted_TP_module}
            \frF\ \cong\ (\frg\otimes\frC)\oplus(R^{(1)}\otimes V^{(1)})\oplus(R^{(2)}\otimes V^{(2)})\oplus\cdots~,
        \end{equation}
        in which $\frC$ is the graded vector space of fields transforming in the adjoint representation (such as the gauge potential or other components of the gauge supermultiplet in supersymmetric gauge theories), and $V^{(i)}$ for $i=1,2,\ldots$ is the graded vector space of fields transforming in the representation $R^{(i)}$. Since there are no invariant pairings between distinct irreducible representations, there are no kinetic terms that mix fields of different representations. Thus, $\frC$ and $V^{(i)}$ are dg vector spaces (i.e.~cochain complexes), each endowed with invariant metrics.
        
        To simplify the discussion, we combine $R\coloneqq\bigoplus_{i\in\IN}R^{(i)}$ and $V\coloneqq\bigoplus_{i\in\IN}V^{(i)}$, such that we can write
        \begin{equation}
            \frF\ \subseteq\ \frL\ \coloneqq\ (\frg\otimes\frC)\oplus(R\otimes V)
        \end{equation}
        for some cochain complexes $\frC$ and $V$ endowed with invariant metrics. The right-hand side is generically larger than~\eqref{eq:restricted_TP_module} since we also get summands $R^{(i)}\otimes V^{(j)}$ for $i\neq j$. We can, however, restrict to the subspace~\eqref{eq:restricted_TP_module} if necessary or desired.\footnote{This is a technical simplification. One can either regard the extra fields in $\frL\setminus\frF$ as free fields that decouple from the rest of the theory, or one can choose to keep track of different kinds of matter, which would technically amount to working with operads (i.e.~convenient tools for encoding algebras, cf.~\cite{0821843621,Loday:2012aa}) with more than two sorts.} The potential cubic interaction vertices encoded in the product $\mu_2$ can then be of a number of types,
        \begin{subequations}
            \begin{eqnarray}
                \mu_2\!\!&:&\!\!(\frg\otimes\frC)\times(\frg\otimes\frC)\ \rightarrow\ (\frg\otimes\frC)~,\label{eq:cubicInteractionI}
                \\
                \mu_2\!\!&:&\!\!(\frg\otimes\frC)\times(R\otimes V)\ \rightarrow\ (R\otimes V)~,\label{eq:cubicInteractionII}
                \\
                \mu_2\!\!&:&\!\!(R\otimes V)\times(R\otimes V)\ \rightarrow\ (\frg\otimes \frC)~,\label{eq:cubicInteractionIII}
                \\
                \mu_2\!\!&:&\!\!(R\otimes V)\times(R\otimes V)\ \rightarrow\ (R\otimes V)~,\label{eq:cubicInteractionIV}
                \\
                \mu_2\!\!&:&\!\!(\frg\otimes\frC)\times(R\otimes V)\ \rightarrow\ (\frg\otimes \frC)~,\label{eq:cubicInteractionV}
                \\
                \mu_2\!\!&:&\!\!(\frg\otimes\frC)\times(\frg\otimes\frC)\ \rightarrow\ (R\otimes V)~.\label{eq:cubicInteractionVI}
            \end{eqnarray}
        \end{subequations}
        Whilst the last three types of products~\eqref{eq:cubicInteractionIV}--\eqref{eq:cubicInteractionVI} are possible, they require additional algebraic structures on $\frg$ and $R$ that go beyond an ordinary Lie algebra representation. The products~\eqref{eq:cubicInteractionIV} still appear in familiar field theories, but~\eqref{eq:cubicInteractionV} and~\eqref{eq:cubicInteractionVI} are very uncommon. We therefore restrict ourselves to the case in which only the first three types~\eqref{eq:cubicInteractionI}--\eqref{eq:cubicInteractionIII} of maps are non-trivial; this certainly covers all field theories in which we are interested.\footnote{It is also mathematically natural. For example, it is reminiscent of the Lie algebra decomposition for symmetric spaces.} We note that cyclicity of the metric on $\frL$ implies in particular
        \begin{equation}\label{eq:special_cyclicity}
            \inner{\chi_1}{\mu_2(\phi,\chi_2)}\ =\ (-1)^{|\phi|\,|\chi_1|+1}\inner{\phi}{\mu_2(\chi_1,\chi_2)}
        \end{equation}
        for all $\chi_{1,2}\in R\otimes V$ and $\phi\in\frg\otimes\frC$, so that the product~\eqref{eq:cubicInteractionIII} is fixed by the product~\eqref{eq:cubicInteractionII}. 
        
        The first two types of product are captured by the Lie bracket on $\frg$, the action of $\frg$ on $R$, a structure of a dg commutative algebra on $\frC$, and an action of $\frC$ on the dg vector space $V$. 
        
        Putting all relevant structures together, we have the following mathematical description of colour--flavour-stripping.
        
        \begin{definition}\label{def:colour-flavour-stripping}
            Given a metric\footnote{sometimes called \uline{quadratic} or \uline{cyclic} instead} Lie algebra $(\frg,[-,-]_\frg,\inner{-}{-}_\frg)$ with a metric representation $(R,\acton_R,\inner{-}{-}_\rmR)$ together with a metric dg commutative algebra $(\frC,\sfd_\frC,\sfm_2,\inner{-}{-}_\frC)$ and a metric $\frC$-module $(V,\sfd_V,\acton_V,\inner{-}{-}_V)$, we define the \uline{tensor product}
            \begin{subequations}\label{eq:colour-flavour-stripping}
                \begin{equation}
                    \frL\ \coloneqq\ (\frg\otimes\frC)\oplus(R\otimes V)
                \end{equation}
                endowed with maps
                \begin{equation}
                    \begin{gathered}
                        \mu_1(\tau_1\otimes\phi_1+r_1\otimes v_1)\ \coloneqq\ \tau_1\otimes\sfd_\frC\phi_1+r_1\otimes\sfd_Vv_1~,
                        \\
                        \kern-6cm\mu_2(\tau_1\otimes\phi_1+r_1\otimes v_1,\tau_2\otimes\phi_2+r_2\otimes v_2)
                        \\
                        \kern-4cm\coloneqq\ [\tau_1,\tau_2]_\frg\otimes\sfm_2(\phi_1,\phi_2)+\mu_2(r_1\otimes v_1,r_2\otimes v_2)
                        \\
                        \kern1cm+(\tau_1\acton_Rr_2)\otimes(\phi_1\acton_Vv_2)-(-1)^{|v_1|\,|\phi_2|}(\tau_2\acton_Rr_1)\otimes(\phi_2\acton_Vv_1)
                    \end{gathered}
                \end{equation}
                with $\mu_2(r_1\otimes v_1,r_2\otimes v_2)$ defined by~\eqref{eq:special_cyclicity} as well as
                \begin{equation}
                    \inner{\tau_1\otimes\phi_1+r_1\otimes v_1}{\tau_2\otimes\phi_2+r_2\otimes v_2}_\frL\ \coloneqq\ \inner{\tau_1}{\tau_2}_\frg\,\inner{\phi_1}{\phi_2}_\frC+\inner{r_1}{r_2}_R\,\inner{v_1}{v_2}_V
                \end{equation}
            \end{subequations}
            for all $\tau_{1,2}\in\frg$, $r_{1,2}\in R$, $\phi_{1,2}\in\frC$, and $v_{1,2}\in V$.
        \end{definition}
        
        \begin{proposition}\label{prop:colour-flavour-stripping}
            The tuple $(\frL,\mu_1,\mu_2,\inner{-}{-})$ defined in~\eqref{eq:colour-flavour-stripping} forms a metric dg Lie algebra.
        \end{proposition}
        
        \begin{proof}
            By direct computation, cf.~\cref{app:postponed_proofs}.
        \end{proof}
        
        Clearly, the tensor product~\eqref{eq:colour-flavour-stripping} can possess metric dg Lie subalgebras of the form~\eqref{eq:restricted_TP_module}. Contrary to the colour-stripping, colour--flavour-stripping hence requires additional information about the desired branching of $R\otimes V$ into the summands $R^{(i)}\otimes V^{(i)}$. 
        
        Altogether, colour--flavour-stripping is a decomposition of the form~\eqref{eq:colour-flavour-stripping} such that the original metric dg Lie algebra $\frF$ is a subalgebra of the full tensor product $\frL$.
        
        We specialise this factorisation further to CK-dual ones in~\cref{ssec:BVbox-modules}, and physical examples are found in \cref{ssec:biadj_bifund_theory,ssec:M2branes}.
        
        \subsection{Kinematic Lie algebras from actions}\label{ssec:kinematic_Lie_algebras}
        
        \paragraph{Motivation.}
        For the action perspective on CK duality and the double copy, we will always assume that the diagrams $\gamma\in \Gamma_{n,L}$ in the expansions~\eqref{eq:CK_amplitudes_parameterization} and~\eqref{eq:CK_currents_parameterization} are indeed the Feynman diagrams of scattering amplitudes, as obtained from the rules derived from an action principle in the usual way. In this case, CK duality implies the existence of a \uline{kinematic Lie algebra}, from which the kinematic numerators $\sfn_\gamma$ are constructed in full analogy with the construction of the colour numerators $\sfc_\gamma$ from the gauge or colour Lie algebra. Put differently, each cubic vertex of the Feynman diagram $\gamma\in\Gamma_{n,L}$ contributes a structure constant to both $\sfc_\gamma$ and $\sfn_\gamma$, and propagators joining vertices amount to index contractions. The kinematic Lie algebra is the vital ingredient in the action perspective on CK duality, and we are not aware of an example of a CK-dual field theory without a kinematic Lie algebra. Moreover, the concept of a kinematic Lie algebra generalises far beyond theories with conventional CK duality, as we shall see. We will therefore always consider CK-dual field theories as a subset of theories with kinematic Lie algebras.
        
        As a fairly general and simple example of such a situation, consider the action
        \begin{equation}\label{eq:ex_CKD_action}
            S\ \coloneqq\ \tfrac12\sfg_{\sfi\sfj}\bar\sfg_{\bar\sfa\bar\sfb}\Phi^{\sfi\bar\sfa}\wave\Phi^{\sfj\bar\sfb}+\tfrac1{3!}\sfg_{\sfi\sfj}\bar\sfg_{\bar\sfa\bar\sfb}\sff^\sfj_{\sfk\sfl}\bar\sff^{\bar \sfb}_{\bar\sfc\bar\sfd}\Phi^{\sfi\bar\sfa}\Phi^{\sfk\bar\sfc}\Phi^{\sfl\bar\sfd}~,
        \end{equation}
        cf.~\cite{Borsten:2020zgj,Borsten:2021hua,Borsten:2021gyl}. Here, $\wave$ is the d'Alembertian, the $\sff^{\sfj}_{\sfk\sfl}$ and $\bar\sff^{\bar\sfb}_{\bar\sfc\bar\sfd}$ are structure constants of the gauge and kinematic Lie algebras, and the $\sfg_{\sfi\sfj}$ and $\sfg_{\bar\sfa\bar\sfb}$ are invariant metrics on each of the two Lie algebras, which are required for writing down an action principle. Note that $\sfi,\sfj,\ldots$ are DeWitt indices combining momentum, species, polarisation, and spinor labels. Among the field theories featuring tree-level CK duality that can be brought into this form are the biadjoint scalar field theory, the non-linear sigma-model, Chern--Simons theory, and Yang--Mills theory. 
        
        \paragraph{Feynman diagram expansion.}
        We will always be concerned with kinematic Lie algebras relative to a Feynman diagram expansion, or, equivalently, relative to a propagator $\sfh$, i.e.~a contracting homotopy in a deformation retract~\eqref{eq:deformation_retract}. The kinematic Lie algebras usually discussed in the literature are obtained when $\sfh$ is the ordinary Feynman propagator, giving a contracting homotopy to the minimal model of the underlying $L_\infty$-algebra, because this Feynman diagram expansion yields the scattering amplitudes. In the case of Chern--Simons theory, the tree-level scattering amplitudes are trivial, and we consider generalised amplitudes of harmonic differential forms.
        
        In particular, we shall follow an idea of Reiterer~\cite{Reiterer:2019dys} which assumes that the contracting homotopy or propagator $\sfh$ can be written as 
        \begin{equation}
            \sfh\ =\ \sfid_\frg\otimes\,\BBox^{-1}\sfb
            \ewith
            [\BBox,\sfb]\ =\ 0
        \end{equation}
        under the factorisation~\eqref{eq:factorisationLGB} such that $\sfb$ is a differential of degree~$-1$, which maps e.g.~physical anti-fields to physical fields, $\BBox$ is a second-order differential operator of degree~$0$ (e.g.~the d'Alembertian) with $\BBox^{-1}$ its inverse defined to vanish on $\ker(\BBox)$, and $\BBox\Pi=0=[\mu_1,\BBox^{-1}]$ for the projector~\eqref{eq:projector_on_onshell_modes}. Then,~\eqref{eq:contracting_homotopy} can be rewritten as  
        \begin{equation}
            \BBox\ \coloneqq\ [\sfd,\sfb]\ =\ \sfd\circ\sfb+\sfb\circ\sfd~.
        \end{equation}
        
        \paragraph{Derived bracket.}
        The operator $\sfb$ now allows us to define the \uline{derived bracket}
        \begin{equation}\label{eq:kinematic_bracket}
            \{\phi_1,\phi_2\}\ \coloneqq\ \sfb(\sfm_2(\phi_1,\phi_2))-\sfm_2(\sfb \phi_1,\phi_2)-(-1)^{|\phi_1|}\sfm_2(\phi_1,\sfb \phi_2)
        \end{equation}
        for all $\phi_{1,2}\in\frB$, which measures the failure of $\sfb$ to be a derivation of the product $\sfm_2$. This derived bracket enters into the construction of the kinematic numerators, analogously to the Lie algebra brackets entering into the colour numerators; and, in particular, it yields the Lie bracket of the kinematic Lie algebra.
        
        Returning to the action~\eqref{eq:ex_CKD_action}, the structure constants $\bar\sff^{\bar\sfb}_{\bar\sfc\bar\sfd}$ are those of the Lie algebra defined by the Lie bracket~\eqref{eq:kinematic_bracket}. This kinematic Lie algebra arises when integrating out modes in the Feynman diagram expansion with propagator $\sfid_\frg\otimes\,\BBox^{-1}\sfb$ and cubic vertices encoded in $\mu_2(-,-)=[-,-]_\frg\otimes\sfm_2(-,-)$. 
        
        \paragraph{Kinematic Lie algebra for currents.}
        Concretely, let us look at an example of a field theory current, i.e.~a Feynman diagram with $n$ incoming fields and one outgoing, propagating field $\phi_0$. This clearly demonstrates how the operator $\sfb$ gets assigned to vertices:
        
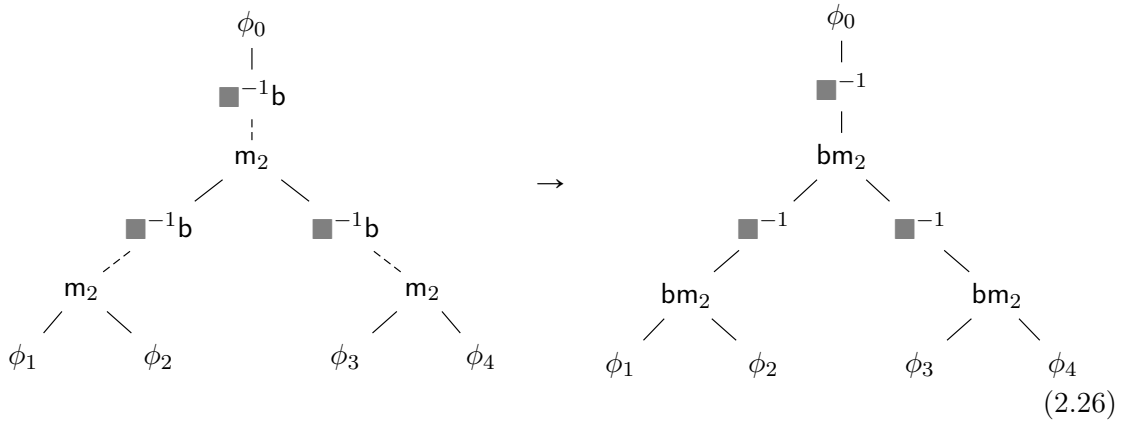
\begin{equation}\label{eq:CK_diagram}
            \begin{tikzcd}[column sep=0.0cm, row sep=0.3cm]
                & & & \phi_0 & & &
                \\
                & & & \BBox^{-1}\sfb \arrow[u,dash]& & &
                \\
                & & & \sfm_2 \arrow[u,dashed,dash] & & &
                \\
                & & \BBox^{-1}\sfb \arrow[ur,dash]& & \BBox^{-1}\sfb \arrow[ul,dash] & &
                \\
                & \sfm_2 \arrow[ur,dashed,dash] & & & & \sfm_2 \arrow[ul,dashed,dash] & 
                \\
                \phi_1 \arrow[ur,dash] & & \phi_2 \arrow[ul,dash] & & \phi_3 \arrow[ur,dash] & & \phi_4 \arrow[ul,dash]
            \end{tikzcd}
            ~~\rightarrow~~
            \begin{tikzcd}[column sep=0.0cm, row sep=0.3cm]
                & & & \phi_0 & & &
                \\
                & & & \BBox^{-1} \arrow[u,dash]& & &
                \\
                & & & \sfb\sfm_2 \arrow[u,dash] & & &
                \\
                & & \BBox^{-1} \arrow[ur,dash]& & \BBox^{-1} \arrow[ul,dash] & &
                \\
                & \sfb\sfm_2 \arrow[ur,dash] & & & & \sfb\sfm_2 \arrow[ul,dash] & 
                \\
                \phi_1 \arrow[ur,dash] & & \phi_2 \arrow[ul,dash] & & \phi_3 \arrow[ur,dash] & & \phi_4 \arrow[ul,dash]
            \end{tikzcd}
        \end{equation}
        Here, a solid line denotes a field and a dashed line denotes an anti-field. The operator $\sfb$ is taken along its unique anti-field line to a vertex and combined with $\sfm_2$ to the kinematic Lie bracket, which maps pairs of fields to fields. Note that $\sfb\sfm_2$ is indeed the kinematic Lie algebra on fields because, as we shall see, these are in the kernel of $\sfb$, at least after gauge fixing.
        
        This prescription clearly extends to currents involving anti-fields, where the outgoing leg can be a field. We thus see that after the re-assignment of the operator $\sfb$, the vertices are turned into the derived bracket~\eqref{eq:kinematic_bracket}, which is therefore the kinematic Lie algebra.
        
        \paragraph{Kinematic Lie algebra for scattering amplitudes.} 
        In the case of scattering amplitudes, the discussion is a bit more subtle. Amplitudes are obtained from the currents by removing the propagator on the outgoing leg of a current and pairing the anti-field coming out of the diagram with the remaining field using the cyclic structure. For example, the amplitude $\scA(\phi_0,\phi_1,\phi_2,\phi_3)$ will receive a contribution from
        \begin{equation}
            \begin{tikzcd}[column sep=0.2cm, row sep=0.3cm]
                & & & \phi_0 & & &
                \\
                & & & \langle-,-\rangle \arrow[u,dash]& & &
                \\
                & & & \sfm_2 \arrow[u,dashed,dash] & & &
                \\
                & & \BBox^{-1} \arrow[ur,dash]& & \BBox^{-1} \arrow[ul,dash] & &
                \\
                & \sfb\sfm_2 \arrow[ur,dashed,dash] & & & & \sfb\sfm_2 \arrow[ul,dashed,dash] & 
                \\
                \phi_1 \arrow[ur,dash] & & \phi_2 \arrow[ul,dash] & & \phi_3 \arrow[ur,dash] & & \phi_4 \arrow[ul,dash]
            \end{tikzcd}
        \end{equation}
        It is then clear that CK duality will hold for any triple of subdiagrams not involving $\phi_0$. For all physically interesting theories, however, the relevant external fields will be $\sfb$-exact, i.e.~in particular $\phi_0=\sfb\psi$. In this case, we can compute the sum of the general $s$-, $t$- and $u$-channels (i.e.~the terms $\sfn_{\gamma_{J1}}$, $\sfn_{\gamma_{J2}}$, and $\sfn_{\gamma_{J3}}$ from~\eqref{eq:CKkinematic}) involving $\phi_0$ as follows:
        \begin{equation}\label{eq:Jacobi_triple}
            \langle \phi_0,\sfm_2(T_1,\sfb\sfm_2(T_2,T_3))\rangle+\langle \phi_0,\sfm_2(T_2,\sfb\sfm_2(T_3,T_1))\rangle+\langle \phi_0,\sfm_2(T_3,\sfb\sfm_2(T_1,T_2))\rangle~,
        \end{equation}
        where $T_1$, $T_2$, and $T_3$ are currents, making up the rest of the diagrams. Again, in all physically interesting examples, $\sfb$ is its own adjoint, and hence we have
        \begin{equation}
            \langle \phi_0,\sfm_2(T_1,\sfb\sfm_2(T_2,T_3))\rangle\ =\ \langle \sfb\psi,\sfm_2(T_1,\sfb\sfm_2(T_2,T_3))\rangle\ =\ \langle \psi,\sfb\sfm_2(T_1,\sfb\sfm_2(T_2,T_3))\rangle~.
        \end{equation}
        If the derived bracket is a Lie bracket, then this reformulation makes it clear that~\eqref{eq:Jacobi_triple} indeed vanishes. We note that, due to cyclic symmetry of the amplitudes, it is sufficient if at least one external field is $\sfb$-exact.
        
        \paragraph{Underlying algebraic structure.}
        Ultimately, the dg commutative algebra $(\frB,\sfd,\sfm_2)$ and the differential $\sfb$ will form the structure of a $\BVbox$-algebra~\cite{Reiterer:2019dys,Borsten:2022vtg}, see also~\cite{Akman:1995tm}. We shall formalise and explore these in the remainder of this paper. Moreover, we shall extend this picture to CK duality involving matter (i.e.~fields taking values in representations of the gauge group that can be different from the adjoint representation). This leads to the notion of $\BVbox$-modules, following the discussion of~\cite{Borsten:2023reb}.
        
        \paragraph{Comment regarding the loop level.}
        Consider now the dg Lie algebra of a cubic BV action $S$ which has been gauge-fixed in the usual manner. Suppose that the dg Lie algebra structure can be colour--stripped and enhanced to a $\BVbox$-algebra with $\BBox$ the d'Alembertian $\wave$ and with a second-order differential $\sfb$. Using the Feynman rules following from $S$, we can write down the loop integrand for a Feynman diagram corresponding to a process by using the propagator $\frac{\sfb}{\wave}$ for each internal edge, the cubic interaction $[-,-]_\frg\otimes\sfm_2(-,-)$ for the vertices, and the cyclic structure to join loops formally. The resulting integrand for a trivalent graph $\Gamma$ is then of the form
        \begin{equation}
            I_\Gamma\ =\ \frac{\mathsf c_\Gamma N_\Gamma}{\prod_{e\in E(\Gamma)}\wave_e}~,
        \end{equation}
        where $N_\Gamma$ is a series of contractions of $\sfm_2$ and $\sfb$.
        
        Note that we can cut all loops open so that the loop diagram $\Gamma$ reduces to a tree. In this tree, we can use the derived bracket~\eqref{eq:kinematic_bracket} to bring all vertices to the form $[-,-]_\frg\otimes\{-,-\}$, cf.~\eqref{eq:CK_diagram}, as long as all fields attached to incoming lines are in $\ker(\sfb)$. Since we are working with a gauge-fixed action, there are no anti-fields running inside loops, so the above condition holds. Altogether, our vertices are described by pairs of Lie algebra structure constants, and CK duality holds at the level of loop integrands. 
        
        We note that the situation regarding the number of $\sfb$-operators that made the transition from currents to amplitudes subtle in the case of tree diagrams is absent for loops: each loop adds a propagator relative to the tree diagrams, increasing the number of $\sfb$-operators by one.
        
        \section{Colour--kinematics duality from \texorpdfstring{$\BVbox$}{BV-box}-algebras and their modules}\label{sec:BVBoxAlgebrasCKDuality}
        
        In this section, we fully develop the mathematical tools for an algebraic description of kinematic Lie algebras and colour--kinematics duality.
        
        \subsection{Pseudo-\texorpdfstring{$\BVbox$}{BV-box}-algebras and kinematic Lie algebras}
        
        \paragraph{Pseudo-$\BVbox$-algebras.}
        We start with the most general definition of an algebra that implies the existence of a kinematic Lie algebra.
        
        \begin{definition}\label{def:pre-BVbox-algebra}
            A \uline{pseudo-$\BVbox$-algebra} is a tuple $(\frB,\sfd,\sfm_2,\sfb)$ such that $(\frB,\sfd,\sfm_2)$ is a dg commutative algebra\footnote{We shall always assume that $\sfm_2$ is associative, that is, $\sfm_2(\sfm_2(\phi_1,\phi_2),\phi_3)=\sfm_2(\phi_1,\sfm_2(\phi_2,\phi_3))$ for all $\phi_{1,2,3}\in\frB$.} endowed with an additional differential $\sfb\colon\frB\rightarrow\frB$ of degree~$-1$ such that the \uline{derived bracket}
            \begin{equation}\label{eq:derived_bracket}
                \{\phi_1,\phi_2\}\ \coloneqq\ \sfb(\sfm_2(\phi_1,\phi_2))-\sfm_2(\sfb\phi_1,\phi_2)-(-1)^{|\phi_1|}\sfm_2(\phi_1,\sfb\phi_2)
            \end{equation}
            for all $\phi_{1,2}\in\frB$ defines a shifted Lie algebra. That is, besides the \uline{shifted anti-symmetry}\footnote{It is shifted graded anti-symmetric since the bracket carries a degree. We choose to work with this convention for shifted algebras, which is operadically natural, in order to simplify later discussion.}
            \begin{subequations}\label{eq:BV_GB1}
                \begin{equation}\label{eq:BV_GB_antisymmetry}
                    \{\phi_1,\phi_2\}\ =\ (-1)^{|\phi_1||\phi_2|}\{\phi_2,\phi_1\}
                \end{equation}
                implied by~\eqref{eq:derived_bracket}, we also have the \uline{shifted Jacobi identity}
                \begin{equation}\label{eq:BV_GB_Jacobi}
                    \begin{aligned}
                        \{\phi_1,\{\phi_2,\phi_3\}\}\ =\ (-1)^{|\phi_1|+1}\{\{\phi_1,\phi_2\},\phi_3\}+(-1)^{(|\phi_1|+1)(|\phi_2|+1)}\{\phi_2,\{\phi_1,\phi_3\}\}
                    \end{aligned}
                \end{equation}
            \end{subequations}
            for all $\phi_{1,2,3}\in\frB$. Furthermore, we set
            \begin{equation}\label{eq:box}
                \BBox\ \coloneqq\ [\sfd,\sfb]\ =\ \sfd\circ\sfb+\sfb\circ\sfd~.
            \end{equation}
        \end{definition}
        
        \noindent
        Hence, the derived bracket measures the failure of $\sfb$ to be a derivation for $\sfm_2$. Note that $[\sfd,\BBox]=0=[\sfb,\BBox]$. 
        
        A pseudo-$\BVbox$-algebra will turn out sufficient for describing CK duality of currents, but in order to extend the picture to amplitudes, we will also need a cyclic structure or metric.
        
        \begin{definition}\label{def:cyclicPreBVBoxAlgebra}
            A \uline{metric pseudo-$\BVbox$-algebra} is a pseudo-$\BVbox$-algebra $(\frB,\sfd,\sfm_2,\sfb)$ endowed with a non-degenerate graded symmetric bilinear map
            \begin{subequations}\label{eq:metricPreBVBox}
                \begin{equation}
                    \inner{-}{-}\,:\,\frB\times\frB\ \rightarrow\ \IR~,
                \end{equation}
                called a \uline{cyclic structure}, \uline{metric}, or \uline{inner product}, which is compatible with the pseudo-$\BVbox$-algebra structure in the sense that
                \begin{equation}\label{eq:axiomsMetric}
                    \begin{aligned}
                        \inner{\sfd \phi_1}{\phi_2}+(-1)^{|\phi_1|}\inner{\phi_1}{\sfd\phi_2}\ &=\ 0~,
                        \\
                        \inner{\sfm_2(\phi_1,\phi_2)}{\phi_3}-(-1)^{|\phi_1||\phi_2|}\inner{\phi_2}{\sfm_2(\phi_1,\phi_3)}\ &=\ 0~,
                        \\
                        \inner{\sfb \phi_1}{\phi_2}-(-1)^{|\phi_1|}\inner{\phi_1}{\sfb\phi_2}\ &=\ 0
                    \end{aligned}
                \end{equation}
            \end{subequations}            
            for all $\phi_{1,2,3}\in\frB$. We say that $\inner{-}{-}$ is \uline{of degree~$n$} if $\inner{\phi_1}{\phi_2}\neq 0$ implies $|\phi_1|+|\phi_2|+n=0$ for all $\phi_{1,2}\in\frB$.
        \end{definition}
        
        \noindent
        Note that combining~\eqref{eq:metricPreBVBox} with~\eqref{eq:derived_bracket} and~\eqref{eq:box}, we see that
        \begin{equation}
            \inner{\{\phi_1,\phi_2\}}{\phi_3}-(-1)^{|\phi_1||\phi_2|}\inner{\phi_2}{\{\phi_1,\phi_3\}}\ =\ 0
            \eand
            \inner{\BBox\phi_1}{\phi_2}\ =\ \inner{\phi_1}{\BBox\phi_2}
        \end{equation}
        for all $\phi_{1,2,3}\in\frB$.
        
        We will want to use the operator $\sfh=\sfid_\frg\otimes \BBox^{-1}\sfb$ for some Lie algebra $\frg$ as the contracting homotopy in a special deformation retract~\eqref{eq:deformation_retract}, and this will produce a Feynman diagram expansion. Among the general choices, the following is particularly relevant.
        \begin{definition}\label{def:b-complete}
            We call the operator $\sfb$ in a $\BVbox$-algebra $(\frB,\sfd,\sfm_2,\sfb)$ \uline{complete} if $\BBox^{-1}\sfb$ is the contracting homotopy in a special deformation retract to the cohomology $H^\bullet_\sfd(\frB)$ of the cochain complex $(\frB,\sfd)$. 
        \end{definition}
        \noindent Note that in this definition, we consider a `colour-stripped' form of the homotopy transfer~\eqref{eq:deformation_retract}. Physically, a $\BVbox$-algebra with complete operator $\sfb$ comes with a natural Feynman diagram expansion in which all non-physical fields are propagating and hence integrated out. 
        
        \paragraph{Kinematic Lie algebras.}
        Importantly, the shifted Jacobi identity~\eqref{eq:BV_GB_Jacobi} allows us to associate a Lie algebra with a pseudo-$\BVbox$-algebra.
        
        \begin{definition}\label{def:kinematicLieAlgebraPreBVBox}
            Given a pseudo-$\BVbox$-algebra $(\frB,\sfd,\sfm_2,\sfb)$ with derived bracket~\eqref{eq:derived_bracket}, we call the associated Lie algebra $\frKin(\frB)$ given by\footnote{We use square brackets $[k]$ with $k\in\IZ$ to denote a degree shift for a graded vector space $V=\bigoplus_{i\in\IZ}V_i$ by $V[k]=\bigoplus_{i\in\IZ}(V[k])_i\coloneqq\bigoplus_{i\in\IZ}V_{i+k}$.}
            \begin{equation}
                \frKin(\frB)\ \coloneqq \ (\frB[1],[-,-]_{\frKin(\frB)})
                \ewith
                \big[\phi_1[1],\phi_2[1]\big]_{\frKin(\frB)}\ \coloneqq\ (-1)^{|\phi_1|}\{\phi_1,\phi_2\}[1]
            \end{equation}
            for all $\phi_{1,2}[1]\in\frKin(\frB)$ the \uline{kinematic Lie algebra}.
        \end{definition}
        
        \noindent
        We note that the map $\frKin$ extends to a functor from the evident category of pseudo-$\BVbox$-algebras to the category of Lie algebras. 
        
        Our discussion in \cref{ssec:kinematic_Lie_algebras}, in particular the argument around~\eqref{eq:CK_diagram}, now yields the following result.
        
        \begin{theorem}\label{thm:preBVBoxImpliesKinematicLieAlgebra}
            A cubic gauge field theory comes with a kinematic Lie algebra if its underlying dg Lie algebra $(\frL,\mu_1,\mu_2)$ factorises into a Lie algebra $(\frg,[-,-]_\frg)$ and a pseudo-$\BVbox$-algebra $(\frB,\sfd,\sfm_2,\sfb)$ such that $\frL\cong\frg\otimes\frB$ and
            \begin{equation}
                \begin{aligned}
                    \mu_1(\tau_1\otimes\phi_1)\ &=\ \tau_1\otimes\sfd\phi_1~,
                    \\
                    \mu_2(\tau_1\otimes\phi_1,\tau_2\otimes\phi_2)\ &=\ [\tau_1,\tau_2]_\frg\otimes\sfm_2(\phi_1,\phi_2)
                \end{aligned}
            \end{equation}
            for all $\tau_{1,2}\in\frg$ and $\phi_{1,2}\in\frB$.
        \end{theorem}
        
        \noindent
        Note that $\frKin(\frB)$ together with $\sfd$ generally fails to be a dg Lie algebra as the following proposition makes clear.
        
        \begin{proposition}\label{prop:algebraRelationsDB}
            For any pseudo-$\BVbox$-algebra $(\frB,\sfd,\sfm_2,\sfb)$ with derived bracket~\eqref{eq:derived_bracket}, we have
            \begin{equation}
                \begin{aligned}
                    \sfd\{\phi_1,\phi_2\}\ &=\ -\{\sfd\phi_1,\phi_2\}-(-1)^{|\phi_1|}\{\phi_1,\sfd\phi_2\}
                    \\
                    &\kern1cm+\BBox(\sfm_2(\phi_1,\phi_2))-\sfm_2(\BBox\phi_1,\phi_2)-\sfm_2(\phi_1,\BBox\phi_2)~,
                    \\
                    \sfb\{\phi_1,\phi_2\}\ &=\ -\{\sfb\phi_1,\phi_2\}-(-1)^{|\phi_1|}\{\phi_1,\sfb\phi_2\}
                \end{aligned}
            \end{equation}
            for all $\phi_{1,2}\in\frB$. 
        \end{proposition}
        
        \begin{proof}
            This follows from a straightforward calculation using the definition of the derived bracket~\eqref{eq:derived_bracket} together with the definition~\eqref{eq:box} of $\BBox$ and the fact that both $\sfd$ and $\sfb$ are differentials. The second equation has already been observed in~\cite{koszul1985crochet}, see also~\cite{Akman:1995tm}. 
        \end{proof}
        
        \noindent
        Put differently, this proposition says that, whilst $\sfb$ is a derivation for the derived bracket, $\sfd$ is not. This proposition also implies the following.
        \begin{corollary}
            With respect to the derived bracket~\eqref{eq:derived_bracket}, $\ker(\sfb)$ is closed. In fact, \eqref{eq:derived_bracket} implies that
            \begin{equation}
                \{\ker(\sfb),\ker(\sfb)\}\ \subseteq\ \im(\sfb)\ \subseteq\ \ker(\sfb)~.
            \end{equation}
        \end{corollary}
        
        Thus, in \cref{def:kinematicLieAlgebraPreBVBox}, we may restrict the kinematic Lie algebra $\frK\coloneqq \frKin(\frB)$ to a shifted Lie subalgebra $\tilde \frK$ with
        \begin{equation}\label{eq:fields_inbetween}
            \im(\sfb)[1]\ \subseteq \ \tilde\frK\ \subseteq \ \ker(\sfb)[1]~.
        \end{equation}
        For most physically interesting field theories, such as e.g.~Yang--Mills theory, we have $\im(\sfb)=\frF=\ker(\sfb)$, where $\frF$ is the space of fields (as opposed to anti-fields), at least after gauge fixing. For other theories, such as e.g.~Chern--Simons theory, $\im(\sfb)$ may be smaller than $\ker(\sfb)$ in general, but after gauge fixing, the space of fields $\frF$ satisfies~\eqref{eq:fields_inbetween}, as we shall see in \cref{ssec:gaugeFixedPreBVBox}. For an explicit example, see \cref{ssec:pure_Chern-Simons}. The kinematic Lie algebra that is usually discussed in the literature is the one restricted to fields, or further to physical fields. We therefore make the following definition:
        
        \begin{definition}\label{def:restricted_kinematic_Lie_algebra}
            The \uline{restricted kinematic Lie algebra} $\frKin^0(\frB)$ of a $\BVbox$-algebra $\frB$ is the Lie subalgebra
            \begin{equation}
                \frKin^0(\frB)\ \coloneqq\ \ker(\sfb)[1]\ \subseteq\ \frKin(\frB)~.
            \end{equation}
        \end{definition}
        
        \paragraph{Colour--kinematics duality.}
        We conclude with a sufficient criterion for CK duality. There are several restrictions for a theory with kinematic Lie algebra or, equivalently, pseudo-$\BVbox$-algebra to exhibit traditional CK duality.
        
        A theory with a pseudo-$\BVbox$-algebra $(\frB,\sfd,\sfm_2,\sfb)$ will produce a Feynman diagram expansion of currents that is naturally of the form~\eqref{eq:CK_currents_parameterization}. The `amputated correlators', i.e.~the currents paired off with the final propagators removed and paired off with fields using the cyclic structure, have a Feynman diagram expansion of the form~\eqref{eq:CK_amplitudes_parameterization} if at least one of the external fields lies in the image of $\sfb$.
        
        If now the operator $\sfb$ is complete, then in the Feynman diagram expansion all non-physical fields are propagating and hence integrated out. The above currents and amputated correlators become `physical currents' and `physical amplitudes' with expansions~\eqref{eq:CK_currents_parameterization} and~\eqref{eq:CK_amplitudes_parameterization}. 
        
        Finally, if the operator $\BBox$ is the d'Alembertian on the underlying space-time, then the amplitude parametrisation~\eqref{eq:CK_amplitudes_parameterization} is of the form conventionally discussed in the literature, i.e.~$d_\gamma$ is the product of $\frac{1}{p^2_\ell}$ ranging over all internal lines $\ell$. \cref{thm:preBVBoxImpliesKinematicLieAlgebra} therefore has the following immediate corollary.
        
        \begin{corollary}\label{cor:pure_gauge_CK_duality}
            Consider a cubic gauge field theory whose underlying dg Lie algebra factorises into a Lie algebra and a pseudo-$\BVbox$-algebra $(\frB,\sfd,\sfm_2,\sfb)$ with complete operator $\sfb$ and $\BBox=\wave$. Then the corresponding Feynman diagram expansion 
            yields a CK-dual parametrisation of the currents~\eqref{eq:CK_currents_parameterization} and a CK-dual parametrisation of the amplitudes~\eqref{eq:CK_amplitudes_parameterization} with at least one external field in the image of $\sfb$.
        \end{corollary}
        
        \noindent
        We note that, when considering physical amplitudes, the physical fields $\phi$ usually satisfy the gauge condition $\sfb\phi=0$, cf.~the examples in~\cref{sec:examples}. Moreover, in most physically interesting cases, the cohomology of $\sfb$ is trivial, so that a pseudo-$\BVbox$-algebra with structure $\BBox=\wave$ and all non-physical modes propagating directly implies CK duality of the amplitudes. 
        
        We also note that a CK-dual field theory does not necessarily have to have a kinematic Lie algebra. In particular, the parametrisation~\eqref{eq:CK_amplitudes_parameterization} does not have to come from the Feynman diagram expansion obtained from a path integral. 
        
        A pseudo-$\BVbox$-algebra structure as in \cref{cor:pure_gauge_CK_duality} with complete $\sfb$ and $\BBox=\wave$ implies full, off-shell CK duality of all tree-level correlators. Given an anomaly-free path-integral measure completing the action to a quantum theory, this is sufficient to obtain full loop level CK duality as we shall see later. In many concrete examples, however, CK duality only exists at the tree level, and this is then visible in various obstacles to obtain the above mentioned situation. For example, we saw that the field redefinitions introduced in~\cite{Borsten:2021gyl} to reformulate the Yang--Mills action such that it has an underlying pseudo-$\BVbox$-algebra introduced Jacobian counterterms leading to anomalies. In another case, the twistor description of supersymmetric Yang--Mills theory that where used to produce pseudo-$\BVbox$-algebra descriptions in~\cite{Borsten:2022vtg} come with a non-standard $\BBox$-operator. Finally, in the case of pure spinors~\cite{Borsten:2023reb}, the tree-level constructions did not lift to the loop level, as there was again a problem with the regularisation, cf.~\cref{ssec:pure_spinors_SYM}. This problem is expected and unavoidable due to the results of~\cite{Bern:2015ooa}.
        
        \subsection{Modules over pseudo-\texorpdfstring{$\BVbox$}{BV-box}-algebras}\label{ssec:BVbox-modules}
        
        \paragraph{Pseudo-$\BVbox$-modules.}
        For CK-dual field theories involving matter fields, that is, fields which do not take values in the gauge Lie algebra $\frg$, we need to extend the concept of a pseudo-$\BVbox$-algebra to a pseudo-$\BVbox$-module.
        
        \begin{definition}\label{def:preBVBoxModule}
            A \uline{module over a pseudo-$\BVbox$-algebra} $(\frB,\sfd_\frB,\sfm_2,\sfb_\frB)$ is a tuple $(V,\sfd_V,\acton_V,\sfb_V)$ such that $(V,\sfd_V,\acton_V)$ is a (left) module over the dg commutative algebra $(\frB,\rmd_\frB,\sfm_2)$ with the action $\acton_V\colon\frB\times V\rightarrow V$ of degree~$0$ and which is endowed with an additional differential $\sfb_V\colon V\rightarrow V$ of degree~$-1$ such that the \uline{derived bracket}
            \begin{equation}\label{eq:derived_bracket_module}
                \{\phi,v\}_V\ \coloneqq\ \sfb_V(\phi\acton_Vv)-(\sfb_\frB\phi)\acton_Vv-(-1)^{|\phi|}\phi\acton_V(\sfb_Vv)
            \end{equation}
            for all $\phi\in\frB$ and $v\in V$ satisfies
            \begin{equation}\label{eq:derivedBracketModuleJacobi}
                \{\phi_1,\{\phi_2,v\}_V\}_V\ =\ (-1)^{|\phi_1|+1}\{\{\phi_1,\phi_2\}_\frB,v\}_V+(-1)^{(|\phi_1|+1)(|\phi_2|+1)}\{\phi_2,\{\phi_1,v\}_V\}_V
            \end{equation}
            for all $\phi_{1,2}\in\frB$ and $v\in V$, where $\{-,-\}_\frB$ is the derived bracket~\eqref{eq:derived_bracket}. Furthermore, we set
            \begin{equation}
                \BBox_V\ \coloneqq\ [\sfd_V,\sfb_V]\ =\ \sfd_V\circ\sfb_V+\sfb_V\circ\sfd_V~.
            \end{equation}
            Finally, in analogy with~\cref{def:b-complete}, we call the operator $\sfb_V$ \uline{complete} if $\BBox_V^{-1}\otimes \sfb_V$ is the contracting homotopy in a special deformation retract to the cohomology $H^\bullet_{\sfd_V}(V)$ of the cochain complex $(V,\sfd_V)$. 
        \end{definition}
        \noindent
        When there is no confusion, we will drop the subscripts $V$ and $\frB$ on all the operations. We also note that, for all physical applications, pseudo-$\BVbox$-modules with $V$ concentrated in degrees $1$ (fields) and $2$ 
        (anti-fields) will turn out to be sufficient. 
        
        Just as for $\BVbox$-algebras, we also need to introduce a metric to talk about action principles and amplitudes.
        \begin{definition}
            A \uline{metric} of degree~$n$ on a module $(V,\sfd_V)$ over a dg Lie algebra $(\frg,\sfd_\frg)$ is a non-degenerate bilinear graded-symmetric map of degree~$n$
            \begin{equation}
                \inner{-}{-}_V\,:\,V\times V\ \rightarrow\ \IR
            \end{equation}
            such that
            \begin{equation}
                \begin{aligned}
                    \inner{v_1}{\sfd_Vv_2}_V+(-1)^{|v_1|}\inner{\sfd_Vv_1}{v_2}_V\ &=\ 0~,
                    \\
                    \inner{\phi\acton_V v_1}{v_2}_V-(-1)^{|\phi||v_1|}\inner{v_2}{\phi\acton_V v_1}_V\ &=\ 0
                \end{aligned}
            \end{equation}
            for all $v_{1,2}\in V$ and $\phi\in\frg$.
            A \uline{metric dg Lie module} is a dg Lie module equipped with a metric.
            
            A \uline{metric} on a pseudo-$\BVbox$ module is defined in the same way, with the evident compatibility condition with $\sfb_V$; a \uline{metric pseudo-$\BVbox$ module} is a pseudo-$\BVbox$ module equipped with a metric.
        \end{definition}
        
        Note that on a cyclic module $V$ over a cyclic dg Lie algebra $\frg$, one can define a graded-anti-symmetric bilinear operation $\wedge_V$ as
        \begin{equation}
            \inner{\phi}{v_1\wedge_V v_2}_\frg\ \coloneqq\ \inner{\phi\acton_V v_1}{v_2}_V
        \end{equation}
        for any $v_{1,2}\in V$ and $\phi\in\frg$. Similarly, on a cyclic module $V$ over a cyclic pseudo-$\BVbox$-algebra $\frB$, one can define a graded-symmetric bilinear operation $\bullet_V$ as
        \begin{equation}
            \inner{\phi}{v_1\bullet_Vv_2}_\frB\ \coloneqq\ \inner{\phi\acton_V v_1}{v_2}
        \end{equation}
        for any $v_{1,2}\in V$ and $\phi\in\frB$.
        
        We now have the following result.
        
        \begin{proposition}\label{prop:kinematic_module}
            Given a module $V=(V,\sfd_V,\acton_V,\sfb_V)$ over a pseudo-$\BVbox$-algebra $\frB=(\frB,\sfd_\frB,\sfm_2,\sfb_\frB)$, we have a graded (left) module $(\frV,\acton_\frV)$ over the kinematic Lie algebra $\frKin(\frB)$ with $\frV\coloneqq V[1]$ and
            \begin{equation}\label{eq:kinematic_module_action}
                \begin{aligned}
                    \acton_\frV\,:\,\frKin(\frB)\times\frV\ &\rightarrow\ \frV~,
                    \\
                    \phi[1]\acton_\frV v[1]\ &\coloneqq\ (-1)^{|\phi|}\{\phi,v\}_V[1]
                \end{aligned}
            \end{equation}
            for all $\phi[1]\in\frKin(\frB)$ and $v[1]\in\frV$ with $\{-,-\}_V$ denoting the derived bracket~\eqref{eq:derived_bracket_module}.
        \end{proposition}
        
        \begin{proof}
            By direct calculation, cf.~\cref{app:postponed_proofs}.
        \end{proof}
        
        \paragraph{Gauge--matter colour--kinematics duality.}
        It is now easy to see that these structures are the appropriate ones for capturing gauge--matter CK duality. Firstly, as a direct extension of \cref{thm:preBVBoxImpliesKinematicLieAlgebra}, we have the following result.
        
        \begin{theorem}
            A cubic gauge--matter theory has a kinematic Lie algebra with Lie algebra module if its underlying dg Lie algebra factorises into a Lie algebra representation and a pseudo-$\BVbox$-algebra with pseudo-$\BVbox$-module.
        \end{theorem}
        
        \noindent
        Explicitly, we consider the Feynman diagram expansion induced by the pseudo-$\BVbox$-algebra and its module, which uses the propagator $\sfid_\frg\otimes \BBox_\frB^{-1}\sfb_\frB+\sfid_V\otimes \BBox_V^{-1}\sfb_V$. The operators $\sfb$ are then moved from the propagators to the interaction vertices, as indicated in~\eqref{eq:CK_diagram}. This turns the interaction vertices into derived brackets of the form~\eqref{eq:derived_bracket} or~\eqref{eq:derived_bracket_module}. Hence, the Feynman diagram expansion of currents possesses a kinematic Lie algebra with Lie algebra module, which extends to amplitudes with at least one external leg in the image of $\sfb_\frB$ or $\sfb_V$.
        
        As in the pure gauge case, the above theorem has the following corollary, the analogue of~\cref{cor:pure_gauge_CK_duality}, which provides a sufficient criterion for gauge--matter theories to possess CK duality.
        
        \begin{corollary}
            The Feynman diagram expansion of a cubic gauge--matter theory whose underlying dg Lie algebra factorises into a Lie algebra representation and a pseudo-$\BVbox$-algebra $(\frB,\sfd_\frB,\sfm_2,\sfb_\frB)$ with $\BBox_\frB=\wave$ together with a module $(V,\sfd_V,\acton_V,\sfb_V)$ over a pseudo-$\BVbox$-algebra with $\BBox_V=\wave$ and both $\sfb_\frB$ and $\sfb_V$ complete yields a gauge--matter CK-dual parametrisation of the physical currents and a gauge--matter CK-dual parametrisation of the physical amplitudes with at least one external field in the image of $\sfb_\frB$ or $\sfb_V$. 
        \end{corollary}
        
        \subsection{Pseudo-\texorpdfstring{$\BVbox$}{BV-box}-algebras and their modules over Hopf algebras}\label{ssec:BVbox-modules-over-Hopf}
        
        For technical reasons, it is convenient to define and work with the notion of a pseudo-$\BVbox$-algebra over a Hopf algebra, following~\cite{Reiterer:2019dys}. The technical reasons are twofold. Firstly, in future work~\cite{Borsten:2022ouu}, we intend to give the full homotopy algebraic picture, lifting the restriction to cubic actions; in this case, it is convenient to work with the framework of operadic Koszul duality, for which the Hopf algebra (that provides an ambient symmetric monoidal category) will be necessary. Secondly, our discussion of the double copy to ordinary space-time (as opposed to a double field theory on doubled space) is most easily understood using tensor products over Hopf algebras.
        
        \paragraph{Hopf algebras.}
        Let us first recall some relevant definitions.
        
        \begin{definition}\label{def:HopfAlgebra}
            A \uline{bialgebra} over $\IR$ is a tuple $(\frH,\Delta,\epsilon,\unit)$, where $(\frH,\unit)$ is an associative unital algebra over $\IR$ and $\Delta\colon\frH\to\frH\otimes\frH$ (the \uline{coproduct}) and $\epsilon\colon\frH\to\IR$ (the \uline{counit}) are unital homomorphisms of $\IR$-algebras such that $\Delta$ is coassociative,
            \begin{equation}\label{eq:coassociativity}
                (\Delta\otimes \id_\frH)\Delta\ =\ (\id_\frH\otimes \Delta)\Delta~,
            \end{equation}
            and $\epsilon$ is indeed a counit,
            \begin{equation}\label{eq:counitality}
                (\id_\frH\otimes\,\epsilon)\Delta\ =\ \id_\frH\ =\ (\epsilon\otimes\id_\frH)\Delta~.
            \end{equation}
        \end{definition}
        It will be convenient to use the common (sumless) Sweedler notation 
        \begin{equation}\label{eq:Sweedler_notation}
            \chi^{(1)}\otimes\chi^{(2)}\ \coloneqq\ \Delta(\chi)
        \end{equation}
        for $\chi\in \frH$, and in this notation, \eqref{eq:coassociativity} and~\eqref{eq:counitality} read as 
        \begin{equation}
            \begin{aligned}
                \chi^{(1)}\otimes\big((\chi^{(2)})^{(1)}\otimes(\chi^{(2)})^{(2)}\big)\ &=\ \big((\chi^{(1)})^{(1)}\otimes(\chi^{(1)})^{(2)}\big)\otimes\chi^{(2)}~, 
                \\ 
                \epsilon(\chi^{(1)})\chi^{(2)}\ =\ &~\chi \ =\ \chi^{(1)}\epsilon(\chi^{(2)})~. 
            \end{aligned}
        \end{equation}
        
        \begin{definition}
            A bialgebra $(\frH,\Delta,\epsilon)$ is called \uline{commutative} if the algebra $\frH$ is commutative; it is called \uline{cocommutative} if it satisfies the condition
            \begin{equation}
                \chi^{(1)}\otimes\chi^{(2)}\ =\ \chi^{(2)}\otimes\chi^{(1)}
            \end{equation}
            for all $\chi\in\frH$.
            
            A \uline{Hopf algebra} over $\IR$ is a tuple $(\frH,\Delta,\epsilon,S)$ where $(\frH,\Delta,\epsilon)$ is a bialgebra and where $S\colon\frH\rightarrow\frH$ is an $\IR$-linear map (the \uline{antipode}) such that
            \begin{equation}
                S(\chi^{(1)})\chi^{(2)}\ =\ \chi^{(1)}S(\chi^{(2)})\ =\ \epsilon(\chi)\unit
            \end{equation}
            for all $\chi\in\frH$.
        \end{definition}
        
        \noindent
        In the following, we shall always work with restrictedly tensorable (see \cref{def:restrictedly_tensorable}) cocommutative Hopf algebras over $\IR$.\footnote{In this paper, we do not really need the antipode, so it suffices to work with bialgebras. However, the antipode will become important for operadic Koszul duality.} A trivial example of such a Hopf algebra is $\IR$ itself with the ordinary product and all other maps trivial. Another important example to our discussion is the following.
        
        \begin{example}\label{ex:H-Box-Minkowski}
            Let $\IM^d\coloneqq\IR^{1,d-1}$ be $d$-dimensional Minkowski space with metric tensor $\eta=\diag(-1,1\ldots,1)$ and Cartesian coordinates $x^\mu$ with $\mu,\nu,\ldots=0,\ldots,d-1$. The Hopf algebra $\frH_{\IM^d}$ is the Hopf algebra of differential operators with constant coefficients on $\IM^d$ that is generated by the partial derivatives $\parder{x^\mu}$. 
            
            Explicitly, $\frH_{\IM^d}$ is the vector space of power series in the partial derivative $\parder{x^\mu}$ with unit $\unit=1$ and evident product. The coproduct on elements in $\frH_{\IM^d}$ is fully defined by unitality and the Leibniz rule, 
            \begin{equation}
                \Delta(1)\ =\ 1\otimes 1
                \eand
                \Delta\left(\parder{x^\mu}\right)\ =\ \parder{x^\mu}\otimes 1+1\otimes \parder{x^\mu}~,
            \end{equation}
            and the counit is the projection onto the constant part of the power series, i.e.
            \begin{equation}
                \eps(1)\ =\ 1
                \eand 
                \eps\left(\parder{x^\mu}\right)\ =\ 0~.
            \end{equation}
            Finally, the antipode is defined by
            \begin{equation}
                S(1)\ =\ 1~,~~~S(\chi_1\chi_2)=S(\chi_2)S(\chi_1)~,\eand
                S\left(\parder{x^\mu}\right)\ =\ -\parder{x^\mu}~.
            \end{equation}
            This Hopf algebra is evidently commutative (hence restrictedly tensorable) and cocommutative.
        \end{example}
        
        \paragraph{Pseudo-$\BVbox$-algebras and modules over Hopf algebras.}
        We start with the obvious notion of a dg commutative algebra over $\frH$.
        
        \begin{definition}\label{def:cdgaHopfAlgebra}
            A \uline{differential graded (dg) commutative algebra} over a cocommutative Hopf algebra $\frH$ is a tuple $(\frC,\sfd,\sfm_2,\acton)$ such that $(\frC,\sfd,\sfm_2)$ is a dg commutative algebra, $(\frC,\acton)$ is a graded (left) module over $\frH$ with an action $\acton\colon\frH\times\frC\rightarrow\frC$ of degree~$0$, and the differential $\sfd$ and the product $\sfm_2$ are $\frH$-linear in the sense that
            \begin{equation}
                \begin{aligned}
                    \chi\acton\sfd\phi_1\ &=\ \sfd(\chi\acton\phi_1)~,
                    \\
                    \chi\acton\sfm_2(\phi_1,\phi_2)\ &=\ \sfm_2(\chi^{(1)}\acton\phi_1,\chi^{(2)}\acton\phi_2)
                \end{aligned}
            \end{equation}
            for all $\chi\in\frH$ and $\phi_{1,2}\in\frC$, where we use again the Sweedler notation~\eqref{eq:Sweedler_notation}.
        \end{definition}
        
        \noindent
        This notion extends to pseudo-$\BVbox$-algebras over $\frH$, where we additionally demand that $\BBox\in\frH$.
        
        \begin{definition}\label{def:preBVBoxAlgebraHopfAlgebra}
            A \uline{pseudo-$\BVbox$-algebra} over a cocommutative Hopf algebra $\frH$ is a tuple $(\frB,\sfd,\sfm_2,\sfb,\acton)$ such that $(\frB,\sfd,\sfm_2,\sfb)$ is a pseudo-$\BVbox$-algebra, $(\frB,\sfd,\sfm_2,\acton)$ is a dg commutative algebra over $\frH$, the differential $\sfb$ is linear over $\frH$, i.e.\
            \begin{equation}
                \chi\acton(\sfb\phi)\ =\ \sfb(\chi\acton\phi)
            \end{equation}
            for all $\chi\in\frH$ and $\phi\in\frB$, and there is a $\BBox_\frH\in \frH$ such that $\BBox\phi=[\sfd,\sfb]\phi=\BBox_\frH\acton \phi$ for all $\phi\in \frB$. (In the following, we will be sloppy and identify $\BBox_\frH=\BBox$ or even write $[\sfd,\sfb]\in \frH$.)
            
            A \uline{metric pseudo-$\BVbox$-algebra} over a cocommutative Hopf algebra $\frH$ is a pseudo-$\BVbox$-algebra $\frB$ equipped with a metric $\inner{-}{-}:\frB\otimes_\IR\!\frB\rightarrow\IR$ that is an $\frH$-linear map, where $\IR$ is equipped with the trivial $\frH$-module structure and $\frB\otimes_\IR\!\frB$ is equipped with the $\frH$-module structure induced by the coproduct.
        \end{definition}
        
        \noindent 
        It remains to extend the notion of a pseudo-$\BVbox$-module to a pseudo-$\BVbox$-module over $\frH$.
        
        \begin{definition}\label{def:preBVBoxModuleHopfAlgebra}
            A \uline{pseudo-$\BVbox$-module} over a cocommutative Hopf algebra $\frH$ is a module $(V,\sfd_V,\acton_V,\sfb_V)$ over a pseudo-$\BVbox$-algebra $(\frB,\sfd_\frB,\sfm_2,\sfb_\frB,\acton_\frB)$ over $\frH$ such that all maps are $\frH$-linear in the sense that
            \begin{equation}
                \begin{aligned}
                    \chi\acton_V(\sfd_Vv)\ &=\ \sfd_V(\chi\acton_Vv)~,
                    \\
                    \chi\acton_V(\sfb_Vv)\ &=\ \sfb_V(\chi\acton_Vv)~,
                    \\
                    \chi\acton_V\{\phi,v\}_V\ &=\ \{\chi^{(1)}\acton_\frB\phi,\chi^{(2)}\acton_Vv\}_V
                \end{aligned}
            \end{equation}
            for all $h\in\frH$, $\phi\in\frB$, and $v\in V$. Here, $\{-,-\}_V$ is the derived bracket~\eqref{eq:derived_bracket_module} associated with $(V,\sfd_V,\acton_V,\sfb_V)$. In addition, we require that $\BBox=[\sfd_\frB,\sfb_\frB]=[\sfd_V,\sfb_V]$.
            
            A \uline{cyclic module over a cyclic pseudo-$\BVbox$-algebra} $\frB$ over a cocommutative Hopf algebra $\frH$ is a pseudo-$\BVbox$-module $V$ equipped with a metric $\langle-,-\rangle\colon V\otimes_\IR\! V\to V$ that is a $\frH$-linear map, where $\IR$ is equipped with the trivial $\frH$-module structure and $V\otimes_\IR\! V$ is equipped with the $\frH$-module structure induced by the coproduct. 
        \end{definition}
        
        \subsection{\texorpdfstring{$\BVbox$}{BV-box}-algebras and their modules}
        
        As we will see, it is both physically and mathematically natural to specialise our pseudo-$\BVbox$-algebras to the case of $\BVbox$-algebras~\cite{Reiterer:2019dys}. These are pseudo-$\BVbox$-algebras in which the operator $\sfb$ is a second-order differential in the sense of\footnote{This concept was first defined for commutative and associative algebras by Koszul~\cite{koszul1985crochet}. Here, we choose to work with the more flexible definition in~\cite{Akman:1995tm}, which extends to non-commutative and non-associative algebras.} Akman~\cite{Akman:1995tm}. We start by recalling the notion of higher-order differentials.
        
        \paragraph{Higher-order differentials.}
        Consider a graded vector space $\frA$ with a multilinear operation $\sfm$ of arity\footnote{The generalisation from binary $\sfm_2$ to arbitrary arity was not considered in~\cite{Akman:1995tm} but is straightforward, although it is not needed in this paper. In particular, if a theory has a 3-Lie algebra~\cite{Filippov:1985aa} colour structure and a corresponding quartic-vertex CK duality~\cite{Huang:2012wr,Huang:2013kca,Sivaramakrishnan:2014bpa}, then the colour-stripped theory is naturally captured by an analogue of a (pseudo-)$\BVbox$-algebra with a totally graded-symmetric ternary $\sfm_3$ and a second-order differential operator.} $k+1$ and degree~$|\sfm|$ and a differential $\delta\colon\frA\rightarrow\frA$ of degree~$|\delta|$. For all $r\in\IN$, we define recursively the maps $\Phi^{r+1}_\delta$ by
        \begin{equation}\label{eq:def_Phir}
            \begin{aligned}
                \Phi^1_\delta(\phi_1)\ &\coloneqq\ \delta\phi_1~,
                \\
                \Phi^2_\delta(\phi_1,\dotsc,\phi_{k+1})\ &\coloneqq\ \Phi^1_\delta(\sfm(\phi_1,\dotsc,\phi_{k+1}))-(-1)^{|\sfm||\delta|}\sfm(\Phi^1_\delta(\phi_1),\phi_2,\dotsc,\phi_{k+1})-\dotsb\\&\hspace{1cm}-(-1)^{(|\sfm|+|\phi_1|+\dotsb+|\phi_k|)\,|\delta|}\sfm(\phi_1,\dotsc,\phi_k,\Phi^1_\delta(\phi_{k+1}))~,
                \\
                &\kern8pt\vdots
                \\
                \Phi^{r+1}_\delta(\phi_1,\dotsc,\phi_{rk+1})\ &\coloneqq\ \Phi^r_\delta(\phi_1,\dotsc,\phi_{(r-1)k},\sfm(\phi_{(r-1)k+1},\dotsc,\phi_{rk+1}))\\&\hspace{1cm}-(-1)^{|\sfm|\,(|\delta|+|\phi_1|+\dotsb+|\phi_{(r-1)k}|)}\\&\hspace{2cm}\times\sfm(\Phi^r_\delta(\phi_1,\dotsc,\phi_{(r-1)k+1}),\phi_{(r-1)k+2},\dotsc,\phi_{rk+1})\\&\hspace{1cm} - \dotsb
                \\
                &\hspace{1cm}-(-1)^{(|\sfm|+|\phi_{(r-1)k+1}|+\dotsb+|\phi_{rk}|)(|\phi_1|+\dotsb+|\phi_{(r-1)k}|+|\delta|)}\\
                &\hspace{2cm}\times\sfm(\phi_{(r-1)k+1},\dotsc,\phi_{rk},\Phi^r_\delta(\phi_1,\dotsc,\phi_{(r-1)k},\phi_{rk+1}))~,
            \end{aligned}
        \end{equation}
        for all $\phi_{1,\dotsc,rk+1}\in\frA$, which measure the failure of $\Phi^r_\delta(\phi_1,\ldots,\phi_{(r-1)k},-)$ to be a derivation of the $(k+1)$-ary product $\sfm$.
        
        \begin{definition}\label{def:order_differential_operator}
            A differential $\delta$ on $(\frA,\sfm_2)$ is said to be a \underline{differential operator of order $r$} if $\Phi^{r+1}_\delta=0$.
        \end{definition}
        
        \noindent
        Note that a differential of $r$-th order is automatically of order $r+1$.
        
        \begin{example}
            For a pseudo-$\BVbox$-algebra $(\frB,\sfd,\sfm_2,\sfb)$, the condition for $\sfb$ being of second order is
            \begin{equation}\label{eq:b_second_order}
                \begin{aligned}
                    \sfb(\sfm_2(\sfm_2(\phi_1,\phi_2),\phi_3))\ &=\ (-1)^{|\phi_1|}\sfm_2(\phi_1,\sfb(\sfm_2(\phi_2,\phi_3)))
                    \\
                    &\kern1cm+(-1)^{(|\phi_1|+1)|\phi_2|}\sfm_2(\phi_2,\sfb(\sfm_2(\phi_1,\phi_3)))
                    \\
                    &\kern1cm+(-1)^{|\phi_3|(|\phi_1|+|\phi_2|+1)}\sfm_2(\phi_3,\sfb(\sfm_2(\phi_1,\phi_2)))
                    \\
                    &\kern1cm+(-1)^{|\phi_1|+|\phi_3|+|\phi_2||\phi_3|+1}\sfm_2(\sfm_2(\phi_1,\phi_3),\sfb\phi_2)
                    \\
                    &\kern1cm+(-1)^{|\phi_1|+|\phi_2|+1}\sfm_2(\sfm_2(\phi_1,\phi_2),\sfb\phi_3)
                    \\
                    &\kern1cm+(-1)^{(|\phi_1|+1)(|\phi_2|+|\phi_3|)+1}\sfm_2(\sfm_2(\phi_2,\phi_3),\sfb\phi_1)
                \end{aligned}
            \end{equation}
            for all $\phi_{1,2,3}\in\frB$.
        \end{example}
        
        \begin{example}
            For a module $(V,\sfd_V,\acton_V,\sfb_V)$ over a pseudo-$\BVbox$-algebra $(\frB,\sfd_\frB,\sfm_2,\sfb_\frB)$, the condition for $\sfb_V$ being of second order amounts to
            \begin{equation}
                \begin{aligned}
                    \sfb_V(\phi_1\acton_V(\phi_2\acton_Vv))\ &=\ (-1)^{|\phi_1|}\phi_1\acton_V\sfb_V(\phi_2\acton_Vv)
                    \\
                    &\kern1cm+(-1)^{(|\phi_1|+1)|\phi_2|}\phi_2\acton_V\sfb_V(\phi_1\acton_Vv)
                    \\
                    &\kern1cm+\sfb_\frB(\sfm_2(\phi_1,\phi_2))\acton_Vv
                    \\
                    &\kern1cm+(-1)^{|\phi_1|+|\phi_2|+1}\phi_1\acton_V(\phi_2\acton_V(\sfb_Vv))
                    \\
                    &\kern1cm+(-1)^{|\phi_1|+1}\phi_1\acton_V((\sfb_\frB\phi_2)\acton_Vv)
                    \\
                    &\kern1cm-(\sfb_\frB\phi_1)\acton_V(\phi_2\acton v)
                \end{aligned}
            \end{equation}
            for all $\phi_1,\phi_2\in\frB$ and $v\in V$.
        \end{example}
        
        \paragraph{$\BVbox$-algebras.}
        We now refine \cref{def:pre-BVbox-algebra,def:cyclicPreBVBoxAlgebra} as follows. 
        
        \begin{definition}\label{def:BVBoxAlgebra}
            A \uline{(cyclic) $\BVbox$-algebra} is a (cyclic) pseudo-$\BVbox$-algebra $(\frB,\sfd,\sfm_2,\sfb)$ in which $\sfb$ is of second order.
        \end{definition}
        
        \noindent
        We have already seen in \cref{def:pre-BVbox-algebra} that the derived bracket~\eqref{eq:derived_bracket} for a  pseudo-$\BVbox$-algebra automatically satisfies the shifted anti-symmetry~\eqref{eq:BV_GB_antisymmetry}. The operator $\sfb$ being of second order now implies that also shifted Jacobi identity~\eqref{eq:BV_GB_Jacobi} automatically holds as the following proposition shows.
        
        \begin{proposition}\label{prop:BVbox_is_pre_BVbox}
            Let $(\frB,\sfd,\sfm_2,\sfb)$ be a pseudo-$\BVbox$-algebra. The condition that $\sfb$ is of second order is equivalent to the \uline{shifted Poisson identity}
            \begin{equation}\label{eq:BV_GB_Poisson}
                \{\phi_1,\sfm_2(\phi_2,\phi_3)\}\ =\ \sfm_2(\{\phi_1,\phi_2\},\phi_3)+(-1)^{(|\phi_1|+1)|\phi_2|}\sfm_2(\phi_2,\{\phi_1,\phi_3\})
            \end{equation}
            for all $\phi_{1,2,3}\in\frB$ for the derived bracket~\eqref{eq:derived_bracket}. The shifted Poisson identity, in turn, implies the shifted Jacobi identity~\eqref{eq:BV_GB_Jacobi}.
        \end{proposition}
        
        \begin{proof}
           By direct computation, cf.~\cref{app:postponed_proofs}. 
        \end{proof}
        
        \begin{proposition}\label{eq:BVBoxSecondOrder}
            For a $\BVbox$-algebra $(\frB,\sfd,\sfm_2,\sfb)$, the operator $\BBox=[\sfd,\sfb]$ is of second order.
        \end{proposition}
        
        \begin{proof}
            This follows from the fact that the (graded) commutator of an $r$-th order differential and an $s$-th order differential is a differential of order $r+s-1$, cf.~\cite[Eq.~(6.iii)]{Akman:1995tm}.
        \end{proof}
        
        \noindent
        Note that by virtue of \cref{prop:BVbox_is_pre_BVbox}, for $(\frB,\sfd,\sfm_2,\sfb)$ a $\BVbox$-algebra, the tuple $(\frB,\sfm_2,\{-,-\})$ with $\{-,-\}$ the derived bracket~\eqref{eq:derived_bracket} is what is commonly known as a \uline{Gerstenhaber algebra}, that is, a Poisson algebra of degree~$-1$. 
        
        \begin{definition}\label{def:BVAlgebra}
            A $\BVbox$-algebra $(\frB,\sfd,\sfm_2,\sfb)$ with $\BBox=[\sfd,\sfb]=0$ is called a \uline{differential graded (dg) Batalin--Vilkovisky (BV) algebra}.
        \end{definition}
        
        We then have the following immediate corollary to \cref{prop:algebraRelationsDB}:
        
        \begin{corollary}\label{cor:kin_alg_is_dg_Lie_for_BV-algebra}
            Consider a BV algebra $\frB$ with differential $\sfd$.
            Together with $\sfd[1]$, the kinematic Lie algebra $\frKin(\frB)$ and the restricted kinematic Lie algebra $\frKin^0(\frB)$ defined in \cref{def:restricted_kinematic_Lie_algebra} become dg Lie algebras.
        \end{corollary}
        
        \paragraph{$\BVbox$-algebra modules.}
        Let us also specialise the notion of modules. Firstly, we define $\BVbox$-modules by refining \cref{def:preBVBoxModule}.
        
        \begin{definition}\label{def:BVbox_module}
            A \uline{module over a $\BVbox$-algebra} $\frB$ is a module $(V,\sfd_V,\acton_V,\sfb_V)$ over $\frB$, regarded as a pseudo-$\BVbox$-algebra, in which $\sfb_V$ is of second order.
        \end{definition}
        
        \noindent
        We now have the analogues of \cref{prop:BVbox_is_pre_BVbox,eq:BVBoxSecondOrder}.
        
        \begin{proposition}
            For $\sfb_V$ of second order, the derived bracket~\eqref{eq:derived_bracket_module} always satisfies~\eqref{eq:derivedBracketModuleJacobi} as well as
            \begin{equation}
                \{\phi_1,\phi_2\acton_Vv\}_V\ =\ \{\phi_1,\phi_2\}_\frB\acton_Vv+(-1)^{(|\phi_1|+1)|\phi_2|}\phi_2\acton_V\{\phi_1,v\}_V
            \end{equation}
            for all $\phi_{1,2}\in\frB$ and $v\in V$.
        \end{proposition}
        
        \begin{proposition}\label{eq:BVBoxModuleSecondOrder}
            For a $\BVbox$-module $(V,\sfd_V,\acton_V,\sfb_V)$, the operator $\BBox_V=[\sfd_V,\sfb_V]$ is of second order.
        \end{proposition}
        
        If we specialise to the situation $\BBox_\frB=\BBox_V=0$, we obtain dg Lie modules.
        
        \begin{proposition}\label{prop:restrictedKinematicLieModule}
            Given a module $V=(V,\sfd_V,\acton,\sfb_V)$ over a dg BV algebra $\frB=(\frB,\sfd,\sfm_2,\sfb)$, then $\frMod^0(V)\coloneqq(\ker\sfb_V)[1]$ is a module over the dg Lie algebra
            $\frKin^0(\frB)$.
        \end{proposition}
        
        \begin{proof}
            By direct calculation, cf.~\cref{app:postponed_proofs}.
        \end{proof}
        
        Finally, we also refine the notions of pseudo-$\BVbox$-algebras and pseudo-$\BVbox$-modules over Hopf algebras introduced in \cref{def:preBVBoxAlgebraHopfAlgebra,def:preBVBoxModuleHopfAlgebra}, see also \cref{def:HopfAlgebra}, as follows.
        
        \begin{definition}\label{def:BVBoxAlgebraHopf}
            A \uline{$\BVbox$-algebra} over a cocommutative Hopf algebra $\frH$ is a tuple $(\frB,\sfd,\sfm_2,\sfb,\acton)$ such that $(\frB,\sfd,\sfm_2,\sfb)$ is a $\BVbox$-algebra, $(\frB,\sfd,\sfm_2,\acton)$ is a dg commutative algebra over $\frH$, the differential $\sfb$ is linear over $\frH$,
            \begin{equation}
                \chi\acton(\sfb\phi)\ =\ \sfb(\chi\acton\phi)
            \end{equation}
            for all $\chi\in\frH$ and $\phi\in\frB$, and we require that $\BBox=[\sfd,\sfb]\in\frH$.
        \end{definition}
        
        \noindent        
        It remains to extend the notion of a $\BVbox$-module to a $\BVbox$-module over $\frH$.
        
        \begin{definition}
            A \uline{$\BVbox$-module} over a cocommutative Hopf algebra $\frH$ is a module $(V,\sfd_V,\acton_V,\sfb_V)$ over a $\BVbox$-algebra $(\frB,\sfd_\frB,\sfm_2,\sfb_\frB,\acton_\frB)$ over $\frH$ such that we have linearity over $\frH$ in the sense of
            \begin{equation}
                \begin{aligned}
                    \chi\acton_V(\sfd_Vv)\ &=\ \sfd_V(\chi\acton_Vv)~,
                    \\
                    \chi\acton_V(\sfb_Vv)\ &=\ \sfb_V(\chi\acton_Vv)~,
                    \\
                    \chi\acton_V\{\phi,v\}_V\ &=\ \{\chi^{(1)}\acton_V\phi,\chi^{(2)}\acton_Vv\}_V
                \end{aligned}
            \end{equation}
            for all $h\in\frH$, $\phi\in\frB$, and $v\in V$. Here, $\{-,-\}_V$ is the derived bracket~\eqref{eq:derived_bracket_module} associated with $(V,\sfd_V,\acton_V,\sfb_V)$, and we have used Sweedler notation~\eqref{eq:Sweedler_notation}. In addition, we require that $\BBox=[\sfd_\frB,\sfb_\frB]=[\sfd_V,\sfb_V]$.
        \end{definition}
        
        \begin{definition}
            A \uline{metric} on a $\BVbox$-module $V$ on a cyclic $\BVbox$-algebra $\frB$ is the same as a metric as a pseudo-$\BVbox$-module. A \uline{cyclic $\BVbox$-module} is a cyclic pseudo-$\BVbox$-module that is a $\BVbox$-module.
        \end{definition}
        
        \paragraph{Comments.}
        Anticipating our upcoming work~\cite{Borsten:2022aa}, we note that the above definitions have a nice operadic formulation, which is crucial for a generalisation to homotopy algebras, providing a generalisation of the present analysis of double copy to theories with interactions beyond cubic terms. Operads are algebraic gadgets that encode the axioms of an algebraic structure. They are formulated inside an ambient setting of symmetric monoidal categories; in the present case, the category is that of cochain complexes of modules over the Hopf algebra $\frH$, with the monoidal operation given by tensor product over $\IR$ (rather than the smaller tensor product over $\frH$). This means that all operations are linear (rather than multi-linear) over $\frH$. Thus, one can construct an operad in the category of cochain complexes of $\frH$-modules such that algebras over this operad are $\BVbox$-algebras over $\frH$. Similarly, one can construct a two-sorted operad over the cochain complexes of $\frH$-modules, with one sort for elements of the $\BVbox$-algebra itself and the other sort for elements of the module; an algebra over this operad is then a $\BVbox$-algebra over $\frH$ together with a $\BVbox$-module over it. 
        
        \subsection{Gauge fixing}\label{ssec:gaugeFixedPreBVBox}
        
        Let us now examine how gauge fixing a BV action of a CK-dual gauge field theory affects the pseudo-$\BVbox$-algebra structure on the colour-stripped dg commutative algebra. We shall focus on ordinary gauge theories; higher gauge theories can be dealt with in a similar fashion.
        
        \paragraph{General gauge-fixing procedure.}
        The traditional gauge-fixing procedure in the BV formalism usually consists of the following three steps~\cite{Batalin:1981jr,Batalin:1984jr}, see also~\cite{Gomis:1994he,Jurco:2018sby} for a detailed review: 
        
        \begin{enumerate}[(i)]\itemsep-4pt
            \item Add trivial pairs of fields to the BV action as needed. For ordinary gauge theories, one such gauge Lie algebra valued pair, consisting of a Nakanishi--Lautrup field and an anti-ghost, is sufficient. For higher gauge theories, one needs a full BV triangle of trivial pairs, cf.~\cite{Batalin:1984jr} and~\cite{Jurco:2018sby} for a review.
            \item Using these fields, define a gauge-fixing fermion $\Psi$, i.e.~a function of ghost degree~$-1$ in the BV fields, which, in turn, defines a symplectomorphism or canonical transformation
            \begin{equation}\label{eq:canonicalTransformation}
                (\phi,\phi^+)\ \mapsto\ (\tilde\phi,\tilde\phi^+)\ \coloneqq\ \left(\phi,\phi^++\delder[\Psi]{\phi}\right)
            \end{equation}
            for the fields $\phi$ and anti-fields $\phi^+$. For simplicity, we always restrict ourselves to the usual quadratic gauge-fixing fermions, for which the canonical transformation becomes a constant rotation.
            \item In most cases of interest to us (ordinary gauge and gauge--matter theories, as well as $\caN=0$ supergravity) the BV action is linear in the anti-fields after this symplectomorphism, and we can simply put to zero all terms containing anti-fields from the gauge-fixed action.
        \end{enumerate}
        
        Even for considering tree-level CK duality, it is helpful (albeit not necessary) to consider the gauge-fixed BV action as the kinematic operator exclusively maps fields to anti-fields.
        
        \paragraph{Step (i): trivial pairs.}
        Consider a cubic gauge field theory with an underlying pseudo-$\BVbox$-algebra $(\frB,\sfd,\sfm_2,\sfb)$. The first step in the gauge-fixing procedure consists of adding trivial pairs which amounts to extending the field space by $V\oplus V[-1]$ for $V$ a graded vector space. 
        
        \begin{proposition}\label{prop:trivialPairs}
            Let $(\frB,\sfd,\sfm_2,\sfb)$ be a pseudo-$\BVbox$-algebra and $V$ a graded vector space with an action of $\BBox$. Then the tuple $(\frB',\sfd',\sfm'_2,\sfb')$ with
            \begin{subequations}
                \begin{equation}
                    \frB'\ \coloneqq\ \frB\oplus V\oplus V[-1]
                \end{equation}
                and\footnote{We use $n_i$ to indicate Nakanishi--Lautrup fields to avoid the notational collision of the usual $b$ with our operator $\sfb$.}
                \begin{equation}
                    \begin{aligned}
                        \sfd'(\phi_1,n_1,\bar c_1)\ &\coloneqq\ (\sfd\phi_1,0,n_1[-1])~,
                        \\
                        \sfm'_2\big((\phi_1,n_1,\bar c_1),(\phi_2,n_2,\bar c_2)\big)\ &\coloneqq\ (\sfm_2(\phi_1,\phi_2),0,0)~,
                        \\
                        \sfb'(\phi_1,n_1,\bar c_1)\ &\coloneqq\ (\sfb\phi_1,(\BBox\bar c_1)[1],0)
                    \end{aligned}
                \end{equation}
            \end{subequations}
            for all $\phi_{1,2}\in\frB$, $n_{1,2}\in V$, and $\bar c_{1,2}\in V[-1]$ is a pseudo-$\BVbox$-algebra with $\BBox'=[\sfd',\sfb']=\BBox$.
        \end{proposition}
        
        \begin{proof}
            It is straightforward to see that $\sfd'^2=0$ and that $\sfd'$ is a derivation for $\sfm'_2$. Thus, $(\frB',\sfm'_2,\sfd')$ is a dg commutative algebra. Likewise, $\sfb'^2=0$ and $[\sfd',\sfb']=\BBox$. In addition, the derived bracket~\eqref{eq:derived_bracket} for $(\frB',\sfd',\sfm'_2,\sfb')$ is
            \begin{equation}
                \{(\phi_1,n_1,\bar c_1),(\phi_2,n_2,\bar c_2)\}'\ =\ (\{\phi_1,\phi_2\},0,0)~,
            \end{equation}
            where $\{-,-\}$ is the derived bracket for $(\frB,\sfd,\sfm_2,\sfb)$. Consequently, the conditions~\eqref{eq:BV_GB1} are also satisfied. Altogether, $(\frB',\sfd',\sfm'_2,\sfb')$ is a pseudo-$\BVbox$-algebra.
        \end{proof}
        
        \paragraph{Step (ii): gauge-fixing fermion and canonical transformation.}
        The second step in the gauge-fixing procedure, namely introducing a gauge-fixing fermion $\Psi$ and performing the canonical transformation~\eqref{eq:canonicalTransformation} preserves the $\BVbox$-algebra structure for the usual quadratic\footnote{Recall that we assume that $\Psi$ is quadratic in all BV fields. This includes the usual gauge-fixing fermions, as e.g.~those for $R_\xi$-gauges.} $\Psi$, as in this case the canonical transformation~\eqref{eq:canonicalTransformation} is merely a constant rotation of all the fields and anti-fields.
        
        We will mostly be interested in the gauge-fixing condition $\sfb A=0$, and we can implement this condition by using the usual gauge-fixing fermion for $R_\xi$-gauges. This leads to an interesting phenomenon. For simplicity and concreteness sake, let us consider the differentials for the pseudo-$\BVbox$-algebra $\frB$ in an ordinary gauge theory on $d$-dimensional Minkowski space $\IM^d$. In degree~$0$ and $1$, we have the following structure before applying the symplectomorphism:
        \begin{equation}
            \begin{tikzcd}[row sep=1cm,column sep=2.7cm]
                \stackrel{c}{\Omega^0(\IM^d)} \arrow[r,shift left,"\rmd"] & \stackrel{A}{\Omega^1(\IM^d)} \arrow[l,shift left,"\rmd^\dagger"] \arrow[r,shift left,"K"] & \stackrel{A^+}{\Omega^1(\IM^d)} \arrow[l,shift left,"K^+"] \arrow[r,shift left,"\rmd^\dagger"]  & \stackrel{c^+}{\Omega^0(\IM^d)} \arrow[l,shift left,"\rmd"]
                \\
                & \stackrel{n}{\Omega^0(\IM^d)} \arrow[rd,shift left,"\sfid", pos=0.1] & \stackrel{n^+}{\Omega^0(\IM^d)} \arrow[ld,shift left,"\BBox", pos=0.1]
                \\
                & \stackrel{\bar c^+}{\Omega^0(\IM^d)} \arrow[ru,shift left,"\sfid", pos=0.1] & \stackrel{\bar c}{\Omega^0(\IM^d)} \arrow[lu,shift left,"\BBox", pos=0.1]
            \end{tikzcd}
        \end{equation}
        where $K$ is a kinematic operator, e.g.~$K=\rmd^\dagger \rmd$ for Yang--Mills theory, or $K=\star\rmd$ for Chern--Simons theory, and $K^+$ is the corresponding operator $\sfb$ so that we have $[\sfd,\sfb]=\BBox$, e.g.~$K^+=\sfid$ for Yang--Mills theory and $K^+=\rmd^\dagger\star$ for Chern--Simons theory, with $\BBox=\wave$ in these two examples. After the symplectomorphism induced by the usual $R_\xi$-gauge-fixing fermion for the gauge $\sfb A=\rmd^\dagger A=0$, we have the following:
        \begin{equation}\label{eq:gauge-fixed-bidirectional-complex}
            \begin{tikzcd}[row sep=1cm,column sep=2.7cm]
                \stackrel{c}{\Omega^0(\IM^d)} \arrow[r,shift left,"\rmd"] \arrow[rdd,shift left,"-\BBox"] & \stackrel{A}{\Omega^1(\IM^d)} \arrow[l,shift left,"\rmd^\dagger"] \arrow[rd,shift right,"-\rmd^\dagger"', pos=0.1] \arrow[r,shift left,"K"] & \stackrel{A^+}{\Omega^1(\IM^d)} \arrow[l,shift left,"K^+"]\arrow[r,shift left,"\rmd^\dagger"]  & \stackrel{c^+}{\Omega^0(\IM^d)} \arrow[l,shift left,"\rmd"]
                \\
                & \stackrel{n}{\Omega^0(\IM^d)} \arrow[rd,shift left,"-\sfid", pos=0.1] \arrow[ru,shift left,"\rmd", pos=0.1] & \stackrel{n^+}{\Omega^0(\IM^d)} \arrow[ld,shift left,"\BBox", pos=0.1]
                \\
                & \stackrel{\bar c^+}{\Omega^0(\IM^d)} \arrow[ru,shift left,"\sfid", pos=0.1] & \stackrel{\bar c}{\Omega^0(\IM^d)} \arrow[lu,shift left,"-\BBox", pos=0.1]\arrow[ruu,shift left,"-\BBox"]
            \end{tikzcd}
        \end{equation}
        
        \paragraph{Step (iii): removing anti-fields.}
        The third and final step is now the most subtle one, as we need to truncate the interaction vertices to a subset to remove the anti-fields. The colour-stripped fields form a subspace $\frF$ of $\frB$ with a natural complement $\frA$ of colour-stripped anti-fields, and we have projectors $\Pi_\frF$ and $\Pi_\frA$ such that
        \begin{equation}
            \sfid_\frB\ =\ \Pi_\frF+\Pi_\frA~,~~~\Pi_\frF^2\ =\ \Pi_\frF~.
        \end{equation}
        Removing the anti-fields from the BV action then changes the dg commutative algebra $(\frB,\sfd,\sfm_2)$ of the colour-stripped action to the dg commutative algebra $(\frB',\sfd',\sfm'_2)$. Because the action contains only fields, the differential and product it encodes can only map fields to anti-fields. Hence,
        \begin{equation}
            \begin{aligned}
                \sfd'\ &\coloneqq\ \Pi_\frA\circ\sfd\circ\Pi_\frF~,
                \\
                \sfm_2'\ &\coloneqq\ \Pi_\frA\circ\sfm_2\circ(\Pi_\frF\otimes\Pi_\frF)
            \end{aligned}
        \end{equation}
        with a potential cyclic structure preserved. This directly extends to modules encoding potential matter fields. 
        
        This projection requires to redefine $\sfb$ to preserve $[\sfd, \sfb]=\BBox$. It is, however, clear that there is a redefinition of $\sfb$ to an operator $\sfb'$ such that all fields are in its kernel and that $[\sfd',\sfb']=\BBox$. In our above example, we have
        \begin{equation}
            \begin{tikzcd}[row sep=1cm,column sep=2.7cm]
                \stackrel{c}{\Omega^0(\IM^d)}  \arrow[rdd,shift left,"-\BBox"] & \stackrel{A}{\Omega^1(\IM^d)} \arrow[rd,shift right,"-\rmd^\dagger"', pos=0.1] \arrow[r,shift left,"K"] & \stackrel{A^+}{\Omega^1(\IM^d)} \arrow[l,shift left,"K^+"]  \arrow[ld,shift left,"\rmd^\dagger", pos=0.1] & \stackrel{c^+}{\Omega^0(\IM^d)} \arrow[ldd,shift left,"-\sfid"]
                \\
                & \stackrel{n}{\Omega^0(\IM^d)} \arrow[ru,shift left,"\rmd", pos=0.1] & \stackrel{n^+}{\Omega^0(\IM^d)}  \arrow[lu,shift right,"-\rmd"', pos=0.1]
                \\
                & \stackrel{\bar c^+}{\Omega^0(\IM^d)}  \arrow[luu,shift left,"-\sfid"]& \stackrel{\bar c}{\Omega^0(\IM^d)} \arrow[ruu,shift left,"-\BBox"]
            \end{tikzcd}
        \end{equation}
        We see that the anti-field of the Nakanishi--Lautrup field takes over the role of the ghost, and this is a generic feature of gauge fixing to $\sfb A=0$. It is therefore clear that a redefinition $\sfb\rightarrow \sfb'$ with $[\sfd',\sfb']=\BBox$ always exists. 
        
        Moreover, it follows that the image of $\sfb'$ is now fully contained in the subspace of fields $\frF\subseteq\frB$:
        \begin{subequations}\label{eq:ker-im-constraints}
            \begin{equation}\label{eq:ker-im-constraints-b}
                \im(\sfb')\ \subseteq\ \frF\ \subseteq\ \ker(\sfb')~.
            \end{equation}
            Analogously, we have for the anti-fields $\frA\subseteq\frB$:
            \begin{equation}
                \im(\sfd')\ \subseteq\ \frA\ \subseteq\ \ker(\sfd')~.
            \end{equation}
        \end{subequations}        
        This generalises to arbitrary gauge theories as well as abelian higher gauge theories, such as $\caN=0$ supergravity.
        
        In all cases of interest, it turns out that gauge fixing in this manner ensures that $\sfb$ is of second order, so that we arrive at a gauge-fixed $\BVbox$-algebra $(\frB,\sfd',\sfm_2',\sfb')$. This observation directly extends to $\BVbox$-modules.
        \begin{definition}\label{def:gaugedFixedBVBoxAlgebra}
            A \uline{gauge-fixed $\BVbox$-algebra} is a $\BVbox$-algebra $\frB$ together with a decomposition $\frB=\frF\oplus \frA$ as graded vector spaces into field and anti-field spaces such that~\eqref{eq:ker-im-constraints} are satisfied.
        \end{definition}
        
        \subsection{Koszul hierarchy: kinematic \texorpdfstring{$L_\infty$}{L-infinity}-algebras}\label{ssec:kinematic_L_infty_algebras}
        
        Let us briefly give an outlook on our forthcoming paper~\cite{Borsten:2022aa}, in which we shall discuss the homotopy generalisation of the picture presented here. That is, the algebras with unary and binary operations (i.e.~differentials and binary products) appearing in our discussion will be replaced by operations with arbitrary arity. But we can encounter such homotopy algebras already here.
        
        Derived brackets of the type~\eqref{eq:derived_bracket} are reminiscent of other derived bracket constructions, cf.~\cite{Getzler:1010.5859,Deser:2018oyg,Borsten:2021ljb}, which naturally produce higher brackets of arbitrary arity. A similar phenomenon can be observed here. Consider a theory with colour-stripped dg commutative algebra $(\frB,\sfd,\sfm_2)$ together with a nilpotent operator $\sfb$ of degree~$-1$ which gives rise to the colour-stripped propagator $\frac{\sfb}{[\sfb,\sfd]}$. While the derived bracket $\{-,-\}$ given in~\eqref{eq:derived_bracket}, which is the operator $\Phi^2_\sfb$ in $\frB$ as defined in~\eqref{eq:def_Phir}, is no longer a Lie bracket, one finds that the Jacobi identity is violated only up to homotopy. Generally, we have the following result.
        
        \begin{proposition}[{\cite[Section 2.5]{Bering:1996kw}}]\label{prop:derived_brackets_form_L_infty_algebra}
            Given a graded commutative algebra $\frA$ with a differential $\delta$ of degree~$-1$, the operations $\Phi_\delta^r$ defined in~\eqref{eq:def_Phir} form the grade-shifted higher products of an $L_\infty$-algebra. This $L_\infty$-algebra is known as the \uline{Koszul hierarchy}. It is quasi-isomorphic to the cochain complex defined by $\delta$.
        \end{proposition}    
        
        \noindent
        For examples, see also~\cite{Bering:1996kw,Bering:2006eb}. 
        
        Another important observation was made in~\cite{Markl:2013pca}, where the Koszul hierarchy was interpreted as a twisting of a cochain complex by a specific twist, see also~\cite{Markl:2014kna}. This observation not only gives a surprisingly simple proof of the above proposition but also provides for new examples of hierarchies of higher brackets, called higher braces there. Such braces are referred to as natural ones if they use only the data that are available for any graded associative commutative algebra with a differential $\delta$. As such, they could possibly also be relevant as kinematic $L_\infty$-algebras.
        
        The Koszul hierarchy is singled out by the requirement that the binary bracket measures the failure of $\delta$ being first order, the coefficient of $\delta(\sfm_2(\sfm_2(\phi_1,\phi_2),\phi_3))$ in  $\Phi^3_\delta(\phi_1,\phi_2,\phi_3)$ in~\eqref{eq:def_Phir} is $\pm 1$ and that  $\Phi^k_\delta=0$ implies $\Phi^{k+1}_\delta=0$ (hereditarity), cf.~\cite{Markl:2013pca}.
        
        Hence, we can define pre-$\BVbox$-algebras.
        
        \begin{definition}\label{def:kinematic_infty_algebra}
            A pre-$\BVbox$-algebra is a dg commutative algebra $(\frB,\sfd,\sfm_2)$ together with a differential $\sfb$ of degree~$-1$. The \uline{kinematic $L_\infty$-algebra} of a pre-$\BVbox$-algebra is the (shifted) $L_\infty$-algebra given by the Koszul hierarchy.
        \end{definition}
        \noindent Recall, however, from above, that there are, in fact, a number of possibly relevant kinematic $L_\infty$-algebras (which are, again, isomorphic to the Koszul hierarchy).
        
        We then have the following immediate specialisations of our above notions:
        
        \begin{corollary}
            A pre-$\BVbox$-algebra in which the higher product $\mu_2$ in the kinematic $L_\infty$-algebra satisfies the shifted Jacobi identity~\eqref{eq:BV_GB_Jacobi} is a pseudo-$\BVbox$-algebra. A pre-$\BVbox$-algebra with strict kinematic $L_\infty$-algebra is a $\BVbox$-algebra.
        \end{corollary}
        
        \noindent
        Recall from~\eqref{eq:def_Phir} that $\sfb$ being of second order is tantamount to $\Phi^3_\delta(-,-,-)$ being trivial.
        
        There are two important points to note. First of all, while $L_\infty$-algebras can always be strictified, there is no reason to believe that any pre-$\BVbox$-algebra is quasi-isomorphic to a pseudo-$\BVbox$-algebra. Hence, we cannot expect all field theories to have an underlying pseudo-$\BVbox$-algebra, or, equivalently, exhibit CK duality. So, although every theory has a kinematic $L_\infty$-algebra associated to it, this does not imply every theory has CK duality; this kinematic algebra should not be regarded as that of some CK-duality respecting numerators nor an off-shell CK-duality manifesting action.

        Secondly, it may be surprising that such a physically evidently non-trivial datum as the kinematic Lie algebra extends to a dg Lie algebra which is quasi-isomorphic to an ordinary cochain complex. Again, however, we have to note that this quasi-isomorphism does not amount to a physical equivalence, which would be captured by a quasi-isomorphism of the underlying pseudo-$\BVbox$-algebras. Moreover, we note that for most interesting field theories, the operator $\sfb$ has trivial cohomology, and hence the kinematic $L_\infty$-algebra is quasi-isomorphic to the trivial one.
        
        We plan to investigate the deeper implications of kinematic $L_\infty$-algebras in future work~\cite{Borsten:2022aa}.
        
        \paragraph{Pseudo-$\BVbox$-algebras vs $\BVbox$-algebras.} 
        Let us close this section on CK duality with a comment on the difference between pseudo-$\BVbox$-algebras and $\BVbox$-algebras. As we saw, a pseudo-$\BVbox$-algebra (and a module over it) is the minimal requirement for having a kinematic Lie algebra manifested on the Feynman diagram expansion of the currents of a field theory. We can now conclude that the restriction to $\BVbox$-algebras is certainly natural from a mathematical perspective: the fact that the operator $\sfb$ in the data of a pseudo-$\BVbox$-algebra is of second order is equivalent to the Poisson identity by \cref{prop:BVbox_is_pre_BVbox}, which, in turn, is equivalent to the Koszul hierarchy being a dg Lie algebra.
        
        From a physics perspective, it is natural to ask for the kinematic Lie algebra to lift uniquely to arbitrary local operators constructed by multiplying the fields in the theory. This unique lift is provided by the additional Poisson identity.
        
        \section{Double copy and syngamies for special \texorpdfstring{$\BVbox$}{BV-box}-algebras}\label{sec:doubleCopySyngamies}
        
        In this section, we shall explain how two $\BVbox$-algebras of field theories can be combined into a syngamy. The double copy of gauge theories to supergravity theories is a special case of this construction.
        
        \paragraph{Outline.}
        Given our discussion of CK duality, we are led to looking for an interpretation of the double copy in terms of $\BVbox$-algebras\footnote{In the rest of the paper, we will focus on $\BVbox$-algebras and comment here and there on the problems of generalising the picture to pseudo-$\BVbox$-algebras.}, and the obvious starting point is the tensor product of two $\BVbox$-algebras~\cite{Borsten:2021gyl,Macrelli:2022bay,Macrelli2022sm,Bonezzi:2022bse}. As we will see below, this tensor product exists, extending the tensor product of two dg commutative algebras. 
        
        This direct tensor product, however, does not match our expectations. To see this, let us sketch the simple example of biadjoint scalar field theory, which is fully developed in~\cref{ssec:biadjoint_scalar_field_theory}. The $\BVbox$-algebra of this theory has an underlying cochain complex $\sfCh(\frB)$ which is concentrated in degrees $1$ and $2$, 
        \begin{equation}
            \sfCh(\frB)\ =\ \left(\frB_1\xrightarrow{~\sfd~}\frB_2\right)\ =\ \left(\frg\otimes \scC^\infty(\IM^d)\xrightarrow{~\sfid_\frg\otimes \wave~}\frg\otimes \scC^\infty(\IM^d)\right)
        \end{equation}
        for $\frg$ some Lie algebra. The kinematic Lie algebra is simply the Lie algebra $\frg$, and the double copy of $\frB$ with itself is expected to yield biadjoint scalar field theory with fields taking values in $\frg\otimes \frg\otimes \scC^\infty(\IM^d)$. 
        
        The tensor product $\sfCh(\frB)\otimes \sfCh(\frB)$, however, is given by the cochain complex
        \begin{equation}
            \left(\left(\frg\otimes \scC^\infty(\IM^d)\right)^{\otimes 2}\xrightarrow{~\hat \sfd~}\IR^2\otimes \left(\frg\otimes \scC^\infty(\IM^d)\right)^{\otimes 2}\xrightarrow{~\hat \sfd~}\left(\frg\otimes \scC^\infty(\IM^d)\right)^{\otimes 2}\right)
        \end{equation}
        concentrated in degrees $2$, $3$, and $4$, which has several problems. First of all, there are no BV fields, as all elements of degree~$1$ are trivial. We will show that this problem can be solved by switching from the tensor product $\BVbox$-algebra $\hat \frB$ to its kinematic Lie algebra $\hat \frK\coloneqq\frKin(\hat\frB)$, which involves a degree shift. After this, we end up with a cochain complex concentrated in degrees $1$, $2$, and $3$. The field space, $\hat\frK_1=\left(\frg\otimes \scC^\infty(\IM^d)\right)^{\otimes 2}$, however, is still larger than the expectation $\frg\otimes \frg\otimes \scC^\infty(\IM^d)$. This issue can be addressed by considering $\BVbox$-algebra over the Hopf algebra $\frH_{\IM^d}$ of \cref{ex:H-Box-Minkowski}, which is generated by the differential operators on space-time $\IM^d$ with constant coefficient and hence allows us to control momentum dependence. As shown in~\cref{app:restricted_tensor_product}, there is a natural notion of restricted tensor product of $\BVbox$-algebras which are modules over a restrictedly tensorable cocommutative Hopf algebra such as $\frH_{\IM^d}$, which here amounts to restricting the tensor product to the kernel of the operators
        \begin{equation}
            \parder{x^\mu}\otimes \sfid-\sfid\otimes \parder{x^\mu}
        \end{equation}
        with $x^\mu$ the Cartesian coordinates on $\IM^d$. As a result the tensor product $\scC^\infty(\IM^d)\otimes \scC^\infty(\IM^d)$ is reduced to $\scC^\infty(\IM^d)$, and we have a new, reduced kinematic Lie algebra $\tilde \frK$. Even after this reduction, however, the homogeneous subspaces of $\tilde \frK$ of degrees $2$ and $3$ are still too large and require further reduction. In fact, we note that the space of BV fields is, in a sense, double its expected size.\footnote{A similar problem arises in the pure spinor formulation of supergravity, see the comments in \cref{ssec:pure_spinors_SYM}.} Moreover, we note that for the biadjoint scalar field theory, $\tilde \frK$ is split in half as
        \begin{equation}
            \tilde \frK\ =\ \coker(\sfb\otimes \sfid-\sfid\otimes \sfb)\oplus \ker(\sfb\otimes \sfid-\sfid\otimes \sfb)~,
        \end{equation}
        and $\ker(\sfb\otimes \sfid-\sfid\otimes \sfb)$ is naturally a dg Lie subalgebra. Hence, we restrict further to the kernel $\sfb\otimes \sfid-\sfid\otimes \sfb$, and the resulting dg Lie algebra turns out to be the expected one for biadjoint scalar field theory. This restriction should be seen analogously to the section condition in double field theory (albeit we only double the functions, not the dimensions of the tensors). In this sense, the double copy is closely related to double field theory, cf.~also~\cite{Bonezzi:2022bse}. After a first version of this paper was finished, we became aware of the paper~\cite{Zeitlin:2014xma}, in which essentially the same restriction was used in the context of a special form of $\caN=0$ supergravity on Hermitian manifolds.
        
        In the following, we will develop the construction sketched above in detail.
        
        \subsection{Tensor products of \texorpdfstring{$\BVbox$}{BV-box}-algebras}
        
        \paragraph{Ordinary tensor product.}
        Recall that $\BVbox$-algebras as defined in \cref{def:BVBoxAlgebra} are dg commutative algebras endowed with an additional operation $\sfb$. The tensor product of two dg commutative algebras $\frC_\rmL=(\frC_\rmL,\sfd_\rmL,\sfm_{2\rmL})$ and $\frC_\rmR=(\frC_\rmR,\sfd_\rmR,\sfm_{2\rmR})$ is another dg commutative algebra $\hat\frC=(\hat\frC,\hat\sfd,\hat\sfm_2)$ with $\hat \frC=\frC_\rmL\otimes\frC_\rmR$ and the differential and product defined by 
        \begin{subequations}\label{eq:ordinary_tensor_product}
            \begin{equation}
                \begin{aligned}
                    \hat \sfd(\phi_{1\rmL}\otimes\phi_{1\rmR})\ &\coloneqq\ \sfd_\rmL\phi_{1\rmL}\otimes\phi_{1\rmR}+(-1)^{|\phi_{1\rmL}|}\phi_{1\rmL}\otimes\sfd_\rmR\phi_{1\rmR}~,
                    \\
                    \hat \sfm_2(\phi_{1\rmL}\otimes\phi_{1\rmR},\phi_{2\rmL}\otimes\phi_{2\rmR})\ &\coloneqq\ (-1)^{|\phi_{1\rmR}||\phi_{2\rmL}|}\sfm_{2\rmL}(\phi_{1\rmL},\phi_{2\rmL})\otimes\sfm_{2\rmR}(\phi_{1\rmR},\phi_{2\rmR})~.
                \end{aligned}
            \end{equation}
            If both $\frC_\rmL$ and $\frC_\rmR$ are endowed with metric $\inner{-}{-}_\rmL$ and $\inner{-}{-}_\rmR$ of degrees $n_\rmL$ and $n_\rmR$, respectively, then the tensor product $\hat \frC$ is endowed with a metric of degree~$n_\rmL+n_\rmR$ given by\footnote{For a proof of the cyclicity of this tensor product, see \cref{app:postponed_proofs}.} 
            \begin{equation}
                \inner{\phi_{1\rmL}\otimes\phi_{1\rmR}}{\phi_{2\rmL}\otimes\phi_{2\rmR}}\ \coloneqq\ 
                (-1)^{|\phi_{1\rmR}||\phi_{2\rmL}|+n_\rmR(|\phi_{1\rmL}|+|\phi_{2\rmL}|)}\inner{\phi_{1\rmL}}{\phi_{2\rmL}}_\rmL\inner{\phi_{1\rmR}}{\phi_{2\rmR}}_\rmR~.
            \end{equation}
            
            This definition extends to a tensor product of two (metric) $\BVbox$-algebras $\frB_\rmL$ and $\frB_\rmR$.
            \begin{definition}
                The tensor product of two $\BVbox$-algebras $\frB_\rmL$ and $\frB_\rmR$ is the $\BVbox$-algebra $\hat \frB$, whose underlying dg commutative algebra is the tensor product of $\frB_\rmL$ and $\frB_\rmR$, both regarded as dg commutative algebras, and whose operator $\hat \sfb$ is defined as  
                \begin{equation}\label{eq:tensor_product_b_ordinary}
                    \hat \sfb(\phi_{\rmL}\otimes\phi_{\rmR})\ \coloneqq\ \hat \sfb_+(\phi_{\rmL}\otimes\phi_{\rmR})\ \coloneqq\ \sfb_\rmL\phi_{\rmL}\otimes\phi_{\rmR}+(-1)^{|\phi_{\rmL}|}\phi_{\rmL}\otimes\sfb_\rmR\phi_{\rmR}
                \end{equation}
                for all $\phi_{\rmL}\otimes\phi_{\rmR}\in \hat \frB$. Correspondingly, we have a natural definition of $\BBox$ on the tensor product,
                \begin{equation}
                    \hat \BBox\ \coloneqq\ [\hat \sfd,\hat \sfb]\ =\ \BBox_\rmL\otimes \sfid+\sfid\otimes \BBox_\rmR~.
                \end{equation}
            \end{definition}
            
            We will be particularly interested in the special case that both $\frB_\rmL$ and $\frB_\rmR$ are $\BVbox$-algebras over a Hopf algebra $\frH$ with $\BBox_\rmL=\BBox_\rmR=\BBox\in\frH$. In this case, $\hat \frB=\frB_\rmL\otimes \frB_\rmR$ is also a module over $\frH$ with 
            \begin{equation}
                \chi\acton(\phi_{\rmL}\otimes\phi_{\rmR})\ \coloneqq\ (\chi\acton \phi_{\rmL})\otimes \phi_{\rmR}+\phi_{\rmL}\otimes (\chi\acton\phi_{\rmR})
            \end{equation}
            for all $\chi\in \frH$ and $\phi_{\rmL}\otimes\phi_{\rmR}\in \hat \frB$.
        \end{subequations}        
        
        \paragraph{Restricted tensor product.}
        As explained above, the ordinary tensor product is not directly suitable for an interpretation of the double copy, and we have to use the restricted tensor product introduced in \cref{app:restricted_tensor_product}. 
        
        \begin{proposition}\label{prop:tensorProductBVBoxAlgebrasOverH}
            Let $\frH$ be a restrictedly tensorable cocommutative Hopf algebra. Given two $\BVbox$-algebras $\frB_\rmL=(\frB_\rmL,\sfd_\rmL,\sfm_{2\rmL},\sfb_\rmL)$ and $\frB_\rmR=(\frB_\rmR,\sfd_\rmR,\sfm_{2\rmR},\sfb_\rmR)$ over $\frH$ with $\BBox_\rmL=\BBox_\rmR=\BBox\in\frH$, the tuple $(\hat \frB,\hat \sfd,\hat \sfm_2,\hat \sfb)$ with 
            \begin{subequations}\label{eq:tensorProductBVBoxAlgebrasOverH}
                \begin{equation}
                    \hat \frB\ \coloneqq\ \frB_\rmL\otimes^\frH\frB_\rmR\ \coloneqq\ \bigcap_{\chi\in\frH}\ker\big((\chi\otimes\unit-\unit\otimes\chi)\acton\big)\ \subseteq\ \frB_\rmL\otimes\frB_\rmR
                \end{equation}
                and\footnote{Note the sign flip between the two summands in $\hat\sfb$ relative to~\eqref{eq:tensor_product_b_ordinary}, which will turn out to be convenient.}
                \begin{equation}\label{eq:def_red_tensor_ops}
                    \begin{aligned}
                        \hat \sfd(\phi_{1\rmL}\otimes^\frH\phi_{1\rmR})\ &\coloneqq\ \sfd_\rmL\phi_{1\rmL}\otimes\phi_{1\rmR}+(-1)^{|\phi_{1\rmL}|}\phi_{1\rmL}\otimes\sfd_\rmR\phi_{1\rmR}~,
                        \\
                        \hat \sfm_2(\phi_{1\rmL}\otimes^\frH\phi_{1\rmR},\phi_{2\rmL}\otimes^\frH\phi_{2\rmR})\ &\coloneqq\ (-1)^{|\phi_{1\rmR}||\phi_{2\rmL}|}\sfm_{2\rmL}(\phi_{1\rmL},\phi_{2\rmL})\otimes\sfm_{2\rmR}(\phi_{1\rmR},\phi_{2\rmR})~,
                        \\
                        \hat \sfb_-(\phi_{1\rmL}\otimes^\frH\phi_{1\rmR})\ &\coloneqq\ \sfb_\rmL\phi_{1\rmL}\otimes\phi_{1\rmR}-(-1)^{|\phi_{1\rmL}|}\phi_{1\rmL}\otimes\sfb_\rmR\phi_{1\rmR}
                    \end{aligned}
                \end{equation}
                for all $\phi_{1\rmL,2\rmL}\in\frB_\rmL$ and $\phi_{1\rmR,2\rmR}\in\frB_\rmR$ forms a dg BV algebra; in particular, $\hat \sfb_-$ is of second order with respect to $\hat \sfm_2$.
                
                If both $\frB_\rmL$ and $\frB_\rmR$ come with $\frH$-linear metrics $\inner{-}{-}_{\rmL}$ and $\inner{-}{-}_{\rmR}$ of degrees $n_\rmL$ and $n_\rmR$, respectively, then
                \begin{equation}\label{prop:tensorProductMetric}
                    \inner{\phi_{1\rmL}\otimes^\frH\phi_{1\rmR}}{\phi_{2\rmL}\otimes^\frH\phi_{2\rmR}}\ \coloneqq\ 
                    (-1)^{|\phi_{1\rmR}||\phi_{2\rmL}|+n_\rmR(|\phi_{1\rmL}|+|\phi_{2\rmL}|)}\inner{\phi_{1\rmL}}{\phi_{2\rmL}}_{\rmL}\inner{\phi_{1\rmR}}{\phi_{2\rmR}}_{\rmR}
                \end{equation}
            \end{subequations}
            defines a metric for $(\hat \frB,\hat \sfd,\hat \sfm_2,\hat \sfb_-)$ of degree~$n_\rmL+n_\rmR$ for all $\phi_{1\rmL,2\rmL}\in\frB_\rmL$ and $\phi_{1\rmR,2\rmR}\in\frB_\rmR$.            
        \end{proposition}
        
        \begin{proof}
            From the discussion in \cref{app:restricted_tensor_product}, it is clear that $(\hat \frB,\hat \sfd,\hat \sfm_2)$ forms a dg commutative algebra, and that $\hat \sfb_-$ is a differential of degree~$-1$. Furthermore, we have $\chi\acton(\phi_\rmL\otimes^\frH\phi_\rmR)=(\chi\acton\phi_\rmL)\otimes^\frH\phi_\rmR=\phi_\rmL\otimes^\frH(\chi\acton\phi_\rmR)$ for all $\chi\in\frH$, $\phi_\rmL\in\frB_\rmL$, and $\phi_\rmR\in\frB_\rmR$. Consequently, 
            \begin{equation}
                [\hat \sfd,\hat \sfb_-](\phi_\rmL\otimes^\frH\phi_\rmR)\ =\ (\BBox\phi_\rmL)\otimes^\frH\phi_\rmR-\phi_\rmL\otimes^\frH(\BBox\phi_\rmR)\ =\ 0~,
            \end{equation}
            because of the assumption $\BBox_\rmL=\BBox_\rmR=\BBox\in\frH$. In addition, the derived bracket~\eqref{eq:derived_bracket} now becomes
            \begin{equation}\label{eq:derivedBracketTensorProduct}
                \begin{aligned}
                    \{\phi_{1\rmL}\otimes^\frH\phi_{1\rmR},\phi_{2\rmL}\otimes^\frH\phi_{\rmR 2}\}\ &=\ (-1)^{|\phi_{1\rmR}||\phi_{2\rmL}|}\{\phi_{1\rmL},\phi_{2\rmL}\}_\rmL\otimes \sfm_{2\rmR}(\phi_{1\rmR},\phi_{2\rmR})
                    \\
                    &\kern1cm-(-1)^{|\phi_{1\rmR}||\phi_{2\rmL}|+|\phi_{1\rmL}|+|\phi_{2\rmL}|}\sfm_{2\rmL}(\phi_{1\rmL},\phi_{2\rmL})\otimes\{\phi_{1\rmR},\phi_{2\rmR}\}_\rmR
                \end{aligned}
            \end{equation}
            for all $\phi_{1\rmL,2\rmL}\in\frB_\rmL$ and $\phi_{1\rmR,2\rmR}\in\frB_\rmR$, and closure on $\hat \frB$ follows from closure of the defining operations on $\hat \frB$. It remains to show that $\hat \sfb_-$ is of second order, which is equivalent to the shifted Poisson identity~\eqref{eq:BV_GB_Poisson}, as we saw in \cref{prop:BVbox_is_pre_BVbox}. Using the Poissonator defined in~\eqref{eq:PoissonatorJacobiator}, some lengthy but straightforward calculation similar to the derivation~\eqref{eq:derivationPoissonImpliesJacobi} shows that
            \begin{equation}
                \begin{aligned}
                    &\widehat{\sfPoiss}(\phi_{1\rmL}\otimes^\frH\phi_{1\rmR}\,,\,\phi_{2\rmL}\otimes^\frH\phi_{2\rmR}\,,\,\phi_{3\rmL}\otimes^\frH\phi_{3\rmR})
                    \\
                    &\kern.5cm=\ (-1)^{|\phi_{2\rmR}||\phi_{3\rmL}|+|\phi_{1\rmR}|(|\phi_{2\rmL}|+|\phi_{3\rmL}|)}\big[\sfPoiss_\rmL(\phi_{1\rmL},\phi_{2\rmL},\phi_{3\rmL})\otimes\sfm_{2\rmR}(\phi_{1\rmR},\sfm_{2\rmR}(\phi_{2\rmR},\phi_{3\rmR}))
                    \\
                    &\kern1.5cm-(-1)^{|\phi_{1\rmL}||\phi_{2\rmL}|+|\phi_{3\rmL}|}\sfm_{2\rmL}(\phi_{1\rmL},\sfm_{2\rmL}(\phi_{2\rmL},\phi_{3\rmL}))\otimes\sfPoiss_\rmR(\phi_{1\rmR},\phi_{2\rmR},\phi_{3\rmR})\big]
                \end{aligned}
            \end{equation}
            for all $\phi_{1\rmL,2\rmL,3\rmL}\in\frB_\rmL$ and $\phi_{1\rmR,2\rmR,3\rmR}\in\frB_\rmR$. Hence, the shifted Poisson identities for $(\frB_\rmL,\sfd_\rmL,\sfm_{2\rmL},\sfb_\rmL)$ and $(\frB_\rmR,\sfd_\rmR,\sfm_{2\rmR},\sfb_\rmR)$ imply that of $(\hat\frB,\hat \sfd,\hat \sfm_2,\hat \sfb_-)$. The properties of the metric follow by restriction of those on the ordinary tensor product, see also \cref{app:postponed_proofs}.
        \end{proof}
        
        \begin{remark}
            When $(\frB_\rmL,\sfd_\rmL,\sfm_{2\rmL},\sfb_\rmL)$ and $(\frB_\rmR,\sfd_\rmR,\sfm_{2\rmR},\sfb_\rmR)$ are mere pseudo-$\BVbox$-algebras, see \cref{def:pre-BVbox-algebra}, then the tuple $(\hat\frB,\hat\sfd,\hat\sfm_2,\hat\sfb_-)$ defined in~\eqref{eq:tensorProductBVBoxAlgebrasOverH} is generally not even a pseudo-$\BVbox$-algebra. This is because the Jacobiator~\eqref{eq:PoissonatorJacobiator} for the derived bracket~\eqref{eq:derivedBracketTensorProduct} does not only involve the Jacobiators for the derived brackets of $(\frB_\rmL,\sfd_\rmL,\sfm_{2\rmL},\sfb_\rmL)$ and $(\frB_\rmR,\sfd_\rmR,\sfm_{2\rmR},\sfb_\rmR)$ but also their Poissonators, defined in~\eqref{eq:PoissonatorJacobiator}. In all physical applications, however, we are dealing exclusively with (gauge-fixed) $\BVbox$-algebras.
        \end{remark}
        
        \begin{remark}
            Consider $\BVbox$-algebras $\frB_\rmL$ and $\frB_\rmR$, which are modules over the Hopf algebra $\frH_{\IM^d}$ defined in \cref{ex:H-Box-Minkowski} and whose homogeneously graded vector spaces are rings over $\scC^\infty(\IM^d)$ and hence fields over space-time $\IM^d$. The above construction of the restricted tensor product will ensure that the homogeneously graded subspaces of $\hat \frB$ are still fields over $\IM^d$, instead of fields over $\IM^d\otimes \IM^d$, as would be the case for the ordinary tensor product.
        \end{remark}
        
        \subsection{Syngamies of pure gauge theories}\label{ssec:syngamy_pure_gauge}
        
        \paragraph{Syngamies.}
        Let us now come to the construction of syngamies, i.e.~the construction of a field theory from two $\BVbox$-algebras. The usual double copy constructions and its variants will turn out to be special cases of this construction. We start with the construction for pure gauge theories, such as pure Yang--Mills or Chern--Simons theory, and theories with a flavour Lie algebra, such as the biadjoint scalar theory; theories with matter, i.e.~fields in general representation of a gauge or flavour Lie algebra, will be discussed in \cref{sec:syngamiesMatterTheories}.
        
        Even after taking the restricted tensor product $\frB_\rmL\otimes^\frH\frB_\rmR$ of two $\BVbox$-algebras $(\frB_\rmL,\sfd_\rmL,\sfm_{2\rmL},\sfb_\rmL)$ and $(\frB_\rmR,\sfd_\rmR,\sfm_{2\rmR},\sfb_\rmR)$ underlying two field theories, we still end up with a BV field space that is twice the expected size. Concretely, each of the factors $\frB_\rmL$ and $\frB_\rmR$ contains subspaces for fields and anti-fields, and hence the tensor product contains the subspaces 
        \begin{equation}\label{eq:quadruple_field_content}
            \text{fields}\otimes\text{fields}~,~~~
            \text{fields}\otimes\text{anti-fields}~,~~~
            \text{anti-fields}\otimes\text{fields}~,~~~
            \text{anti-fields}\otimes\text{anti-fields}~,
        \end{equation}
        which is twice the expected field content of a syngamy.\footnote{Again we note that the same problem arises in the pure spinor formulation of supergravity, see the comments in \cref{ssec:pure_spinors_SYM}.} We therefore have to restrict to the correct subspace, and a convenient choice in the case of gauge-fixed $\BVbox$-algebras is the restriction to 
        \begin{equation}
            \ker(\hat \sfb_-)\ = \ \ker(\sfb_\rmL\otimes \sfid-\sfid\otimes \sfb_\rmR)
        \end{equation}
        with $\hat \sfb_-$ defined in~\eqref{eq:def_red_tensor_ops}, 
        which naturally extends the restriction from the ordinary tensor product to the restricted tensor product over $\frH$. Recall from \cref{def:gaugedFixedBVBoxAlgebra} that for gauge-fixed $\BVbox$-algebras, the kernel of $\sfb$ contains the field space.\footnote{In many examples, $\ker(\sfb)=\frF=\im(\sfb)$.} Considering the kernel of $\sfb$ means to work with a slightly enlarged BV field space $\frF'=\ker(\sfb)\supseteq \frF$, which will turn out to be harmless in all relevant examples. Denoting the cokernel by $\frA'$, the kernel of $\hat\sfb_-$ will consist of the space $\frF'_\rmL\otimes^\frH\frF'_\rmR$ as well as elements of $\frF'_\rmL\otimes^\frH \frA'_\rmR~\oplus~\frA'_\rmL\otimes^\frH\frF'_\rmR$ that are symmetrised such that $\hat \sfb_-$ annihilates them.
        
        The BV algebra structure on $\ker(\hat\sfb_-)$ yields a (metric) dg Lie algebra, which defines the syngamy field theory.
        
        \begin{definition}
            Let $\frH$ be a restrictedly tensorable cocommutative Hopf algebra. Furthermore, let $\frB_\rmL=(\frB_\rmL,\sfd_\rmL,\sfm_{2\rmL},\sfb_\rmL)$ and $\frB_\rmR=(\frB_\rmR,\sfd_\rmR,\sfm_{2\rmR},\sfb_\rmR)$ be two gauge-fixed $\BVbox$-algebras over $\frH$ with $\BBox_\rmL=\BBox_\rmR=\BBox\in\frH$ and let $\hat \frB=(\hat \frB,\hat \sfd,\hat \sfm_2,\hat \sfb_-)$ be the restricted tensor product over $\frH$ as defined in~\eqref{eq:tensorProductBVBoxAlgebrasOverH}. The \uline{syngamy of $\frB_\rmL$ and $\frB_\rmR$} is the restricted kinematic dg Lie algebra $\frKin^0(\hat \frB)$ of \cref{cor:kin_alg_is_dg_Lie_for_BV-algebra}.
        \end{definition}
        
        \paragraph{Inner product.}
        As we shall show now, $\frK^0$ can naturally be endowed with a metric of degree~$-3$, which is necessary for the definition of the action. The following construction may seem a bit abstract and not particularly well-motivated. Nevertheless, it will be the one reproducing all expected features when we will look at concrete examples in \cref{sec:examples}.\footnote{
            There is an analytical subtlety here, which we will largely gloss over outside of this footnote. If $\inner--_{\frB_\rmL}$ and $\inner--_{\frB_\rmR}$ are finite and well-defined, then of course so is the tensor product $\inner--_{\hat\frB}$ and, thus, $\inner--_{\frK^0}$ defined with respect to it. Now, for analytic purposes it is convenient to have $\frB_\rmL$ and $\frB_\rmR$ be nuclear topological vector spaces, e.g.\ spaces of smooth functions $\scC^\infty(\IM^d)$ with the usual Fréchet topology, such that the topological tensor product behaves well. In that case, however, the naïve inner product $\inner fg=\int_{\IM^d}fg$ fails to be finite for general smooth $f$ and $g$, and if one double-copies this, the restricted tensor product will consist of functions that are translation-invariant along $d$ directions in $\IM^d\times\IM^d$, which means that $\inner--_{\hat\frB}$ and therefore $\inner--_{\frK^0}$ have an `infinite volume factor' $\operatorname{vol}(\IM^d)$ that must be cancelled or otherwise regulated away. Working naïvely, one thus runs into such harmless but annoying infinite factors, which are an artefact of the lack of infrared regulators. For more sophisticated approaches to infrared regulators, see~e.g.~\cite{Macrelli:2019afx,Dutsch:2019wxk}.
        }
        
        Firstly, in view of the tensor product~\eqref{eq:tensorProductBVBoxAlgebrasOverH}, let us write
        \begin{equation}\label{eq:dpm}
            \hat \sfd_\pm(\phi_\rmL\otimes^\frH\phi_\rmR)\ \coloneqq\  \sfd_\rmL\phi_\rmL\otimes\phi_\rmR\pm(-1)^{|\phi_\rmL|}\phi_\rmL\otimes\sfd_\rmR\phi_\rmR    
        \end{equation}
        for all $\phi_{\rmL,\rmR}\in\frB_{\rmL,\rmR}$. Evidently, $\hat \sfd_+=\hat \sfd$. It is then easy to check that
        \begin{subequations}\label{eq:propertiesDpm}
            \begin{equation}
                [\hat \sfd_\pm,\hat \sfb_-]\ =\ (1\mp1)\BBox
                \eand
                [\hat \sfd_+,\hat \sfd_-]\ =\ 0~.
            \end{equation}
            and
            \begin{equation}
                \inner{\hat \sfd_\pm\hat \phi_1}{\hat \phi_2}\ =\ -(-1)^{|\hat \phi_1|}\inner{\hat \phi_1}{\hat \sfd_\pm\hat \phi_2}
            \end{equation}
        \end{subequations}
        for all $\hat \phi_{1,2}\in\hat \frB$. 
        
        \begin{definition}\label{def:compatible_metric}
            Let $\hat\frB=(\hat\frB,\hat\sfd,\hat\sfm_2,\hat\sfb_-)$ be the dg BV algebra defined in~\eqref{eq:tensorProductBVBoxAlgebrasOverH} and let $\frK^0=\frKin^0(\hat \frB)$ be the associated syngamy. Suppose that $\hat\frB$ has a metric $\inner{-}{-}_{\hat\frB}$ of degree~$-6$. We say that a metric $\inner{-}{-}_{\frK^0}$ on the syngamy of degree~$-3$ is \uline{compatible} if 
            \begin{equation}
                \inner{\BBox\hat{\phi}_1[1]}{\hat{\phi}_2[1]}_{\frK^0}\ =\ (-1)^{|\hat{\phi}_1|}\inner{\hat \sfd_-\hat{\phi}_1}{\hat{\phi}_2}_{\hat \frB}
            \end{equation}
            for all $\hat{\phi}_{1,2}[1]\in\frK^0$ with $\hat\sfd_-$ defined in~\eqref{eq:dpm}.
        \end{definition}
        
        \noindent 
        If $\BBox$ is invertible, there is a unique compatible metric on $\frK^0$.
        
        \begin{proposition}\label{prop:syngamyCyclicJustification}
            Consider again the situation in \cref{def:compatible_metric} and suppose that the action of $\BBox$ is invertible. Then
            \begin{equation}
                \inner{\hat{\phi}_1[1]}{\hat{\phi}_2[1]}_{\frK^0}\ \coloneqq\ (-1)^{|\hat{\phi}_1|}\inner{ \BBox^{-1}\hat \sfd_-\hat{\phi}_1}{\hat{\phi}_2}_{\hat \frB}
            \end{equation}
            for all $\hat{\phi}_{1,2}[1]\in\frK^0$ is a compatible metric on the syngamy.
        \end{proposition}
        
        \begin{proof}
            Note that as pointed out in~\eqref{eq:propertiesDpm}, we have that $[\hat\sfd_-,\hat\sfb_-]=2\BBox$. Consequently, $\ker(\hat\sfd_-)\cap\im(\hat\sfb_-)\subseteq\ker(\hat\sfd_-)\cap\ker(\hat\sfb_-)\subseteq\ker(\BBox)$ and so, our assumption that $\BBox$ is invertible implies that $\hat \sfd_-$ is injective. Thus, $\inner{-}{-}_{\frK^0}$ is non-degenerate.
            
            Next, we must show that
            \begin{equation}
                \begin{aligned}
                    \inner{\hat{\phi}_1[1]}{\hat{\phi}_2[1]}_{\frK^0}\ &=\ (-1)^{|\hat{\phi}_1[1]||\hat{\phi}_2[1]|}\inner{\hat{\phi}_2[1]}{\hat{\phi}_1[1]}_{\frK^0}~,
                    \\
                    \inner{\sfd_{\frK^0}\hat{\phi}_1[1]}{\hat{\phi}_2[1]}_{\frK^0}\ &=\ -(-1)^{|\hat{\phi}_1[1]|}\inner{\hat{\phi}_1[1]}{\sfd_{\frK^0}\hat{\phi}_2[1]}_{\frK^0}~,
                    \\
                    \inner{[\hat{\phi}_1[1],\hat{\phi}_2[1]]_{\frK^0}}{\hat{\phi}_3[1]}_{\frK^0}\ &=\ -(-1)^{|\hat{\phi}_1[1]||\hat{\phi}_2[1]|}\inner{\hat{\phi}_2[1]}{[\hat{\phi}_1[1],\hat{\phi}_3[1]]_{\frK^0}}_{\frK^0}
                \end{aligned}
            \end{equation}
            for all $\hat{\phi}_{1,2,3}[1]\in\frK^0$.
            
            Firstly, again using~\eqref{eq:propertiesDpm}, we find
            \begin{equation}
                \begin{aligned}
                    \inner{\hat{\phi}_2[1]}{\hat{\phi}_1[1]}_{\frK^0}\ &=\ (-1)^{|\hat{\phi}_2|}\inner{\BBox^{-1}\hat\sfd_-\hat{\phi}_2}{\hat{\phi}_1}_\frB
                    \\
                    &=\ -\inner{\hat{\phi}_2}{\BBox^{-1}\hat\sfd_-\hat{\phi}_1}_\frB
                    \\
                    &=\ -(-1)^{(|\hat{\phi}_1|+1)|\hat{\phi}_2|}\inner{\BBox^{-1}\hat \sfd_-\hat{\phi}_1}{\hat{\phi}_2}_\frB
                    \\
                    &=\ (-1)^{(|\hat{\phi}_1|+1)(|\hat{\phi}_2|+1)}\inner{\hat{\phi}_1[1]}{\hat{\phi}_2[1]}_{\frK^0}
                    \\
                    &=\ (-1)^{(|\hat{\phi}_1[1]||\hat{\phi}_2[1]|}\inner{\hat{\phi}_1[1]}{\hat{\phi}_3[1]}_{\frK^0}~.
                \end{aligned}
            \end{equation}
            Furthermore,~\eqref{eq:propertiesDpm} also yields
            \begin{equation}
                \begin{aligned}
                    \inner{\sfd_{\frK^0}(\hat{\phi}_1[1])}{\hat{\phi}_2[1]}_{\frK^0}\ &=\ (-1)^{|\hat{\phi}_1|+1}\inner{\BBox^{-1}\hat\sfd_-\hat\sfd_+\hat{\phi}_1}{\hat{\phi}_2}_\frB
                    \\
                    &=\ \inner{\BBox^{-1}\hat\sfd_-\hat{\phi}_1}{\hat\sfd_+\hat{\phi}_2}_\frB
                    \\
                    &=\ (-1)^{|\hat{\phi}_1|}\inner{\hat{\phi}_1[1]}{\sfd_{\frK^0}(\hat{\phi}_2[1])}_{\frK^0}
                    \\
                    &=\ -(-1)^{|\hat{\phi}_1[1]|}\inner{\hat{\phi}_1[1]}{\sfd_{\frK^0}(\hat{\phi}_2[1])}_{\frK^0}~.
                \end{aligned}
            \end{equation}
            Finally,
            \begin{equation}
                \begin{aligned}
                    \inner{[\hat{\phi}_1[1],\hat{\phi}_2[1]]_{\frK^0}}{\hat{\phi}_3[1]}_{\frK^0}\ &=\ (-1)^{|\hat{\phi}_1|}\inner{\{\hat{\phi}_1,\hat{\phi}_2\}[1]}{\hat{\phi}_3[1]}_{\frK^0}
                    \\
                    &=\ (-1)^{|\hat{\phi}_2|+1}\inner{\BBox^{-1}\hat\sfd_-\{\hat{\phi}_1,\hat{\phi}_2\}}{\hat{\phi}_3}_\frB
                    \\
                    &=\ (-1)^{|\hat{\phi}_1|+1}\inner{\hat\sfb_-\hat\sfm_2(\hat{\phi}_1,\hat{\phi}_2)}{\BBox^{-1}\hat\sfd_-\hat{\phi}_3}_\frB
                    \\
                    &=\ (-1)^{|\hat{\phi}_2|+1}\inner{\hat\sfm_2(\hat{\phi}_1,\hat{\phi}_2)}{\BBox^{-1}\hat\sfb_-\hat\sfd_-\hat{\phi}_3}_\frB
                    \\
                    &=\ 2(-1)^{|\hat{\phi}_2|+1}\inner{\hat\sfm_2(\hat{\phi}_1,\hat{\phi}_2)}{\hat{\phi}_3}_\frB
                    \\
                    &=\ 2(-1)^{|\hat{\phi}_1||\hat{\phi}_2|+|\hat{\phi}_2|+1}\inner{\hat{\phi}_2}{\hat\sfm_2(\hat{\phi}_1,\hat{\phi}_3)}_\frB
                    \\
                    &=\ 2(-1)^{|\hat{\phi}_2||\hat{\phi}_3|+|\hat{\phi}_2|+1}\inner{\hat\sfm_2(\hat{\phi}_1,\hat{\phi}_3)}{\hat{\phi}_2}_\frB
                    \\
                    &=\ (-1)^{|\hat{\phi}_2||\hat{\phi}_3|+|\hat{\phi}_2|+|\hat{\phi}_3|}\inner{[\hat{\phi}_1[1],\hat{\phi}_3[1]]_{\frK^0}}{\hat{\phi}_2[1]}_{\frK^0}
                    \\
                    &=\ -(-1)^{|\hat{\phi}_1[1]||\hat{\phi}_2[1]|}\inner{\hat{\phi}_2[1]}{[\hat{\phi}_1[1],\hat{\phi}_3[1]]_{\frK^0}}_{\frK^0}~,
                \end{aligned}
            \end{equation}
            where in the third step we have used~\eqref{eq:propertiesDpm}, inserted the definition~\eqref{eq:derived_bracket} of the derived bracket, and used that $\hat{\phi}_{1,2,3}\in\ker(\hat \sfb_-)$, and in the sixth step we have used the cyclicity of $\sfm_2$.
        \end{proof}
        
        We note, however, that in most cases, $\BBox$ is not invertible. Indeed, the kernel of $\BBox$ usually consists of the asymptotically free fields in the perturbative expansion. Nevertheless, this is a set of measure zero in the space of all fields, and the action is expected to be continuous on this space. We can therefore always extend the inner product between the interacting fields to the full field space, up to technical issues of mathematical analysis that are of little consequence for physical computations. Moreover, operator insertions closely analogous to $\BBox^{-1}\hat\sfd_-$ have also been introduced in the context of Kodaira--Spencer theory or Bershadsky--Cecotti--Ooguri--Vafa (BCOV) theory~\cite{Bershadsky:1993cx}.\footnote{We are grateful to Pietro Antonio Grassi for bringing this to our attention.}
        
        We also note that the inner product is non-local, such that the resulting (Maurer--Cartan) action may be also non-local, since it contains a factor of $\BBox^{-1}$; this happens for instance in the double copy of Chern--Simons theory, which agrees with the non-local action found in~\cite{Ben-Shahar:2021zww}. However, in the common case where $\sfb$ is merely a degree shift (as in the biadjoint scalar) and hence $\hat \sfd$ is a degree-shifted version of $\BBox$, so that $\BBox^{-1}\hat \sfd_-$ combine to a degree shift. Hence, the inner product and the resulting action are local if the original left and right theories are local.
        
        \paragraph{Relation to the double copy.}
        Let us briefly explain how the above construction relates to the usual double copy construction. Recall the two perspectives on the CK duality depicted in~\eqref{eq:CK_diagram}. There was a freedom as to whether to assign the operator $\sfb$ in the colour-stripped propagator $\frac{\sfb}{\BBox}$ to the propagator or to the interaction vertex. In the combination of two kinematic Lie algebras as e.g.~in the double copy construction, we double copy everything except for the operator $\BBox$. Correspondingly, if we consider the combination of the kinematic Lie algebras of two $\BVbox$-algebras $\frB_\rmL$ and $\frB_\rmR$, we can either work with the propagator and interaction vertex\footnote{Here and in the following, we are a bit cavalier with the action of $\frac{1}{\BBox}\notin \frH$, but the meaning should be obvious.} very schematically written as follows:
        \begin{equation}
            \tilde P\ =\ \frac{\sfb_\rmL\otimes \sfb_\rmR}{\BBox}\eand \tilde \mu_2\ =\ \sfm_{2\rmL}\otimes \sfm_{2\rmR}
        \end{equation}
        or, and this is the picture emerging from our tensor product construction,
        \begin{equation}
            P\ =\ \frac{\sfb_\rmL\otimes \sfid+\sfid\otimes \sfb_\rmR}{2\BBox}\eand \mu_2\ =\ \{-,-\}=\{-,-\}_\rmL\otimes \sfm_{2\rmR}+\sfm_{2\rmL}\otimes \{-,-\}_\rmR~.
        \end{equation}
        This is the same choice as made in~\cite{Bonezzi:2022bse, Bonezzi:2023pox, Bonezzi:2023ciu} when defining the double copy.  Note that we have indeed
        \begin{equation}
            \tilde P\tilde\mu_2\ =\ \frac{1}{\BBox}\{-,-\}_\rmL\otimes \{-,-\}_\rmR\ =\ P\mu_2
        \end{equation}
        on $\ker(\hat\sfb_-)$, as required for the equivalence of the two perturbative expansions. The kinematic operator $\sfd_\frK$ should be, when defined, the inverse of the propagator $\tilde P$. Note that for the differential operator $\hat \sfd=\hat \sfd_+$ of the tensor product, we have 
        \begin{equation}
            P\sfd_\frK\ =\ \frac{\sfb_\rmL\otimes\sfid+\sfid\otimes \sfb_\rmR}{2\BBox}(\sfd_\rmL\otimes\sfid+\sfid\otimes\sfd_\rmR)\ =\ \frac{\BBox\otimes \sfid+\sfid\otimes\BBox}{2\BBox}\ =\ \frac{\BBox}{\BBox}~,
        \end{equation}
        as required. Hence, the perturbative expansions of currents of the reduced kinematic dg Lie algebra $\frK^0$ of the restricted tensor product $\frB_\rmL\otimes^\frH\frB_\rmR$ indeed reproduces the expected result of a combination of the kinematic Lie algebras.
        
        To see that the cyclic structure is the correct one is a bit more subtle. It turns out that the differential and the Lie bracket in $\frK^0$ are such that they can be rescaled by a factor $\BBox^{-1}\hat \sfd_-$ to produce local expressions, modulo a few technical subtleties. Instead of presenting an abstract discussion, we simply refer to the concrete examples in \cref{sec:examples}.
        
        \paragraph{Tensoring by colour.}
        The above procedure allows us to produce a field theory or dg Lie algebra from two $\BVbox$-algebras. The inverse of colour-stripping, namely tensoring a dg commutative algebra by a Lie algebra also yields a field theory in the form of a dg Lie algebra. It will turn out that there is a special $\BVbox$-algebra, namely that of the biadjoint scalar field theory, for which both constructions are equivalent. Further details are found in \cref{rem:color_tensoring_by_BVbox}.
        
        \paragraph{Relation to our previous construction.} 
        In our previous work~\cite{Borsten:2021hua}, we considered the factorisation of the dg Lie algebra $\frL$ of a gauge field theory into three parts:
        \begin{equation}
            \frL\ \cong\ \frg\otimes(\frk\otimes_\tau\frScal)~,
        \end{equation}
        where $\frg$ is the gauge Lie algebra, $\frk$ is a kinematic vector space and $\frScal\coloneqq\scC^\infty(\IM^d)[-1]\oplus\scC^\infty(\IM^d)[-2]$ is the BV field space of a field theory of a single, real-valued scalar field. Moreover, $\otimes_\tau$ is a twisted tensor product, a generalisation of a semi-direct product, allowing $\frk$ to act on $\scC^\infty(\IM^d)$.
        
        In our new picture, the $\BVbox$-algebra $\frB$ is an algebraic enhancement of the dg commutative algebra $\frk \otimes_\tau\frScal$. Moreover, if $\frB$ carries an action of the Hopf algebra $\frH_{\IM^d}$, then the kernel of this action can naturally be associated with the space $T^*[-1]\frk$. 
        
        In~\cite{Borsten:2021hua}, we constructed the double copy by doubling the kinematic Lie algebra:
        \begin{equation}
            \frL_\text{double}\ \cong\ \frk\otimes_\tau(\frk \otimes_\tau \frScal)~,
        \end{equation}
        which makes intuitive sense. Here, we tensor together two copies of $\frB$, which results in an unwanted doubling of $\frScal$. This is eliminated by considering the restricted tensor product $\otimes^\frH$, reducing the functions to those on a single copy space-time $\IM^d$, and the kernel of $\hat\sfb_-$, reducing the quadrupled BV field space to the expected one.
        
        \paragraph{Syngamies via compactified space-time.}
        In the case of concrete field theories over Minkowski space $\IM^d$, we run into the usual analytical problems of field theories. For example, the metrics are really defined only for a subset of fields that do not include e.g.~asymptotically free fields. While inconsequential for concrete considerations, trying to resolve these issues leads to some interesting observations.
        
        A natural way to cure these is to compactify space-time from $\IM^d$ to the torus $\IM^d/\Lambda\IZ^d$ with size $\Lambda$ and work with the space $\scC$ of finite linear combinations of (possibly off-mass-shell) plane waves on the torus. Note that the Hopf algebra $\frH_{\IM^d}$ has a natural action on $\scC$ after compactification. Moreover, we can replace the restricted tensor product of \cref{app:restricted_tensor_product} by the ordinary tensor product over the Hopf algebra, because
        \begin{equation}
            \scC\otimes_{\frH_{\IM^d}}\scC\ \cong\ \scC~,
        \end{equation}
        as shown in \cref{thm:compactification-setting-justification}.
        
        Such a compactification is certainly useful since it cures all infrared divergences, but it raises also some conceptual issues: what does it mean to consider scattering amplitudes in a compact space and --- worse --- periodic time? The answer is that, formally, one can always define the scattering amplitudes via the homological perturbation lemma, and this, in turn, is equivalent to computing the scattering amplitudes on flat space subject to the condition that all incoming and outgoing momenta lie on the dual lattice to $\Lambda\IZ^d$. Thus, by setting the radii of the compactified torus appropriately, one can recover all scattering amplitudes.
        
        \subsection{Syngamies of theories with matter fields}\label{sec:syngamiesMatterTheories}
        
        Our above constructions readily extend to theories containing matter fields. In the following, we briefly explain the required constructions. The relevant theorems are more or less the same as for pure gauge theories, and we will omit the proofs if they parallel to those for the pure gauge theory case up to minor and evident changes. 
        
        In the pure gauge case, the syngamy was constructed from a tensor product of $\BVbox$-algebras. The evident generalisation for theories with fields in general representations of a gauge or flavour Lie algebra is to consider tensor products of $\BVbox$-modules.
        
        \paragraph{Tensor products of $\BVbox$-modules.}
        In the following, let $\frH$ be again a restrictedly tensorable cocommutative Hopf algebra. We then have the following result for the tensor product of $\BVbox$-modules.      
        
        \begin{proposition}\label{prop:tensorProductBVBoxModules}
            Given two $\BVbox$-algebras $\frB_\rmL=(\frB_\rmL,\sfd_\rmL,\sfm_{2\rmL},\sfb_\rmL)$ and $\frB_\rmR=(\frB_\rmR,\sfd_\rmR,\sfm_{2\rmR},\sfb_\rmR)$ with $\BBox_\rmL=\BBox_\rmR=\BBox\in\frH$ over $\frH$ and modules $V_\rmL=(V_\rmL,\sfd_{V_\rmL},\acton_{V_\rmL},\sfb_{V_{\rmL}})$ and $V_\rmR=(V_\rmR,\sfd_{V_\rmR},\acton_{V_\rmR},\sfb_{V_{\rmR}})$ over them respectively, the tuple $\hat V=(\hat V,\sfd_{\hat V},\acton_{\hat V},\sfb_{\hat V-})$ with 
            \begin{subequations}\label{eq:tensorProductBVBoxModulesOverH}
                \begin{equation}
                    \hat V\ \coloneqq\ V_\rmL\otimes^\frH V_\rmR
                \end{equation}
                and 
                \begin{equation}
                    \begin{aligned}
                        \sfd_{\hat V}(v_\rmL\otimes^\frH v_\rmR)\ &\coloneqq\ \sfd_{V_\rmL}v_\rmL\otimes v_\rmR+(-1)^{|v_\rmL|}v_\rmL\otimes\sfd_{V_\rmR} v_\rmR~,
                        \\
                        (\phi_\rmL\otimes^\frH\phi_\rmR)\acton_{\hat V}(v_\rmL\otimes^\frH v_\rmR)\ &\coloneqq\ (-1)^{|\phi_\rmR||v_\rmL|}(\phi_\rmL\acton_{V_\rmL}v_\rmL)\otimes(\phi_\rmR\acton_{V_\rmR}v_\rmR)~,
                        \\
                        \sfb_{\hat V-}(v_\rmL\otimes^\frH v_\rmR)\ &\coloneqq\ \sfb_{V_\rmL} v_\rmL\otimes v_\rmR-(-1)^{|v_\rmL|}v_\rmL\otimes\sfb_{V_\rmR} v_\rmR
                    \end{aligned}
                \end{equation}
                for all $v_\rmL\in V_L$, $v_\rmR\in V_R$, $\phi_\rmL\in\frB_\rmL$, and $\phi_\rmR\in\frB_\rmR$ forms a dg BV module over the dg BV algebra $\hat \frB\coloneqq\frB_\rmL\otimes^\frH\frB_\rmR$ defined in \cref{prop:tensorProductBVBoxAlgebrasOverH}.
                The extension of the derived bracket on $\hat\frB$ to $\hat V$ reads as
                \begin{equation}\label{eq:derived_bracket_in_tensor_product_module}
                    \begin{aligned}
                        \{\phi_{\rmL}\otimes^\frH \phi_{\rmR},v_{\rmL}\otimes^\frH v_{\rmR}\}
                        &=(-1)^{|\phi_{\rmR}||v_{\rmL}|}\{\phi_{\rmL},v_{\rmL}\}\otimes (\phi_{\rmR}\acton_{V_R} v_{\rmR})
                        \\
                        &\hspace{0.5cm}
                        -(-1)^{|\phi_{\rmR}||v_{\rmL}|+|\phi_{\rmL}|+|v_{\rmL}|}(\phi_{\rmL}\acton_{V_L} v_{\rmL})\otimes\{\phi_{\rmR},v_{\rmR}\}~.
                    \end{aligned}
                \end{equation}
            \end{subequations}
            
            Provided that both $\frB_\rmL$ and $\frB_\rmR$ come with $\frH$-linear metrics $\inner{-}{-}_\rmL$ and $\inner{-}{-}_\rmR$ of degrees $n_\rmL$ and $n_\rmR$, respectively, and both $V_\rmL$ and $V_\rmR$ come with $\frH$-linear metrics $\inner{-}{-}_{V_\rmL}$ and $\inner{-}{-}_{V_\rmR}$ of degrees $n_\rmL$ and $n_\rmR$, respectively,
            then
            \begin{equation}
                \inner{v_{1\rmL}\otimes v_{1\rmR}}{v_{2\rmL}\otimes v_{2\rmR}}_{\hat V}\ \coloneqq\ (-1)^{|v_{1\rmR}||v_{2\rmL}|+n_\rmR(|v_{1\rmL}|+|v_{2\rmL}|)}\inner{v_{1\rmL}}{v_{2\rmL}}_{V_\rmL}\inner{v_{1\rmR}}{v_{2\rmR}}_{V_\rmR}
            \end{equation}
            defines a $\frH$-bilinear metric for $\hat V$ of degree~$n_\rmL+n_\rmR$ for all $v_{1\rmL,2\rmL}\in V_\rmL$ and $v_{1\rmR,2\rmR}\in V_\rmR$.
        \end{proposition}
        
        \begin{proof}
            The proof follows closely that of \cref{prop:tensorProductBVBoxAlgebrasOverH}.
        \end{proof}
        
        \paragraph{Syngamies.}
        We can straightforwardly generalise syngamies to gauge theories with matter as follows.
        
        \begin{definition}
            Let $\frB_\rmL=(\frB_\rmL,\sfd_\rmL,\sfm_{2\rmL},\sfb_\rmL)$ and $\frB_\rmR=(\frB_\rmR,\sfd_\rmR,\sfm_{2\rmR},\sfb_\rmR)$ be two $\BVbox$-algebras over $\frH$ with $\BBox_\rmL=\BBox_\rmR=\BBox\in\frH$, and let $\hat \frB=(\hat \frB,\sfd_{\hat \frB},\hat \sfm_2,\sfb_{\hat \frB-})$ be their tensor product over $\frH$ as defined in~\eqref{eq:tensorProductBVBoxAlgebrasOverH}. Let $V_\rmL=(V_\rmL,\sfd_{V_\rmL},\acton_{V_\rmL},\sfb_{V_\rmL})$ and $V_\rmR=(V_\rmR,\sfd_{V_\rmR},\acton_{V_\rmR},\sfb_{V_\rmR})$ be $\BVbox$-modules over $\frB_\rmL$ and $\frB_\rmR$, respectively, and let $\hat V=(\hat V,\sfd_{\hat V},\acton_{\hat V},\sfb_{\hat V-})$ be their tensor product over $\frH$ as defined in~\eqref{eq:tensorProductBVBoxModulesOverH}. The \uline{syngamy of the pairs ($\frB_\rmL,V_\rmL$) and ($\frB_\rmR,V_\rmR)$} is the restricted kinematic Lie algebra $\frKin^0(\hat \frB)$ as defined in \cref{def:restricted_kinematic_Lie_algebra}, together with the restricted kinematic Lie algebra module $\frMod^0(\hat V)$ over $\frKin^0(\hat \frB)$, defined in \cref{prop:restrictedKinematicLieModule}.
        \end{definition}
        
        \noindent
        By \cref{prop:restrictedKinematicLieModule}, we know that the syngamy is a dg Lie module over a dg Lie algebra. 
        
        To complete the syngamy, we have to endow $\frMod^0(\hat V)$ with a metric. Analogously to the case of $\BVbox$-algebras and in view of the tensor product~\eqref{eq:tensorProductBVBoxModulesOverH}, let us define
        \begin{equation}\label{eq:dpmModule}
            \sfd_{\hat V\pm}(v_\rmL\otimes v_\rmR)\ \coloneqq\ \sfd_{V_\rmL}v_\rmL\otimes v_\rmR\pm(-1)^{|v_\rmL|}v_\rmL\otimes\sfd_{V_\rmR}v_\rmR    
        \end{equation}
        for all $v_{\rmL,\rmR}\in V_{\rmL,\rmR}$; evidently, $\hat \sfd_+=\hat \sfd$. It is then easy to check that
        \begin{subequations}\label{eq:propertiesDpmModule}
            \begin{equation}
                [\sfd_{\hat V\pm},\sfb_{\hat V-}]\ =\ (1\mp1)\BBox
                \eand
                [\sfd_{\hat V+},\sfd_{\hat V-}]\ =\ 0
            \end{equation}
            and
            \begin{equation}
                \inner{\sfd_{\hat V\pm} v_1}{v_2}_{\hat V}\ =\ -(-1)^{|v_1|}\inner{v_1}{\sfd_{\hat V\pm}v_2}_{\hat V}
            \end{equation}
        \end{subequations}
        for all $v_{1,2}\in \hat V$. Next, we introduce the notion of compatible metrics for modules.
        
        \begin{definition}\label{def:compatible_metric_module}
            Let $(\hat V,\sfd_{\hat V},\acton_{\hat V},\sfb_{\hat V-})$ be the dg BV module over the dg BV algebra $\hat \frB=(\hat \frB,\sfd_{\hat \frB},\hat \sfm_2,\sfb_{\hat \frB-})$ defined in~\eqref{eq:tensorProductBVBoxModulesOverH} and let $(\frK^0,\sfd_{\frK^0},[-,-]_{\frK^0})$ and $(\frV^0,\sfd_{\frV^0},\acton_{\frV^0},\sfb_{\frV^0})$ be the associated syngamy. Suppose that $\frB$ has a metric $\inner{-}{-}_{\hat \frB}$ of degree~$-6$. We say that a metric $\inner{-}{-}_{\frV^0}$ on the dg Lie algebra module $V_0$ in the syngamy $(\frK^0,\frV^0)$ of degree~$-3$ is \uline{compatible} if
            \begin{equation}
                \inner{\BBox v_1[1]}{v_2[1]}_{\frV^0}\ =\ (-1)^{|\phi_1|}\inner{\sfd_{V-}v_1}{v_2}_{\hat V}
            \end{equation}
            for all $v_{1,2}[1]\in \frV^0=\ker(\sfb_{\hat V-})[1]$ with $\sfd_{\hat V-}$ as defined in~\eqref{eq:dpmModule}.
        \end{definition}
        
        \begin{proposition}
            Consider again the situation of \cref{def:compatible_metric_module} and assume that the actions of $\BBox$ on both the BV algebra and the BV module are invertible. Then,
            \begin{equation}
                \inner{v_1[1]}{v_2[1]}_{\frV^0}\ \coloneqq\ (-1)^{|v_1|}\inner{\BBox^{-1}\sfd_{V-}v_1}{v_2}_{\hat V}
            \end{equation}
            for all $v_{1,2}[1]\in \frV^0=\ker(\sfb_{\hat V-})[1]$ with $\sfd_{\hat V-}$ as defined in~\eqref{eq:dpm} is a compatible metric on the syngamy.
        \end{proposition}
        
        \begin{proof}
            The proof is a minor variation of that of \cref{prop:syngamyCyclicJustification}.
        \end{proof}
        
        \section{Examples}\label{sec:examples}
        
        \subsection{Biadjoint scalar field theory}
        \label{ssec:biadjoint_scalar_field_theory}
        
        The simplest and archetypal example of a theory with colour--kinematics duality is certainly the theory of a biadjoint scalar field with evident cubic interaction, a theory that is frequently used as a toy model in the scattering amplitudes literature~\cite{Hodges:2011wm,Vaman:2010ez,Cachazo:2013iea,Monteiro:2013rya,Cachazo:2014xea,Monteiro:2014cda,Chiodaroli:2014xia,Naculich:2014naa,Luna:2015paa,Naculich:2015zha,Chiodaroli:2015rdg,Luna:2016due,White:2016jzc,Cheung:2016prv,Chiodaroli:2017ngp,Brown:2018wss}.
        
        \paragraph{Differential graded Lie algebra.}
        Consider two flavour metric Lie algebras $\frg$ and $\bar\frg$ with bases $\sfe_a$ and $\bar\sfe_{\bar a}$, structure constants $f_{ab}{}^c$ and $\bar f_{\bar a\bar b}{}^{\bar c}$ and metrics $g_{ab}$ and $\bar g_{\bar a\bar b}$, respectively. Classically, a biadjoint scalar field $\varphi$ is a $(\frg\otimes\bar\frg)$-valued function on $\IM^d$, and we write 
        \begin{equation}\label{eq:BAS_fields}
            \varphi\ =\ \sfe_a\otimes\bar\sfe_{\bar a}\otimes\varphi^{a\bar a}\ \in\ (\frg\otimes\bar\frg)\otimes\scC^\infty(\IM^d)~.
        \end{equation}
        We shall be interested in the theory with action functional
        \begin{equation}\label{eq:BAS_action_components}
            S^\text{biadj}\ \coloneqq\ \int\rmd^dx\Big\{\tfrac12\varphi^{a\bar a}g_{ab}g_{\bar a\bar b}\wave \varphi^{b\bar b}+\tfrac{1}{3!}\varphi^{a\bar a}g_{ab}g_{\bar a\bar b}f_{cd}{}^b\bar f_{\bar c\bar d}{}^{\bar b}\varphi^{c\bar c}\varphi^{d\bar d}\Big\}\,.
        \end{equation}
        
        The $L_\infty$-algebra corresponding to this field theory is the dg Lie algebra $\frL^\text{biadj}=\bigoplus_{p\in\IZ}\frL^\text{biadj}_p$ with underlying cochain complex
        \begin{subequations}\label{eq:L_infty_BAS}
            \begin{equation}\label{eq:biadjointComplex}
                \sfCh(\frL^\text{biadj})\ \coloneqq\ 
                \Big(
                \begin{tikzcd}[column sep=20pt]
                    *\arrow[r] & \underbrace{(\frg\otimes\bar\frg)\otimes\scC^\infty(\IM^d)}_{\eqqcolon\,\frL^\text{biadj}_1} \arrow[r,"\sfid_{\frg\otimes\bar\frg}\otimes\wave"] &[30pt] \underbrace{(\frg\otimes\bar\frg)\otimes\scC^\infty(\IM^d)}_{\eqqcolon\,\frL^\text{biadj}_2} \arrow[r] & * 
                \end{tikzcd}
                \Big)~,
            \end{equation}
            where $*$ denotes the trivial vector space. In particular, we have the field\footnote{in the sense of the BV formalism, i.e.~as opposed to an anti-field} $\varphi\in\frL^\text{biadj}_1$, the corresponding anti-fields $\varphi^+=\varphi^+_{a\bar a}\sfe^{a}\bar \sfe^{\bar a}\in\frL^\text{biadj}_2$ for $\sfe^a=g^{ab}\sfe_a$ and $\bar \sfe^{\bar a}=\bar g^{\bar a\bar b}\bar \sfe_{\bar a}$, and the only non-trivial component of the differential $\mu_1\coloneqq\sfid_{\frg\otimes\bar\frg}\otimes\wave$. The non-vanishing components of the cyclic inner product are
            \begin{equation}\label{eq:inner_prod_biadj}
                \inner{\varphi}{\varphi^+}\ \coloneqq\ \int\rmd^dx\,\varphi^{a\bar a}\varphi^+_{a\bar a}~.
            \end{equation}
            The interactions are encoded in the Lie bracket $\mu_2\colon\frL^\text{biadj}\times\frL^\text{biadj}\rightarrow \frL^\text{biadj}$, and the only non-trivial components are
            \begin{equation}\label{eq:binary_product_biadjoint}
                \mu_2(\varphi_1,\varphi_2)\ \coloneqq\ f_{ab}{}^c\sfe_c\otimes\bar f_{\bar a\bar b}{}^{\bar c}\bar\sfe_{\bar c}\otimes\varphi^{a\bar a}_1\varphi^{b\bar b}_2
            \end{equation}
        \end{subequations}
        for all $\varphi_{1,2}\in\frL^\text{biadj}_1$.
        
        \paragraph{$\BVbox$-algebra and colour--kinematics duality.}
        Regarding one of the two Lie algebras (say $\frg$) as colour, we may strip it off to form a $\BVbox$-algebra. This amounts to the factorisation
        \begin{equation}\label{eq:factorisation_C_infty_biadj}
            \frL^\text{biadj}\ \cong\ \frg\otimes\frB^\text{biadj}.
        \end{equation}
        Explicitly, $\frB^\text{biadj}$ has the underlying cochain complex
        \begin{subequations}\label{eq:C_bar_BAS}
            \begin{equation}
                \sfCh(\frB^\text{biadj})\ \coloneqq\
                \Big(
                \begin{tikzcd}[column sep=40pt]
                    *\arrow[r] & \underbrace{\bar\frg\otimes\scC^\infty(\IM^d)}_{\eqqcolon\,\frB^\text{biadj}_1} \arrow[r,"\sfid_{\bar\frg}\otimes\wave"] & \underbrace{\bar\frg\otimes\scC^\infty(\IM^d)}_{\eqqcolon\,\frB^\text{biadj}_2} \arrow[r] & * 
                \end{tikzcd}
                \Big)
            \end{equation}
            with $\varphi=\bar\sfe_{\bar a}\,\varphi^{\bar a}\in\frB^\text{biadj}_1$, $\varphi^+=\bar\sfe^{\bar a}\,\varphi^+_{\bar a}\in\frB^\text{biadj}_2$, and $\sfd\coloneqq\sfid_{\bar \frg}\otimes\wave$. Note that we continue to label colour-stripped fields by $\varphi$, slightly abusing notation. Furthermore, we have
            \begin{equation}
                \sfm_2(\varphi_1,\varphi_2)\ \coloneqq\ \bar f_{\bar a\bar b}{}^{\bar c}\bar\sfe_{\bar c}\otimes\varphi^{\bar a}_1\varphi^{\bar b}_2
                \eand
                \inner{\varphi}{\varphi^+}\ \coloneqq\ \int\rmd^dx\,\varphi^{\bar a}\varphi^+_{\bar a}~.
            \end{equation}
        \end{subequations}
        
        To extend $\bar\frB^\text{biadj}$ to a $\BVbox$-algebra, we need to endow it with an operator $\sfb$ such that $[\sfd,\sfb]=\wave$. The evident choice here is the shift isomorphism (denoted $[1]$).
        \begin{subequations}\label{eq:BVbox-biadj}
            \begin{equation}
                \sfb\ \coloneqq\ [1]\,:\,\frB_2^\text{biadj}\ \xrightarrow{~\cong~}\ \frB_1^\text{biadj}~.
            \end{equation}
            The derived bracket $\{-,-\}$ of~\eqref{eq:derived_bracket} is then 
            \begin{equation}\label{eq:biadj_Lie_bracket}
                \begin{gathered}
                    \{\varphi_1,\varphi_2\}\ =\ \sfb(\sfm_2(\varphi_1,\varphi_2))\ =\ \bar f_{\bar a\bar b}{}^{\bar c}~\bar\sfe_{\bar c}\otimes\varphi_1^{\bar a}\varphi_2^{\bar b}\in \frB^\text{biadj}_1~,
                    \\
                    \{\varphi_1,\varphi^+_2\}\ =\ \sfm_2(\varphi_1,\sfb\varphi^+_2)\ =\ \bar f_{\bar a\bar b}{}^{\bar c}g^{\bar b\bar d}~\bar\sfe_{\bar c}\otimes\varphi_1^{\bar a}\varphi^+_{2\bar d}\ =\ \{\varphi^+_2,\varphi_1\}\in \frB^\text{biadj}_2~.
                \end{gathered}
            \end{equation}
        \end{subequations}        
        It is then easy to check that all the remaining axioms are satisfied; in particular, $\sfb$ is of second order, which amounts to the following specialisation of~\eqref{eq:b_second_order}:
        \begin{equation}
            \begin{aligned}
                0\ =\ -\sfm_2(\varphi_1,\sfb(\sfm_2(\varphi_2,\varphi_3)))+\sfm_2(\varphi_2,\sfb(\sfm_2(\varphi_1,\varphi_3)))-\sfm_2(\varphi_3,\sfb(\sfm_2(\varphi_1,\varphi_2)))
            \end{aligned}
        \end{equation}
        for all $\varphi_{1,2,3}\in \frB^\text{biadj}_1$, a consequence of the Jacobi identity. We will denote the resulting $\BVbox$-algebra also by $\frB^{\bar\frg}$, to indicate the choice of Lie algebra $\bar \frg$. This $\BVbox$-algebra will play an important role as a replacement for the colour Lie algebra $\bar\frg$ later.
        
        According to \cref{cor:pure_gauge_CK_duality}, the existence of the $\BVbox$-algebra $\frB^{\bar\frg}$ proves that the biadjoint scalar field theory possesses CK duality on its currents. Because $\sfb$ is a shift isomorphism, all fields $\varphi$ are of the form $\varphi=\sfb\varphi^+$ for some anti-field $\varphi^+$, and hence CK-duality extends to the amplitudes.
        
        \paragraph{Syngamy.}
        We now follow the approach of \cref{sec:doubleCopySyngamies} and consider the syngamy of the two $\BVbox$-algebras $\frB^\frg$ and $\frB^{\bar\frg}$ for $\frg,\bar\frg$ some Lie algebras. To this end, we note that the $\BVbox$-algebra comes with a natural action of the Hopf algebra $\frH_{\IM^d}$ from \cref{ex:H-Box-Minkowski}, and it is easy to check that all operations are $\frH_{\IM^d}$-linear with respect to this action. 
        
        The restricted tensor product $\hat\frB\coloneqq\frB^\frg\otimes^\frH\frB^{\bar\frg}$ has then the underlying cochain complex
        \begin{equation}\label{eq:rough_chain_complex_biadj}
            \Big(
            \begin{tikzcd}[column sep=20pt]
                *\arrow[r] & \underbrace{\frg\otimes\bar\frg\otimes\scC^\infty(\IM^d)}_{\eqqcolon\,\hat \frB_2} \arrow[r] & \underbrace{\IR^2\otimes \frg\otimes\bar\frg\otimes\scC^\infty(\IM^d)}_{\eqqcolon\,\hat \frB_3} \arrow[r] & \underbrace{\frg\otimes\bar\frg\otimes\scC^\infty(\IM^d)}_{\eqqcolon\,\hat \frB_4} \arrow[r] & * 
            \end{tikzcd}
            \Big)~,
        \end{equation}
        and we have a corresponding kinematic Lie algebra $\frK$ concentrated in degrees $1,2,3$. We will be interested in the shifted Lie bracket $[-,-][1]$ on fields $\varphi_{1,2}\in \frB^\frg\otimes^\frH\frB^{\bar \frg}$, which reads as 
        \begin{equation}\label{eq:shifted_Lie_bracket_rough}
            [\varphi_1,\varphi_2][1]\ =\ \sfb\sfm_2(\varphi^{(1)}_1,\varphi^{(1)}_2)\otimes \sfm_2(\varphi_1^{(2)},\varphi_2^{(2)})+\sfm_2(\varphi^{(1)}_1,\varphi^{(1)}_2)\otimes \sfb\sfm_2(\varphi_1^{(2)},\varphi_2^{(2)})~,
        \end{equation}
        where we used again Sweedler notation $\varphi_{1,2}=\varphi_{1,2}^{(1)}\otimes \varphi_{1,2}^{(2)}$.
        
        We note that the cochain complex~\eqref{eq:rough_chain_complex_biadj} is split in half, into the kernel and cokernel of the operator 
        \begin{equation}
            \hat\sfb_-\ \coloneqq\ [1]\otimes\sfid-\sfid\otimes[1]~.
        \end{equation}
        In particular, the kernel is given by $\hat \frB_2$ as well as the symmetrised sum of the two copies of $\frg\otimes\bar\frg\otimes\scC^\infty(\IM^d)$ contained in $\hat \frB_3$.        
        
        With \cref{cor:kin_alg_is_dg_Lie_for_BV-algebra}, we note that the restricted kinematic Lie algebra $\frK^0=\frKin^0(\hat \frB)$, i.e.~$\frK$ restricted to the kernel, cf.~\cref{def:restricted_kinematic_Lie_algebra}, together with the differential $\hat\sfd[1]$ becomes a dg Lie algebra $\frK^0$ with underlying cochain complex
        \begin{equation}\label{eq:restricted_chain_complex_biadj}
            \sfCh(\frK^0)\ \coloneqq\ \Big(
            \begin{tikzcd}[column sep=20pt]
                *\arrow[r] & \underbrace{\frg\otimes\bar\frg\otimes\scC^\infty(\IM^d)}_{\eqqcolon\, \frK^0_1} \arrow[r,"\sfid_\frg\otimes \sfid_{\bar \frg}\otimes\wave"] &[30pt] \underbrace{\frg\otimes\bar\frg\otimes\scC^\infty(\IM^d)}_{\eqqcolon\,\frK^0_2} \arrow[r] & * 
            \end{tikzcd}
            \Big)~.
        \end{equation}
        Moreover, the product $\mu_2$ can be read off from~\eqref{eq:shifted_Lie_bracket_rough}, and its non-trivial components are given by 
        \begin{equation}
            \mu_2(\varphi_1,\varphi_2)\ \coloneqq\ \sfe_c\otimes\bar\sfe_{\bar c}\otimes\varphi^{a\bar a}_1\varphi^{b\bar b}_2f_{ab}{}^c\bar f_{\bar a\bar b}{}^{\bar c}
        \end{equation}
        for all $\varphi_{1,2}\in\frK^0_1$. 
        
        On fields that are not in the kernel of $\BBox=\wave$ (e.g.~Schwartz-type functions describing interacting fields), a metric can be defined by means of~\cref{prop:syngamyCyclicJustification}:
        \begin{equation}
            \inner{\varphi_1}{\varphi_2}_{\frK^0}\ \coloneqq \ \inner{\wave^{-1}\hat\sfd_-\varphi_1}{\varphi_2}_{\hat\frB}~.
        \end{equation}
        Because of the symmetry of $\inner{-}{-}_{\frK^0}$ established in~\cref{prop:syngamyCyclicJustification}, we can assume that $\varphi_1$ is a field, i.e.~and element in $\hat\frB_2[1]$, without loss of generality. In this case, $|\varphi_1^{(1)}|=1$ and hence 
        \begin{equation}
            \begin{aligned}
                \wave^{-1}(\hat \sfd_-\varphi_1)\ &=\ \wave^{-1}(\wave \varphi_1^{(1)}\otimes \varphi_1^{(2)}-\varphi_1^{(1)}\otimes\wave\varphi_1^{(2)})
                \\
                &=\ (\varphi_1^{(1)}[-1]\otimes\varphi_1^{(2)}-\varphi_1^{(1)}\otimes\varphi_1^{(2)}[-1])~,
            \end{aligned}
        \end{equation}
        where we used that $(\wave^{-1}\varphi_1^{(1)})\otimes \varphi_1^{(2)}=\varphi_1^{(1)}\otimes (\wave^{-1}\varphi_1^{(2)})$. The restriction to $\ker(\hat\sfb_-)\subseteq\hat\frB_2\oplus \hat \frB_3$ with $\frK^0$, together with the removal of the infinite volume factor along the constant directions (cf.~the discussion in \cref{ssec:syngamy_pure_gauge}), then leads to the expected inner product~\eqref{eq:inner_prod_biadj}.
        
        Altogether, we see that $\frK^0=\frL^\text{biadj}$ and, as expected, the resulting double copy is the biadjoint scalar theory with Lie algebras $\frg$ and $\bar \frg$. 
        
        \begin{remark}\label{rem:color_tensoring_by_BVbox}
            Note that, as predicted above, the role of the colour Lie algebras is played by the $\BVbox$-algebras $\frB^\frg$ and $\frB^{\bar \frg}$. In particular, constructing the syngamy of a $\BVbox$-algebra $\frB$ with the $\BVbox$-algebra $\frB^\frg$ produces the same field theory (in the form of a dg Lie algebra) as if we tensored the dg commutative algebra underlying $\frB$ with $\frg$. This relation is quite evident for field theories where the differential in $\frB$ is $\sfd=\wave$, such as biadjoint scalar and conventional rewritings of Yang--Mills theory, but it also extends to Chern--Simons theory, as we shall see in \cref{ssec:pure_Chern-Simons}.
        \end{remark}
        
        \subsection{Biadjoint scalar theory with bifundamental matter}\label{ssec:biadj_bifund_theory}
        
        The simplest example including matter fields is the biadjoint scalar theory coupled to a  bifundamental scalar, cf.~\cite{Naculich:2014naa, Naculich:2015zha, Anastasiou:2016csv, Brown:2018wss}, i.e.~a scalar field taking values in the (metric) fundamental representations\footnote{The choice of fundamental representation is just for concreteness sake; the theory straightforwardly generalises to arbitrary metric representations.} $R\otimes \bar R$ of the Lie algebras $\frg\otimes\bar\frg$.
        
        \paragraph{Differential graded Lie algebra.}
        Explicitly, we couple the biadjoint scalar field theory~\eqref{eq:BAS_action_components} to the action for a bifundamental scalar field
        \begin{equation}
            S^{\text{biadj-fun}}\ \coloneqq\ S^{\text{biadj}}+\int\rmd^dx\Big\{\tfrac12\psi^{i\bar\imath}\wave g_{ij}\bar g_{\bar \imath\bar\jmath}\psi^{j\bar\jmath}+\tfrac{1}{2}\psi^{i\bar\imath}g_{ij}\bar g_{\bar\imath\bar\jmath}T_{aj}{}^{i}\bar T_{\bar a\bar\jmath}{}^{\bar\imath}\varphi^{a\bar a}\psi^{j\bar\jmath}\Big\}\,,
        \end{equation}
        where
        \begin{equation}
            \psi\ =\ \sfe_i\otimes\sfe_{\bar \imath}\otimes\psi^{i\bar\imath}\ \in\ (R\otimes\bar R)\otimes\scC^\infty(\IM^d)~,
        \end{equation}
        and where we have introduced bases, $\sfe_i$ and $\sfe_{\bar\imath}$, metrics $g_{ij}$ and $g_{\bar \imath\bar \jmath}$ with respect to these bases, and structure constants, $T_{ a j}{}^{i}$ and $\bar T_{\bar a\bar\jmath}{}^{\bar\imath}$, describing the interactions, for $R$ and $\bar R$, respectively. 
        
        The underlying cochain complex of the dg Lie algebra $\frL^\text{biadj-fun}$ is that of $\frL^\text{biadj}$ enlarged to 
        \begin{equation}\label{eq:bi_adj_fun_cochain}
            \sfCh(\frL^\text{biadj-fun})\ \coloneqq\ 
            \left(
            \begin{tikzcd}[
                row sep=-3pt,column sep=12pt
                ]
                &(\frg\otimes\bar\frg)\otimes\scC^\infty(\IM^d) \arrow[r,"\sfid_{\frg\otimes\bar\frg}\otimes\wave"] &[25pt] (\frg\otimes\bar\frg)\otimes\scC^\infty(\IM^d) 
                \\
                *\arrow[r,shorten >=8ex]  & \oplus & \oplus \arrow[r,shorten <=8ex] & * 
                \\
                & (R\otimes\bar R)\otimes\scC^\infty(\IM^d)\arrow[r,"\sfid_{R\otimes\bar R}\otimes\wave"]  & (R \otimes\bar R)\otimes\scC^\infty(\IM^d)
            \end{tikzcd}
            \right),
        \end{equation}
        where the anti-fields $\varphi^+$ and $\psi^+$ belong to the degree shifted copies of $(\frg\otimes\bar\frg)\otimes\scC^\infty(\IM^d)$ and  $ (R\otimes\bar R)\otimes\scC^\infty(\IM^d)$, respectively. The fields $\varphi, \psi$ and anti-fields $\varphi^+, \psi^+$ have dg Lie algebra degree~$1$ and $2$ (and, thus, ghost degree~$0$ and $1$), respectively. 
        
        The interactions are encoded in the graded anti-symmetric Lie bracket 
        \begin{equation}
            \mu_2\,:\,\frL^\text{biadj-fun}\times\frL^\text{biadj-fun}\ \rightarrow\ \frL^\text{biadj-fun}~,
        \end{equation}
        which has non-trivial components 
        \begin{equation}\label{eq:binary_product_biadjoint_bifun}
            \begin{split}
                \mu_2(\varphi_1,\varphi_2)\ &\coloneqq\ \sfe_cf_{ab}{}^c\otimes\bar\sfe_{\bar c}\bar f_{\bar a\bar b}{}^{\bar c} \otimes \varphi^{a\bar a}_1\varphi^{b\bar b}_2~,
                \\
                \mu_2(\varphi,\psi)\ &\coloneqq\ \sfe_j  T_{ ai}{}^{j}\otimes\bar\sfe_{\bar \jmath} \bar T_{\bar a\bar \imath}{}^{\bar \jmath}\otimes \varphi^{a\bar a}\psi^{i\bar \imath}=\mu_2(\varphi, \psi)~,
                \\
                \mu_2(\psi_1,\psi_2)\ &\coloneqq\ \sfe_a  T^{ a}{}_{ij}\otimes\bar\sfe_{\bar a} \bar T^{\bar a}{}_{\bar \imath\bar \jmath}  \otimes \psi_{1}^{i\bar \imath}\psi_{2}^{j\bar \jmath}
            \end{split}
        \end{equation}
        for all $\varphi,\varphi_{1,2}$ and $\psi,\psi_{1,2}$ in the evident subspaces of $\frL^\text{biadj-fun}_1$. The assumption that $R,\bar R$ are metric implies the existence of a cyclic structure with non-vanishing components
        \begin{equation}\label{eq:inner_product_biadj-fun}
            \inner{\varphi+\psi}{\varphi^++\psi^+}\ \coloneqq\ \int\rmd^dx\,\left\{\varphi^{a\bar a}\varphi^+_{a\bar a} + \psi^{i\bar \jmath}\psi^+_{i\bar \jmath}\right\}
        \end{equation}
        for all $\varphi,\psi$ and $\varphi^+,\psi^+$ in the evident subspaces of $\frL^\text{biadj-fun}_1$ and $\frL^\text{biadj-fun}_2$, respectively. Altogether, $\frL^\text{biadj-fun}$ is a metric nilpotent dg Lie algebra.  
        
        \paragraph{Colour--flavour-stripping.}
        The next step is to perform a colour--flavour-stripping as explained in \cref{ssec:colour-flavour-stripping}. Without loss of generality, we can chose $\frg$ and $R$ to be the colour--flavour factors and we expect a factorisation of the dg Lie algebra $\frL^\text{biadj-fun}$ as follows:
        \begin{subequations}\label{eq:colour_flavour_stripping}
            \begin{equation}
                \frL^\text{biadj-fun}\ \cong\ \frg\otimes \frB^\text{biadj}~\oplus~R\otimes V^\text{bifun}
            \end{equation}
            with $\frB^\text{biadj}$ as defined in~\eqref{eq:C_bar_BAS} and $V^\text{bifun}=(V^{\text{bifun}},\sfd_{V^{\text{bifun}}},\acton_{V^{\text{bifun}}})$ is a (dg) module over $\frB^\text{biadj}$ with underlying cochain complex 
            \begin{equation}\label{eq:bi_adj_fun_cochain_cs}
                \sfCh(V^{\text{bifun}})\ \coloneqq\ 
                \left(
                \begin{tikzcd}[
                    row sep=-3pt
                    ]
                    *\arrow[r] & \bar R\otimes\scC^\infty(\IM^d) \arrow[r,"\sfid_{\bar R}\otimes\wave"] &  \bar R\otimes\scC^\infty(\IM^d) \arrow[r] & * 
                \end{tikzcd}
                \right).
            \end{equation}
            The action is defined as 
             \begin{equation}
                 \varphi\acton_{V^{\text{bifun}}}\psi\ \coloneqq\ \bar\sfe_{\bar \imath}T_{a\bar\jmath}{}^{\bar\imath}\otimes \varphi^a\psi^{\bar\jmath}~.
             \end{equation}
            A short computation then verifies the factorisation~\eqref{eq:colour_flavour_stripping}.
        \end{subequations}
        
        \paragraph{$\BVbox$-module structure and colour--kinematics duality.}
        We have already seen that the dg commutative algebra $\frB^\text{biadj}$ can be enriched to a $\BVbox$-algebra $\frB^{\bar g}$; it remains to enriched $V^\text{bifun}$ to a $\BVbox$-algebra module, which we will denote by the same letter. As in the case of the dg commutative algebra, also here the required additional operator $\sfb_{V^{\text{bifun}}}$ is given by the evident degree shift
        \begin{subequations}
            \begin{equation}
                \sfb_{V^{\text{bifun}}}\ \coloneqq\ [1]\,:\,V_2^\text{biadj}\ \xrightarrow{~\cong~}\ V_1^\text{biadj}~.
            \end{equation}
            The derived bracket $\{-,-\}_{V^{\text{bifun}}}\colon\frB^{\text{biadj}} \times V^{\text{bifun}}\rightarrow V^{\text{bifun}}$ as defined in~\eqref{eq:derived_bracket_module}
            reads as
            \begin{equation}\label{eq:derived_bracket_module_bifun}
                \begin{split}
                    \{\varphi,\psi\}_{V^{\text{bifun}}}\ &=\ (\bar\sfe_{\bar\imath}\bar T_{\bar a\bar\jmath}{}^{\bar\imath}\otimes\varphi^{\bar a}\psi^{\bar \jmath})[1]~,
                    \\
                    \{\varphi,\psi^+\}_{V^{\text{bifun}}}\ &=\ \bar\sfe_{\bar\imath}\bar T_{\bar a\bar\jmath}{}^{\bar\imath}\otimes\varphi^{\bar a}(\psi^{+\bar\jmath}[1])~,
                    \\
                    \{\varphi^+,\psi\}_{V^{\text{bifun}}}\ &=\ \bar\sfe_{\bar\imath} \bar T_{\bar a\bar\jmath}{}^{\bar\imath}\otimes(\varphi^{+\bar a}[1])\psi^{\bar\jmath}~,
                    \\
                    \{\varphi^+,\psi^+\}_{V^{\text{bifun}}}\ &=\ 0
                \end{split}
            \end{equation}
        \end{subequations}        
        for all $\varphi\in\frB^\text{biadj}_1$, $\varphi^+\in\frB^\text{biadj}_2$, $\psi\in V^\text{bifun}_1$, and $\psi^+\in V^\text{bifun}_2$. Together with the derived bracket of the biadjoint scalar theory, see~\eqref{eq:biadj_Lie_bracket}, it follows that $\{-,-\}_{V^{\text{bifun}}}$ satisfies the shifted Poisson identity \eqref{eq:derivedBracketModuleJacobi} and 
        \begin{equation}
            (V^{\text{bifun}},\sfd_{V^{\text{bifun}}},\acton_{V^{\text{bifun}}},\sfb_{V^{\text{bifun}}})
        \end{equation}
        is a $\text{BV}^{\Box}$-module over the $\BVbox$-algebra $\frB^\text{biadj}$.
        
        \paragraph{Double copy.}
        To illustrate syngamies involving matter fields, let us consider the syngamy of two copies of $(\frB^\text{biadj},V^\text{bifun})$ with Lie algebras and metric fundamental representations $(\frg,R)$ and $(\bar\frg,\bar R)$, respectively. The restricted tensor product of the two $\BVbox$-algebras is given in~\eqref{eq:rough_chain_complex_biadj}, and the restricted tensor product $\hat V$ of the $\BVbox$-modules similarly has underlying cochain complex 
        \begin{equation}\label{eq:rough_chain_complex_bifun}
            \Big(
            \begin{tikzcd}[column sep=17pt]
                *\arrow[r] & \underbrace{R\otimes\bar R\otimes\scC^\infty(\IM^d)}_{\eqqcolon\,\hat V_2} \arrow[r] & \underbrace{\IR^2\otimes R\otimes\bar R\otimes\scC^\infty(\IM^d)}_{\eqqcolon\,\hat V_3} \arrow[r] & \underbrace{R\otimes\bar R\otimes\scC^\infty(\IM^d)}_{\eqqcolon\,\hat V_4} \arrow[r] & * 
            \end{tikzcd}
            \Big)~,
        \end{equation}
        and by \cref{prop:kinematic_module}, there is a corresponding underlying module $\frV$ for the kinematic Lie algebra $\frK$ of $\hat \frB$ defined in~\eqref{eq:rough_chain_complex_biadj}. Again, the cochain complex~\eqref{eq:rough_chain_complex_bifun} is split in half into the kernel and cokernel of the operator 
        \begin{equation}
            \sfb_{\hat V-}\ \coloneqq\ [1]\otimes\sfid-\sfid\otimes[1]~,
        \end{equation}
        and $\ker(\sfb_{\hat V-})$ consist of $\hat \frB_2$ and a symmetrised sum of the two copies of $R\otimes\bar R\otimes\scC^\infty(\IM^d)$ in $\hat \frB_3$. Restricted to this kernel, $\frV$ becomes a dg module $\frV^0$ over the reduced kinematic dg Lie algebra $\frK^0$ of $\hat \frB$ by \cref{prop:restrictedKinematicLieModule}.
        
        The reduced kinematic dg Lie algebra $\frK^0$ and the dg module $\frV^0$ now combine into a single dg Lie algebra, and it is not hard to see that this dg Lie algebra is $\frL^\text{biadj-fun}$, the dg Lie algebra we started from. In particular, the double copy of the metric~\eqref{eq:inner_product_biadj-fun} is fully analogous to that of the metric in biadjoint scalar theory. Hence, the syngamy of two copies of $(\frB^\text{biadj},V^\text{bifun})$ yields a biadjoint scalar theory coupled to bifundamental matter.
        
        \subsection{The sesquiadjoint scalar and kinematic \texorpdfstring{$L_\infty$}{L-infinity}-algebras}
        
        In order to illustrate at least one case of a kinematic $L_\infty$-algebra (again, anticipating our future work~\cite{Borsten:2022aa}), we introduce a sesquiadjoint scalar field theory.
        
        \paragraph{Differential graded Lie algebra and colour-stripping.}
        The setup is almost identical to the biadjoint scalar, except that we replace $\bar\frg$ in $\frg\otimes\bar\frg$ with a vector space $W$ equipped with an anti-symmetric binary operation $[-,-]:W\times W\to W$ that does not (necessarily) fulfil the Jacobi identity.\footnote{Such products were considered, e.g., in~\cite[\S3]{Hohm:2017cey}.}
        
        Colour-stripping, we have a dg commutative algebra $\frC^\text{seqadj}$ with underlying cochain complex,           
        \begin{equation}
            \sfCh(  \frC^\text{seqadj})\ \coloneqq\ 
            \left(
            \begin{tikzcd}[
                row sep=-3pt
                ]
                *\arrow[r] & W\otimes\scC^\infty(\IM^d) \arrow[r,"\sfid_{W}\otimes\wave"] &  W \otimes\scC^\infty(\IM^d) \arrow[r] & * 
            \end{tikzcd}
            \right),
        \end{equation}
        and non-trivial graded symmetric product
        \begin{equation}
            \begin{gathered}
                \sfm_2\,:\,W\otimes\scC^\infty(\IM^d)\times W\otimes\scC^\infty(\IM^d)\ \rightarrow\ (W\otimes\scC^\infty(\IM^d))[-1]~,
                \\
                \sfm(\varphi_1,\varphi_2)\ \coloneqq\ \sfe_cf_{ab}{}^c\otimes(\varphi^a_1\varphi^b_2)~,
            \end{gathered}
        \end{equation}        
        where we have introduced a basis, $\sfe_a$, for $W$ and structure constants $f_{ab}{}^c$ for the binary operation $[-,-]$ that does not obey the Jacobi identity. 
        
        \paragraph{Kinematic $L_\infty$-algebra.}
        As before, the shift isomorphism 
        \begin{equation}
            \sfb\ \coloneqq\ [1]\,:\,\frC_2^\text{sesqadj}\ \xrightarrow{~\cong~}\ \frC_1^\text{sesqadj}~.
        \end{equation}
        satisfies $\sfd\sfb+\sfb\sfd=\wave$. The non-trivial higher-order differentials, as defined in~\eqref{eq:def_Phir} with $\delta=\sfb=[1]$, are given by  
        \begin{equation}
            \begin{aligned}
                \Phi^1_\sfb(\phi^+_1)\ &\coloneqq\ \phi^+_1[1]~,
                \\
                \Phi^2_\sfb(\phi_1,\phi_2)\ &\coloneqq\ \sfm(\phi_1,\phi_2)[1]~,
                \\
                \Phi^2_\sfb(\phi_1,\phi^+_2)\ &\coloneqq\ \sfm(\phi_1,\phi^+_2[1])~,
                \\
                \Phi^2_\sfb(\phi^+_1,\phi_2)\ &\coloneqq\ - \sfm(\phi^+_1[1], \phi_2)~,
                \\
                \Phi^{3}_\sfb(\phi_1,\phi_2,\phi_3)\ &\coloneqq\ \sfm(\phi_1[1],\sfm(\phi_2,\phi_3))-\sfm(\sfm (\phi_1,\phi_{2})[1],\phi_3)
                +\sfm(\phi_2,\sfm(\phi_1,\phi_3)[1])~.
            \end{aligned}
        \end{equation}
        By \cref{prop:derived_brackets_form_L_infty_algebra}, the higher products $\mu_i\coloneqq\Phi^i_\sfb$ define an $L_\infty$-algebra on the shifted cochain complex $\sfCh(\frC^\text{seqadj})[1]$. Here, $\mu_3$ (as always) describes the homotopy that encodes the failure of $\mu_2$ to satisfy the Jacobi identity, which in turn is due to the bracket $[-,-]$ not satisfying the Jacobi identity. This derived $L_\infty$-algebra is directly analogous to the derived Lie algebra of the kinematic Lie algebra. It is an example of the kinematic $L_\infty$-algebras described in \cref{ssec:kinematic_L_infty_algebras}. We stress that the homotopy Jacobi relations in this example are non-trivial.
        
        \paragraph{General setting.}
        Since there is always a graded commutative product $\sfm_2$, every perturbative Lagrangian BV theory has such a kinematic $L_\infty$-algebra (under the very weak assumption that there is a suitable $\sfb$). We plan to explore the significance of this observation further in future work. The most radical implication that one might envisage, is that every theory can be double-copied using the kinematic $L_\infty$-algebra structure. This seems (at least superficially) unlikely, and the standard double copy argument~\cite{Bern:2010yg} for scattering amplitudes is certainly not generalised in an obvious fashion.  
        
        In the above example, in particular, the differential $\Phi^1_\sfb$ has trivial cohomology, and hence the $L_\infty$-algebra of the Koszul hierarchy is quasi-isomorphic to the trivial one\footnote{This is in close analogy to the Lie or $L_\infty$-algebra of inner derivations of a Lie or $L_\infty$-algebra being contractible or quasi-isomorphically trivial.}. By contrast, the usual kinematic Lie algebra is non-trivial precisely because we can halve the field content and render the (cohomology of the) kinematic algebra non-trivial. This possibly suggests that generic kinematic $L_\infty$-algebras are not of use in the double copy.
        
        \subsection{Pure Chern--Simons theory}\label{ssec:pure_Chern-Simons}
        
        So far, we encountered scalar field theories which directly exhibited CK duality. In this example, we increase the complexity by introducing gauge symmetry while still maintaining manifest CK duality.       
        
        \paragraph{Differential graded Lie algebra.}
        Let $\frg$ be a metric Lie algebra with basis $\sfe_a$ relative to which we have structure constants $f_{ab}{}^c$ and a metric $g_{ab}$. Furthermore, let $\Omega^p(\IM^3)$ be the differential $p$-forms on $\IM^3$ with the exterior differential $\rmd\colon\Omega^p(\IM^3)\rightarrow\Omega^{p+1}(\IM^3)$ and let $\star:\Omega^p(\IM^3)\rightarrow \Omega^{3-p}(\IM^3)$ be the usual Hodge operator with respect to the Minkowski metric on $\IM^3$.
        
        The field content of Chern--Simons theory consists of the Chern--Simons gauge potential $A=\sfe_a\otimes A^a$ with $A^a\in\Omega^1(\IM^3)$ and its ghost $c=\sfe_a\otimes c^a$ with $c^a\in\Omega^0(\IM^3)$ paired with their anti-fields $A^+=\sfe_a\otimes A^{+a}$ with $A^{+a}\in\Omega^2(\IM^3)$ and its ghost $c^+=\sfe_a\otimes c^{+a}$ with $c^{+a}\in\Omega^3(\IM^3)$. In addition to this usual BV field content, we also add a Nakanishi--Lautrup field $n=\sfe_a\otimes n^a$ with $n^a\in \Omega^1(\IM^3)$ and an anti-ghost $\bar c=\sfe_a\otimes \bar c^a$ with $\bar c^a\in \Omega^1(\IM^3)$ together with the corresponding anti-fields $n^+$ and $\bar c^+$. After gauge fixing with the gauge-fixing fermion $\Psi=\int \left\{g_{ab}\bar c^a\wedge \star (\rmd^\dagger A^b-\tfrac12 n^b)\right\}$, the action functional looks as follows:\footnote{For the $L_\infty$-algebra before gauge-fixing, see e.g.~\cite{Borsten:2021hua}.}
        \begin{equation}\label{eq:CSBVAction}
            \begin{aligned}
                S^\text{CS}\ &\coloneqq\ \int\Big\{\tfrac12g_{ab}A^a\wedge\rmd A^b+\tfrac13g_{ab}f_{cd}{}^bA^a\wedge A^c\wedge A^d
                \\
                &\kern2cm-g_{ab}\bar c^a\wedge \star\rmd^\dagger(\nabla c)^b+\tfrac12 g_{ab}n^a\wedge \star n^b+g_{ab}n^a\wedge \star\rmd^\dagger A^b\Big\}\,.
            \end{aligned}
        \end{equation}
        The dg Lie algebra structure is readily read off, and we directly continue with colour-stripping.
        
        \paragraph{Colour-stripping and $\BVbox$-algebra structure.}
        All of the fields take values in the colour Lie algebra, and after colour-stripping, we obtain a dg commutative algebra $\frB^\text{CS}$, which comes with a natural operator $\sfb$, and has the following underlying  bidirectional complex, cf.~\eqref{eq:gauge-fixed-bidirectional-complex}:
        \begin{equation}\label{eq:CS_bidirectional_complex}
            \begin{tikzcd}[row sep=1cm,column sep=2.7cm]
                \stackrel{c}{\Omega^0(\IM^d)}  \arrow[rdd,shift left,"-\wave"] & \stackrel{A}{\Omega^1(\IM^d)} \arrow[rd,shift right,"-\rmd^\dagger"', pos=0.1] \arrow[r,shift left,"\rmd"] & \stackrel{A^+}{\Omega^2(\IM^d)} \arrow[l,shift left,"\rmd^\dagger"] \arrow[ld,shift left,"\star \rmd", pos=0.1] & \stackrel{c^+}{\Omega^4(\IM^d)} \arrow[ldd,shift left,"-\star"]
                \\
                & \stackrel{n}{\Omega^0(\IM^d)} \arrow[ru,shift left,"\star\rmd", pos=0.1] & \stackrel{n^+}{\Omega^0(\IM^d)} \arrow[lu,shift right,"-\rmd"', pos=0.1]
                \\
                \underbrace{\phantom{\stackrel{c}{\Omega^0(\IM^d)}}}_{\coloneqq\,\frB^\text{CS}_0} & \underbrace{\stackrel{\bar c^+}{\Omega^0(\IM^d)}}_{\coloneqq\,\frB^\text{CS}_1}  \arrow[luu,shift left,"-\sfid"]& \underbrace{\stackrel{\bar c}{\Omega^0(\IM^d)}}_{\coloneqq\,\frB^\text{CS}_2} \arrow[ruu,shift left,"-\wave\star"] &
                \underbrace{\phantom{\stackrel{c^+}{\Omega^0(\IM^d)}}}_{\coloneqq\,\frB^\text{CS}_3}
            \end{tikzcd}
        \end{equation}
        The binary products are given as follows:
        \begin{equation}
            \begin{aligned}
                \sfm_2\left(\begin{pmatrix} A_1 \\ n_1 \\ \bar c^+_1\end{pmatrix},\begin{pmatrix} A_2 \\ n_2 \\ \bar c^+_2\end{pmatrix}\right)\ &\coloneqq\ \begin{pmatrix}A_1\wedge A_2 \\ 0 \\ 0 \end{pmatrix}\ \in\ \frB_2^\text{CS}~,
                \\
                \sfm_2\left(\begin{pmatrix} A_1 \\ n_1 \\ \bar c^+_1\end{pmatrix},c_2\right)\ &\coloneqq\ \begin{pmatrix} 0 \\ 0 \\ \rmd^\dagger(A_1 c_2)\end{pmatrix}\ \in\ \frB_1^\text{CS}~,
                \\
                \sfm_2\left(c_1,\begin{pmatrix} A_2^+ \\ n_2^+ \\ \bar c_2\end{pmatrix}\right)\ &\coloneqq\ \begin{pmatrix} c\rmd \bar c\\ 0 \\ 0\end{pmatrix}\ \in\ \frB_2^\text{CS}~,
            \end{aligned}
        \end{equation}
        where the notation and positions of the components in the arguments and images in these expressions correspond to those of diagram~\eqref{eq:CS_bidirectional_complex}. We clearly see that the operator $\sfb$ implied by~\eqref{eq:CS_bidirectional_complex} is of second order with respect to these binary products, and we obtain indeed a $\BVbox$-algebra structure. Moreover, there is an evident metric with the following, non-vanishing components:
        \begin{equation}
            \begin{aligned}
                \inner{A}{A^+}\ &\coloneqq\ \int A\wedge A^{+}
                ~,~~~
                &\inner{c}{c^+}\ &\coloneqq\ \int c\wedge c^{+}~,
                \\
                \inner{n}{n^+}\ &\coloneqq\ \int n\wedge \star n^{+}
                ~,~~~
                &\inner{\bar c}{\bar c^+}\ &\coloneqq\ \int \bar c\wedge \star \bar c^{+}~.                
            \end{aligned}
        \end{equation}
        
        \paragraph{Colour--kinematics duality.}
        We recall that the tree-level amplitudes of Chern--Simons theory on $\IM^d$ are all trivial. However, following e.g.~\cite{Ben-Shahar:2021zww}, we can consider the homotopy transfer to harmonic forms\footnote{i.e.~amputated correlation functions with external legs being harmonic forms} on $\IM^d$, and it is the CK duality for this Feynman diagram expansion that the $\BVbox$-algebra $\frB^\text{CS}$ manifests. Moreover, we have $\BBox=[\sfd,\sfb]=\wave$, which is evident from the diagram~\eqref{eq:CS_bidirectional_complex}, so that the arising kinematic Lie algebra is indeed for the ordinary form of CK duality with propagator $\frac{1}{\wave}$. Note that here, we have full loop level CK-duality.
        
        \paragraph{Comments.}
        Before coming to the double copy, let us comment on a new feature in Chern--Simons theory. Contrary to previous theories, the $\sfb$-operator, concretely the component $\sfb|_{\frB^\text{CS}_2}$, is no longer simply a shift isomorphism. Therefore the kernel of $\sfb$  no longer cleanly cuts the BV field space into fields and anti-fields, and some parts of the anti-fields are left in $\ker(\sfb)$. These parts, however, are very small; they consists of exact and coexact anti-fields $A^+$ of the gauge potential (which on $\IM^3$ amounts to a harmonic scalar field) as well as constant Nakanishi--Lautrup anti-fields $n^+$. We can usually ignore this issue, as the common constraints on a quantum field theory such as locality etc.~allow us to truncate away subspaces that are not full $\scC^\infty(\IM^d)$-modules. If one feels uncomfortable about this truncation, one can also extend our notion of $\BVbox$-algebra to $\BVbox$-algebras with polarisations, i.e.~structures that compatibly split the field space into fields and complementing anti-fields, respecting in particular~\eqref{eq:fields_inbetween}. Because of the additional technicalities that do not add much in concrete discussions, we refrained from using these notions.
        
        \paragraph{Double copy.}
        With the above technicality out of the way, we can follow our usual prescription using the evident Hopf algebra $\frH_{\IM^3}$ generated by the translation operators on $\IM^3$, and consider the kernel of $\hat \sfb_-$, cf.~\eqref{eq:def_red_tensor_ops}. This leads to a BV field space with the fields, i.e.~the (truncated) elements of $\ker(\sfb_\rmL)\otimes^\frH\ker(\sfb_\rmR)\subseteq\ker(\hat \sfb_-)$ given by the direct sums of the spaces
        \begin{equation}
            \begin{tikzcd}[row sep=0cm, column sep=0.2cm]
                \stackrel{c_\rmL\otimes c_\rmR}{\Omega^0(\IM^3)} & \stackrel{c_\rmL\otimes A_\rmR}{\Omega^1(\IM^3)}\oplus \stackrel{A_\rmL\otimes c_\rmR}{\Omega^1(\IM^3)}
                & \stackrel{A_\rmL\otimes A_\rmR}{\Omega^1(\IM^3)\otimes\Omega^1(\IM^3)}
                \\
                & \stackrel{c_\rmL\otimes n_\rmR}{\Omega^0(\IM^3)}\oplus \stackrel{n_\rmL\otimes c_\rmR}{\Omega^0(\IM^3)}
                & \stackrel{A_\rmL\otimes n_\rmR}{\Omega^1(\IM^3)}\oplus \stackrel{n_\rmL\otimes A_\rmR}{\Omega^1(\IM^3)}
                \\
                & & \stackrel{c_\rmL\otimes\bar c_\rmR}{\Omega^0(\IM^3)}\oplus \stackrel{\bar c_\rmL\otimes c_\rmR}{\Omega^0(\IM^3)}
                & \stackrel{A_\rmL\otimes \bar c_\rmR}{\Omega^1(\IM^3)}\oplus \stackrel{\bar c_\rmL\otimes A_\rmR}{\Omega^1(\IM^3)}
                \\
                \underbrace{\phantom{\stackrel{c_\rmL\otimes c_\rmR}{\Omega^0(\IM^3)}}}_{\eqqcolon\,\frL^\text{CSCS}_{-1}}
                &
                \underbrace{\stackrel{c_\rmL\otimes A_\rmR}{\Omega^1(\IM^3)}\oplus \stackrel{A_\rmL\otimes c_\rmR}{\Omega^1(\IM^3)}}_{\eqqcolon\,\frL^\text{CSCS}_{0}}
                &
                \underbrace{~~~~~~~\stackrel{n_\rmL\otimes n_\rmR}{\Omega^0(\IM^3)}~~~~~~~}_{\eqqcolon\,\frL^\text{CSCS}_{1}}
                &
                \underbrace{\stackrel{n_\rmL\otimes \bar c_\rmR}{\Omega^0(\IM^3)}\oplus \stackrel{\bar c_\rmL\otimes n_\rmR}{\Omega^0(\IM^3)}}_{\eqqcolon\,\frL^\text{CSCS}_{2}}
                & \underbrace{\stackrel{\bar c_\rmL\otimes \bar c_\rmR}{\Omega^0(\IM^3)}}_{\eqqcolon\,\frL^\text{CSCS}_{3}}
            \end{tikzcd}
        \end{equation}
        where we have indicated the origin of the subspaces using the component notation of~\eqref{eq:CS_bidirectional_complex}, and we have also indicated the degree of the fields in the resulting double-copied dg Lie algebra $\frL^\text{CSCS}$. The corresponding anti-fields form a grade-shifted and flipped copy dual of this field space, and together they form the graded vector space of the dg Lie algebra $\frL^\text{CSCS}$.
        
        The differential and the product of the dg Lie algebra $\frL^\text{CSCS}$ are straightforwardly constructed, but the cyclic structure is a bit more complicated. For the propagating field components, i.e.~those components of fields that are not in the kernel of $\wave$, we can use~\cref{prop:syngamyCyclicJustification} to define this inner product. We can then continue the resulting expression to all fields by locality. 
        
        Altogether, the double copy leads to a rather unusual BV field theory, whose physical part was first presented in~\cite{Ben-Shahar:2021zww}. Explicitly, the kinetic term of the action for the physical fields given by the $(1,1)$-biforms $A_\rmL\otimes A_\rmR\in \Omega^1(\IM^3)\otimes \Omega^1(\IM^3)$ reads as 
        \begin{equation}\label{eq:CSCS_kin}
            \frac14\int\left\{(A_\rmL\otimes A_\rmR)\bullet\wave^{-1}\hat\sfd_-\mu_1(A_\rmL\otimes A_\rmR)\right\}\ =\ \frac12\int \left\{(A_\rmL\otimes A_\rmR)\bullet\frac{\rmd\otimes \rmd}{\wave}A_\rmL\otimes A_\rmR)\right\},
        \end{equation}
        where the product $\bullet:\Omega^{p_1}(\IM^3)\otimes\Omega^{q_1}(\IM^3)\times\Omega(\IM^3)^{p_1}\otimes\Omega^{q_2}(\IM^3)\rightarrow\Omega^{p_1+p_2}(\IM^3)\otimes\Omega^{q_1+q_2}(\IM^3)$ on biforms is defined as 
        \begin{equation}\label{eq:CSCS_int}
            (A_1\otimes B_1)\bullet(A_2 \otimes B_2)\ \coloneqq\ (A_1\wedge A_2)\otimes(B_1\wedge B_2)~.
        \end{equation}
        
        The interaction terms for the physical fields are given by 
        \begin{equation}
            \int\tfrac1{3!}(A_\rmL\otimes A_\rmR)\bullet(A_\rmL\otimes A_\rmR)\bullet(A_\rmL\otimes A_\rmR)~,
        \end{equation}
        and together,~\eqref{eq:CSCS_kin} and~\eqref{eq:CSCS_int} are the double-copied Chern--Simons action of~\cite{Ben-Shahar:2021zww} in the $(p,q)$-formalism of~\cite{deMedeiros:2002qpr,deMedeiros:2003osq}. A further study of this action is certainly warranted, particularly, since it will also appear in \cref{ssec:M2branes} in the context of M2-brane models.        
        
        We note that a useful outcome of our double copy construction is the full BV triangle required for studying biform theories.
        
        \subsection{Self--dual Yang--Mills theory and self--dual gravity}\label{sec:twistorSpaceSDYM}
        
        The field theories studied in the previous sections came with in a $\BVbox$-algebra in their original formulation. This is contrary to the case of Yang--Mills theory, where the action has to be rewritten in an equivalent form in order to manifest CK duality, cf.~\cite{Bern:2010yg,Tolotti:2013caa} and the detailed discussion in~\cite{Borsten:2021hua}. A theory that is in between both cases is self-dual Yang--Mills (SDYM) theory, which features CK duality on its currents~\cite{Monteiro:2011pc}. Presented in light-cone gauge, it is essentially a biadjoint scalar field theory, and therefore manifestly CK-dual. In the gauge-invariant form of the Chalmers--Siegel action~\cite{Chalmers:1996rq}, which contains an enlarged field content featuring also an anti-self-dual 2-form field, however, it does require an equivalent rewriting in order to manifest CK duality. As stated in the introduction, CK duality is ultimately a symmetry of the action and therefore we may expect an organisational principle that leads to a manifest formulation. 
        
        In~\cite{Borsten:2022vtg}, we showed that the twistor space $Z$, i.e.~the total space of the holomorphic vector bundle $\caO(1)\oplus \caO(1)$ over $\IC P^1$ can serve as such an organising principle. Explicitly, SDYM theory can be equivalently formulated as a holomorphic Chern--Simons theory on $Z$, and, as for ordinary Chern--Simons theories, there is a natural adjoint of the Dolbeault differential that is of second order with respect to the binary product, and hence an operator $\sfb$ that enhances the evident dg commutative algebra structure for holomorphic Chern--Simons theory on $Z$ to a $\BVbox$-algebra structure. Even better, we have $\BBox=\wave$, the d'Alembertian on space-time in this situation, so that the kinematic Lie algebra describes indeed ordinary CK duality on currents and, in the maximally supersymmetric case, even loop level amplitudes. An elegant example of the formalism presented in this paper can be found in~\cite{Borsten:2023paw}, where we consider an action equivalent to and reminiscent of the light-cone formulation of SDYM theory on twistor space, which elegantly double copies to an analogous formulation of self-dual gravity, also on twistor space. For all the technical details of the above, we refer to~\cite{Borsten:2022vtg} and~\cite{Borsten:2023paw}.
        
        Instead, let us briefly compare this result with that of~\cite{Bonezzi:2023pox}. In this paper, the authors considered the equations of motion and gauge transformations of SDYM theory  on space-time, together with its colour-stripped dg commutative algebra, in order to study the kinematic algebra in the absence of space-time gauge-fixing (as opposed to the light-cone gauge analysis of~\cite{Monteiro:2011pc}). As for Chern--Simons theory, there is a natural candidate for the $\sfb$-operator, namely $\sfb=\rmd^\dagger$, the usual Hodge dual of the de Rham differential. As it stands, this differential is not second order with respect to the binary product, as the latter is not just a wedge product of forms, but at least on fields, it contains a projection operator. Therefore, as observed in this paper, the derived bracket~\eqref{eq:derived_bracket} in this picture is not a Lie bracket, but as explained in \cref{ssec:kinematic_L_infty_algebras}, the binary bracket in a kinematic $L_\infty$-algebra. This is precisely what the authors of~\cite{Bonezzi:2023pox} observe to lowest order: there is a ternary operation, given by the expression from the Koszul hierarchy, so that the derived bracket satisfies the homotopy Jacobi identity of an $L_\infty$-algebra.
        
        The authors of~\cite{Bonezzi:2023pox}, however, obtain more. They show that the graded Poisson relation~\eqref{eq:BV_GB_Poisson} of the derived bracket~\eqref{eq:derived_bracket} is violated in a controlled way, and they compute the correction to this order. This leads to parts of a $\BVbox_\infty$-algebra~\cite{Reiterer:2019dys}, see also~\cite{Borsten:2022aa}. In this sense, CK duality is not manifested literally, but only `up to homotopy'. The usual strictification theorem for homotopy algebras applies, and hence one can rewrite the theory in an equivalent form that makes use of an ordinary $\BVbox$-algebra, and therefore manifests CK duality. We note that the 3-bracket inserted in~\cite{Bonezzi:2023pox} corresponds, after inserting a metric, and further an action principle, to a Tolotti--Weinzierl-type term that may be added to the action to manifest CK-duality to this order.
        
        We also note that our rewriting on twistor space directly produces such a rewriting. Twistor space $Z$ is diffeomorphic to the space\footnote{In the supersymmetric case, $\IR^{4}$ is replaced by $\IR^{4|2\caN}$.} $\IR^4\times \IC P^1$, and one can perform a mode expansion along $\IC P^1$. Some of these infinitely many modes correspond to physical fields on space-time, the rest will be the auxiliary fields that produce the Tolotti--Weinzierl terms\footnote{These are terms in the action that vanish due to the Jacobi identity of the colour algebra, cf.~\cite{Tolotti:2013caa} and also~\cite{Borsten:2021hua}.} in the action necessary for manifesting CK duality. The obtained action will hence be the usual first order formulation of SDYM theory given by the Chalmers--Siegel action~\cite{Chalmers:1996rq} plus additional trivial terms, which will become non-trivial after colour-stripping. Note that the twistor formulation allows for a choice of gauge, usually called space-time gauge, that directly leads to the Chalmers--Siegel action~\cite{Mason:2005zm,Boels:2006ir}, see also~\cite[\S5.2]{Wolf:2010av}.
        
        Altogether, we saw that twistor space can serve as an organising principle that naturally leads to CK-dual formulations of field theories. In the case of full Yang--Mills theory, one can use ambitwistor space, and while this description still yields a kinematic Lie algebra, the operator $\BBox$ is not the space-time d'Alembertian operator, so we only obtain a generalised form of CK duality. For this case, a more suitable organisational principle is found in pure spinor space, to which we turn next.
        
        \subsection{Pure spinor formulation of supersymmetric  Yang--Mills theory}\label{ssec:pure_spinors_SYM}
        
        Closely related to the twistor construction of self-dual Yang--Mills theory mentioned in the previous section is the pure spinor formulation of supersymmetric gauge theories. In particular, ten-dimensional supersymmetric Yang--Mills theory can be formulated as Chern--Simons type action on pure spinor space, providing a natural $\BVbox$-algebra structure. Contrary to the ambitwistor space construction of four-dimensional supersymmetric Yang--Mills theory in~\cite{Borsten:2022vtg}, however, there is a natural operator $\sfb$ that leads to $\BBox=\wave$, the d'Alembertian, so that conventional CK duality can be established~\cite{Ben-Shahar:2021doh} for amplitude currents. As explained in~\cite{Borsten:2023reb}, however, reducing the currents to tree-level numerators in this picture involves a diverging integral over the pure spinors. This can be fixed by an alternative choice of $\sfb$~\cite{Borsten:2023reb}, and we briefly review this construction.
        
        \paragraph{Pure spinor space.}
        For the ten-dimensional supersymmetric Yang--Mills theory, we start from the superspace
        \begin{equation}
            \hat\scM_{\text{10d}\,\caN=1}\ \coloneqq\ \IM^{10|16}\times (\IR^{2|1}\otimes \caS_{\text{10d}\,\text{MW}})~,
        \end{equation}
        where $\IM^{10|16}$ is the ten-dimensional $\caN=1$ Minkowski superspace and $\caS_{\text{10d}\,\text{MW}}$ is the space of Majorana--Weyl spinors in ten dimensions. Hence, $\IR^{2|1}\otimes\caS_{\text{10d}\,\text{MW}}$ is the $(32|16)$-dimensional superspace with coordinates\footnote{Note that $\rmd\bar\lambda_A$ is indeed common notation for a coordinate.} $(\lambda^A,\bar\lambda_A,\rmd\bar\lambda_A)$, which transform in the $\textbf{16}$, $\textbf{16}$, and $\overline{\textbf{16}}$ of $\sfSpin(1,9)$, respectively. The pure spinor space $\scM_{\text{10d}\,\caN=1}$ is obtained from this space as the quadric 
        \begin{equation}\label{eq:pure_spinor_quadrics_10d}
            \lambda^A\gamma^M_{AB}\lambda^B\ =\ \bar\lambda_A\gamma^{M\,AB}\bar\lambda_B\ =\ \bar\lambda_A\gamma^{M\,AB}\mathrm d\bar\lambda_B\ =\ 0~,
        \end{equation}
        where $\gamma^M_{AB}$ and $\gamma^{M\,AB}$ are the evident Clifford algebra generators. Operationally, we will work with fields on $\hat \scM_{\text{10d}\,\caN=1}$ and identify the fields on $\scM_{\text{10d}\,\caN=1}$ as a quotient of these by the ideal $\caI$ generated by the quadrics~\eqref{eq:pure_spinor_quadrics_10d}.
        
        The space $\hat \scM_{\text{10d}\,\caN=1}$ comes with a natural vector field $Q$, 
        \begin{equation}
            Q\ =\ \lambda^A D_A+\rmd \bar \lambda_A \parder{\bar\lambda_A}~,
        \end{equation}
        where the $D_A$ are the usual covariant superderivatives on $\IM^{10|16}$, satisfying
        \begin{equation}
            D_AD_B+D_BD_A\ =\ -2\gamma^M_{AB}\parder{x^M}~.
        \end{equation}
        This vector field $Q$ descends to a differential on the functions on $\scM_{\text{10d}\,\caN=1}$ due to~\eqref{eq:pure_spinor_quadrics_10d}; in particular, $\caI$ is a differential ideal.
        
        There is now a family of operators $\sfb$ such that 
        \begin{equation}\label{eq:b-relations}
            \sfb^2\ =\ 0
            \eand
            Q\sfb+\sfb Q\ =\ \wave
        \end{equation}
        with $\wave$ the d'Alembertian on $\IM^{10}$~\cite{Bjornsson:2010wm,Bjornsson:2010wu,Berkovits:2013pla,Jusinskas:2013sha,Cederwall:2022qfn}. Usually, a Lorentz-covariant choice 
        \begin{equation}\label{eq:def_b_Lcov_10}
            \sfb_\text{Lorentz}\ \coloneqq\ \frac{\bar\lambda_A\gamma^{M\,AB}D_B}{2(\lambda^A\bar\lambda_A)}\parder{x^M}+\cdots~,
        \end{equation}
        $M=0,\ldots,9$, is made, but this choice is less suitable for our purposes; instead, we work with the $\sfb$-operator of the $Y$-formalism~\cite{Matone:2002ft,Oda:2005sd,Oda:2007ak}, 
        \begin{equation}\label{eq:def_b_operator_10}
            \sfb\ \coloneqq\ -\frac{v_A\gamma^{M\,AB}D_B}{2\lambda^A v_A}\parder{x^M}~,
        \end{equation}
        where we have chosen a reference pure spinor $v$,  satisfying $v_A\gamma^{M\,AB}v_B=0$. Evidently, this operator is of second order, and it is straightforward to verify that the relations~\eqref{eq:b-relations} are satisfied. 
        
        We summarise the properties of all the objects introduced so far in \cref{tab:coordinatesOperators}.
        
        \begin{table}[ht]
            \vspace{15pt}
            \begin{center}
                \begin{tabular}{ccccc}
                    \toprule
                    & \multirow{2}{*}{$\sfSpin(1,9)$} & mass & Gra{\ss}mann & ghost
                    \\[-3pt]
                    & & dimension & degree & number
                    \\
                    \midrule
                    $x$ & $\mathbf{10}$ & $-1\phantom+$ & $0$ & $\phantom{+}0\phantom+$
                    \\
                    $\theta$ & $\mathbf{16}$ & $-\frac12\phantom+$ & $1$ & $\phantom{+}0\phantom+$
                    \\
                    $\lambda$ & $\mathbf{16}$& $-\frac12\phantom+$ & $0$ & $\phantom{+}1\phantom+$
                    \\
                    $\bar\lambda$ & $\overline{\mathbf{16}}$ & $\phantom{+}\frac12\phantom+$ & $0$ & $-1\phantom+$
                    \\
                    $\rmd\bar\lambda$ & $\overline{\mathbf{16}}$ & $\phantom{+}\frac12\phantom+$ & $1$ & $\phantom{+}0\phantom+$
                    \\
                    \midrule
                    $D$ & $\overline{\mathbf{16}}$ & $\phantom{+}\frac12\phantom+$ & $1$ & $\phantom{+}0\phantom+$
                    \\
                    $Q$ & $\mathbf{1}$ & $\phantom{+}0\phantom+$ & $1$ & $\phantom{+}1\phantom+$
                    \\
                    $\sfb$ & $\mathbf{1}$ & $\phantom{+}2\phantom+$ & $1$ & $-1\phantom+$
                    \\
                    \bottomrule
                \end{tabular}
                \caption{Properties of ten-dimensional coordinates and operators.}
                \label{tab:coordinatesOperators}
            \end{center}
        \end{table}
        
        \paragraph{Pure spinor action and Siegel gauge.} 
        There is now a simple, Chern--Simons type formulation of the BV action of ten-dimensional supersymmetric Yang--Mills theory~\cite{Berkovits:2001rb,Movshev:2003ib}. The field content is organised into a single scalar superfield $\Psi$ on $\scM_{\text{10d}\,\caN=1}$ of ghost number $1$, mass dimension $0$, and Gra{\ss}mann degree~$1$, which takes values in the metric gauge Lie algebra $(\frg,\inner{-}{-}_\frg)$. Together with the natural volume form $\Omega_{\text{10d}\,\caN=1}$ on pure spinor space $\scM_{\text{10d}\,\caN=1}$ that was given in~\cite{Berkovits:2005bt}, we can write down the action functional 
        \begin{equation}\label{eq:action_10d}
            S^{\text{10d}\,\caN=1}\ \coloneqq\ \int\Omega_{\text{10d}\,\caN=1}~\inner{\Psi}{Q\Psi+\tfrac13[\Psi,\Psi]}_\frg~.
        \end{equation}
        
        The underlying cochain complex of the pure spinor BV $L_\infty$-algebra is compactly encoded in the  space of smooth functions on the  pure spinor space, 
        \begin{equation}
            \sfCh(\frL^\text{psYM})\ \cong\ \scC^\infty(\frg\otimes \scM_{\text{10d}\,\caN=1})~.
        \end{equation}
        To recover the component (anti-)fields and identify the graded vector spaces to which they belong, one Taylor-expands the $\frg$-valued superfield $\Psi(x^M,\theta^A,\lambda^A,\bar\lambda_A,\rmd\bar\lambda_A)$ with respect to the $\lambda^A,\bar\lambda_A,\rmd\bar\lambda_A$ coordinates. 
        
        There is an evident dg Lie algebra structure on $\scC^\infty(\frg\otimes\scM_{\text{10d}\,\caN=1})$. The differential is given by $\sfid_\frg\otimes Q$ and 
        \begin{equation}
            \mu_2(\Psi_1,\Psi_2)\ \coloneqq\ [\Psi_1,\Psi_2]\ =\ f_{ab}{}^c\sfe_c\otimes\Psi^{a}_1\cdot\Psi^{b}_2~,
        \end{equation}
        where $-\cdot-$ is just the pointwise product on $\scC^\infty(\scM_{\text{10d}\,\caN=1})$. 
        
        In order to compute perturbative scattering amplitudes, cf.~\cite{Bjornsson:2010wm,Bjornsson:2010wu}, we can work in Siegel gauge,
        \begin{equation}\label{eq:defSiegelGauge}
            \sfb\Psi\ =\ 0~.
        \end{equation}
        Note that our choice~\eqref{eq:def_b_operator_10} of $\sfb$ imposes a form of axial gauge along $v$.
        
        The propagator in this gauge is simply $\frac{\sfb}{\wave}$, and, evidently, this is a generalisation of the propagator we encountered in the discussion of pure Chern--Simons theory in \cref{ssec:pure_Chern-Simons}.
        
        \paragraph{$\BVbox$-algebra structure and colour--kinematics duality.} 
        It is now rather evident that the metric dg commutative algebra induced by the action~\eqref{eq:action_10d} becomes a $\BVbox$-algebra 
        \begin{equation}
            \frB^\text{psSYM}\ \coloneqq\ (\scC^\infty(\scM_{\text{10d}\,\caN=1}),Q,-\cdot-,\sfb)
        \end{equation}
        with $\sfb$ given by~\eqref{eq:def_b_operator_10} from the $Y$-formalism. The only fact to check is that $\sfb$ is of second order with respect to the function product on pure spinor space $\scM_{\text{10d}\,\caN=1}$, but this is evident from the explicit expression for $\sfb$ in~\eqref{eq:def_b_operator_10}. Note that the pure spinor field already contains the Nakanishi--Lautrup field and anti-ghosts (as well as the corresponding anti-fields), so that it indeed packages up all the BV fields required for a gauge-fixed action, cf.~\cite{Cederwall:2010wf}. 
        
        By \cref{cor:pure_gauge_CK_duality}, we thus have a theory with manifest CK-dual parametrisation of its currents, and this observation had been made before in~\cite{Ben-Shahar:2021zww} for the commonly used, covariant $\sfb$-operator~\eqref{eq:def_b_Lcov_10}. Using the $\sfb$-operator~\eqref{eq:def_b_operator_10} of the $Y$-formalism, this result extends to the amplitudes, as we explain now, following the argument in~\cite{Borsten:2023reb}.
        
        Recall from the discussion in \cref{ssec:kinematic_Lie_algebras} that in order to convert a current into an amplitude, we have to remove the propagator on the outgoing leg and pair it off with another incoming, asymptotically free field. This latter pairing involves an integral over pure spinor space, which may lead to divergences. These divergences certainly cancel in the tree-level amplitudes, but they do not necessarily cancel in individual diagrams. This is a problem since we can only establish CK duality, if we can extract finite numerators of a CK-dual parametrisation of the scattering amplitudes. 
        
        The tree-level numerators can suffer from two types of divergences. Firstly, we have to account for the fact that pure spinor space\footnote{contrary e.g.~to the base of twistor space, which provides an alternative ordering principle for CK duality~\cite{Borsten:2022vtg}} is non-compact, and therefore we will encounter \uline{infrared-like} divergences from integrating over the unbounded $(\lambda,\bar\lambda)$-domains. These divergences are mostly harmless, and there is a well-known $Q$-invariant regularisation of the integral measure by a factor of the form
        \begin{equation}
            \rme^{-\epsilon\{Q,\chi\}}\ =\ \rme^{-\epsilon(\lambda^A\bar\lambda_A+\cdots)}~,
        \end{equation}
        where $\epsilon$ is a real positive constant and $\chi$ is a pure spinor field of ghost degree $-1$ which can be chosen a $\chi=-\bar\lambda_A\theta^A+\cdots$~\cite{Berkovits:2005bt,Marnelius:1990eq}, cf.~also~\cite{Cederwall:2022fwu}. This manifestly suppresses the would-be infrared divergences~\cite{Berkovits:2006vi} while preserving the kinematic Lie algebra (and hence CK duality) since the bracket is merely scaled. 

        Secondly, there are \uline{ultraviolet-like} divergences arising when $(\lambda,\bar\lambda)\to0$. Those are more difficult to deal with when trying to establish CK duality as we shall explain next. In particular, in the covariant non-minimal formalism~\eqref{eq:def_b_Lcov_10}, the scattering amplitude integrands will contain singularities of the form $\frac{1}{(\lambda^A\bar\lambda_A)^n}$ due to the propagator $\frac{\sfb_\text{Lorentz}}{\wave}$ and the Siegel gauge~\eqref{eq:defSiegelGauge}. However, the regulator~\cite{Berkovits:2006vi}
        \begin{equation}
            \sfb_{\text{Lorentz},\,\epsilon}\ \coloneqq\ \rme^{-\epsilon(w_A\bar w^A+\cdots)}\sfb_\text{Lorentz}
        \end{equation}
        will render these singularities harmless since $\sfb_{\text{Lorentz},\,\epsilon}$ is $Q$-cohomologuous to $\sfb_{\text{Lorentz},\,\epsilon=0}$. Here, $w_A,\bar w^A$ are conjugate to $\lambda^A,\bar\lambda_A$ and whilst this superficially spoils the second-orderness of the $\sfb_\text{Lorentz}$, all that is needed that the difference between this operator and the one that is used in the end is $Q$-exact. Moreover, to establish CK duality, the singular contributions ought to be integrals of $Q$-exact terms as such terms will ultimately drop out due to the gauge invariance of the total scattering amplitudes. This was made explicit in~\cite{Ben-Shahar:2021doh}, where it was inductively proven for supersymmetric Yang--Mills theory and illustrative examples at low points were given. 

        Importantly, this conclusion also applies to the $\sfb$-operator~\eqref{eq:def_b_operator_10} in the $Y$-formalism. In fact, we first note that
        \begin{equation}
            \sfb\Psi\ =\ \sfb_{\text{Lorentz},\,\epsilon=0}\Psi
        \end{equation}
        for all representatives $\Psi$ of the $Q$-cohomology~\cite{Oda:2007ak}. Consequently, since in the covariant non-minimal formalism all the singular contributions to the total scattering amplitudes are $Q$-exact, the same holds true in the $Y$-formalism. Hence, we can employ the $Y$-formalism to compute  scattering amplitudes since all the potential singularities will sum into a $Q$-exact terms which, in turn, drop out due to gauge invariance. This is all that is needed to establish CK duality in the $Y$-formalism. 
    
        It is important to realise that the preceding argument does not prevent individual Feynman diagrams, and hence their numerators, from having singular contributions. In the covariant non-minimal formalism using $\sfb_{\text{Lorentz},\,\epsilon=0}$, these singularities could spoil CK duality~\cite{Ben-Shahar:2021doh}. On the other hand, we obtain $Y$-formalism Siegel gauge physical states (unintegrated vertex operators) by starting with the non-singular representatives $\lambda^A\lambda^B\caA_{AB}$ of the antifield cohomology classes and applying $\sfb$ to them~\cite{Aisaka:2009yp},
        \begin{equation}
            \Psi\ =\ \sfb(\lambda^A\lambda^B\caA_{AB})\ =\ -\frac{v_A\gamma^{M\,AB}D_B}{2\lambda^Av_A}\parder{x^M}(\lambda^C\lambda^D\caA_{CD})~.
        \end{equation}
        This implies that the singularities of external states and Feynman diagrams are of the form $\frac{1}{(\lambda^Av_A)^n}$. Furthermore, the kinematic Jacobi identities hold order by order in $\frac{1}{\lambda^Av_A}$ but they need to be regulated. As before, in total scattering amplitudes, divergent terms from each diagram will either cancel or combine into $Q$-exact terms and thus drop out in the end. Nevertheless, it is desirable to drop the singular terms in each individual diagram, before summing into a $Q$-exact term, so as to regulate the individual numerators in a minimal-subtraction-like scheme. To do so, one may worry that $Q$ could change the degree of divergence which, in turn, would imply that finite terms from each diagram might be needed to construct the ultimate singular $Q$-exact term. Then, when minimally subtracting the singular terms in each diagram individually, these finite terms would also need to be dropped. In turn, this would change the finite part of the numerators and so, spoil CK duality. However, this cannot happen since the operator $Q$ will not affect the degree of singularity near $\lambda^Av_A=0$ as it is independent of $v$. Consequently, the terms in the numerators that must be discarded may be restricted to singular terms only. To be explicit, we split each Feynman diagram, $\gamma_i$, into three terms,
        \begin{subequations}
            \begin{equation}
                \gamma_i\ =\ \gamma_i^{0}+\gamma_i^{Q,\text{finite}}+\gamma_i^{Q,\text{singular}}~, 
            \end{equation}
            where the finite and singular terms, $\gamma_i^{Q,\text{finite}}$ and $\gamma_i^{Q,\text{singular}}$, contribute to the $Q$-exact part of the total amplitude integrand
            \begin{equation}
                I\ =\ I^0+Q\Lambda~, 
            \end{equation}
            where
            \begin{equation}
                Q\Lambda\ =\ \sum_i\left(\gamma_i^{Q,\text{finite}}+\gamma_i^{Q,\text{singular}}\right).
            \end{equation}
        \end{subequations}
        Since, as mention above, $Q$ preserves the degree of singularity near $\lambda^Av_A=0$, both sums $\sum_i\gamma_i^{Q,\text{finite}}$ and $\sum_i\gamma_i^{Q,\text{singular}}$ are separately $Q$-exact. This implies that we can drop $\gamma_i^{Q, \text{singular}}$ in each diagram separately, while preserving the total scattering amplitude. Since CK duality holds order by order in $\frac{1}{(\lambda^A v_A)^n}$, the resulting `minimally-subtracted' numerators obey the kinematic Jacobi identities. 
    
        In summary, we can truncate away the singular terms in the numerators without losing kinematic Jacobi identities, akin to minimal subtraction in dimensional regularisation. The minimally subtracted numerators provide a CK-dual parametrisation of the scattering amplitudes with finite numerators.
    
        Therefore, we have established all-order tree-level CK duality for ten-dimensional supersymmetric Yang--Mills theory. By dimensional reduction and embedding non-maximally supersymmetric Yang--Mills tree diagrams into maximal ones (cf.~\cite{Chiodaroli:2013upa}), this establishes tree-level CK-duality for all pure Yang--Mills theories with arbitrary amounts of supersymmetry in any dimension.

        Finally, let us remark that had we used the usual, covariant operator~\eqref{eq:def_b_Lcov_10}, our argument would not have worked. In this case, the ultraviolet divergences arise at the tip of the cone $\lambda^A\bar\lambda_A=0$ in pure spinor space, but $Q$ does change the degree of singularity near $\lambda^A\bar\lambda_A=0$ due to the derivative with respect to $\bar\lambda_A$. This leads to a potential mixing of singularities, and therefore CK duality is not guaranteed order by order. In this case, there is no subtraction scheme as for the $\sfb$-operator in the $Y$-formalism.        
        
        \paragraph{Double copy.}
        The $\BVbox$-algebra obtained above can be double-copied using our formalism in a straightforward manner. We choose to work with the evident cocommutative Hopf algebra $\frH_{\IM^{10}}$ to control the momentum dependence. Correspondingly, we use the restricted tensor product 
        \begin{equation}
            \hat\frB\ \coloneqq\ \frB^\text{psSYM}\otimes^{\frH_{\IM^{10}}}\frB^\text{psSYM}~.
        \end{equation}
        Upon factorising the pure spinor space for supersymmetric Yang--Mills theory as
        \begin{equation}
            \scM_{\text{10d}\,\caN=1}\ \coloneqq\ \IM^{10|16}\times \scM^\text{ps}_{\text{10d}\,\caN=1}~,
        \end{equation}
        we find that the graded vector space underlying $\hat\frB$ is simply
        \begin{equation}\label{eq:doubled_space_SYM}
            \scC^\infty(\IM^{10|32}\times \scM^\text{ps}_{\text{10d}\,\caN=1}\times\scM^\text{ps}_{\text{10d}\,\caN=1})~,
        \end{equation}
        and we note that both the odd superspace coordinates $\theta$ as well as all the auxiliary coordinates $\lambda^A$, $\bar\lambda_A$, and $\rmd\bar\lambda_A$ get doubled. In this larger space, we now have to consider the kernel of $\hat\sfb_-=\sfb\otimes \sfid-\sfid\otimes \sfb$, 
        \begin{equation}
            \ker(\hat\sfb_-)\ =\ \left\{\,f\in\scC^\infty\left(\IM^{10|32}\times \scM^\text{ps}_{\text{10d}\,\caN=1}\times\scM^\text{ps}_{\text{10d}\,\caN=1}\right)~\middle|~(\sfb\otimes \sfid) f=(\sfid\otimes \sfb) f\,\right\}~,
        \end{equation}
        which underlies the restricted kinematic Lie algebra $\frKin^0(\hat \frB)$. This turns out to be a metric dg Lie algebra, and the resulting action principle reads as 
        \begin{equation}\label{eq:ps_double_copy_SYM}
            S\ \coloneqq\ \int\Omega_{\text{10d}\,\caN=1}\wedge_{\IM^{10|16}}\Omega_{\text{10d}\,\caN=1}\inner{\Psi}{(Q\otimes \sfid+\sfid\otimes Q)\Psi+\tfrac13[\Psi,\Psi]}_{\frKin^0(\hat \frB)}~,
        \end{equation}
        where $\Omega_{\text{10d}\,\caN=1}\wedge_{\IM^{10}}\Omega_{\text{10d}\,\caN=1}$ denotes the evident integral on the space~\eqref{eq:doubled_space_SYM} (where we have again removed the infinite volume factor from the additional integral over the second copy of $\IM^{10}$, cf.~the discussion in \cref{ssec:syngamy_pure_gauge}).
        
        We regard our cubic double-copied action~\eqref{eq:ps_double_copy_SYM} as a rather exciting new result in the pure spinor formulation of supergravity. In eleven dimensions, the currently available action contains quartic terms in the pure spinor field~\cite{Cederwall:2009ez,Cederwall:2010tn}, see also~\cite{Grassi:2023yfe} and reference therein for more recent work using integral forms. In ten dimension, a pure spinor formulation of the vertex operators of closed superstrings was given in~\cite{Grassi:2004ih}, cf.~also~\cite{Chandia:2019klv}. These are precisely the double copy without the restriction to $\ker(\hat\sfb_-)$ (which would amount to imposing the section condition), and hence the field content is initially too large. In~\cite{Grassi:2004ih}, a different solution to this problem has been proposed, but this does not allow for the direct link between world-sheet ghost number and target-space ghost number that we observe in our prescription; also, it would lead to a non-cubic action. Hence, to our knowledge,~\eqref{eq:ps_double_copy_SYM} presents the first cubic form of a pure spinor action for ten-dimensional supergravity. Further study of this action is certainly warranted, in particular regarding the link to the pure spinor formulations of open and closed string, but this has to be left to future work.
        
        \subsection{Pure spinor formulation of M2-brane models}\label{ssec:M2branes}
        
        \paragraph{Pure spinor space.}
        The Bagger--Lambert--Gustavsson (BLG) M2-brane model~\cite{Bagger:2007jr,Gustavsson:2007vu} can also be formulated as a Chern--Simons--matter theory on pure spinor spaces~\cite{Cederwall:2008xu}.
        
        Here, we start from the space
        \begin{equation}
            \hat\scM_{\text{3d}\,\caN=8}\ \coloneqq\ \IM^{3|16}\times (\IR^{2|1}\otimes\caS_{\text{10d}\,\text{MW}})~,
        \end{equation}
        where $\IM^{3|16}$ is the three-dimensional $\caN=8$ Minkowski superspace and $\caS_{\text{10d}\,\text{MW}}$ again the space of Majorana--Weyl spinors in ten-dimensional, but now with indices reflecting the branching $\sfSpin(1,9)\rightarrow \sfSpin(1,2)\times \sfSpin(7)$. Explicitly, $\IR^{2|1}\otimes\caS_{\text{10d}\,\text{MW}}$ is coordinatised by $(\lambda^{\alpha i},\bar \lambda_{\alpha i},\rmd \bar \lambda_{\alpha i})$ with $\alpha=0,\ldots 2$ and $i=1,\ldots,8$, transforming in the $\mathbf{2}\otimes\mathbf{8}$, $\mathbf{2}\otimes\bar{ \mathbf{8}}$, and $\mathbf{2}\otimes\bar{\mathbf{8}}$ of $\sfSpin(1,2)\times\sfSpin(7)$. Note that indices in the $\mathbf{2}$ are raised and lowered as usual with $\eps_{\alpha\beta}$ and its inverse. Also, the R-symmetry group is enlarged from $\sfSpin(7)$ to $\sfSpin(8)$, and we use indices $m,n=1,\ldots,8$ for the vector representation $\mathbf{8_v}$ of $\sfSpin(8)$.
        
        The pure spinor space $\scM_{\text{3d}\,\caN=8}$ is then the quadric in $\hat \scM_{\text{3d}\,\caN=8}$ with the following relations:
        \begin{equation}\label{eq:pure_constraint_3d}
            \lambda^{\alpha i}\gamma_{\alpha\beta}^\mu\lambda_i^\beta\ =\ \bar\lambda^{\alpha i}\gamma_{\alpha\beta}^\mu\bar\lambda^\beta_i\ =\ \bar\lambda^{\alpha i}\gamma_{\alpha\beta}^\mu\rmd\bar\lambda^\beta_i\ =\ 0~,
        \end{equation}
        where $\gamma^\mu_{\alpha\beta}$ are the generators of the Clifford algebra of $\sfSpin(1,2)$. 
        
        Together with the supersymmetric covariant derivatives $D_{i\alpha}$ which satisfy the relations
        \begin{equation}
            \{D_{i\alpha},D_{j\beta}\}\ =\ \gamma_{\alpha\beta}^\mu \delta_{ij}\parder{x^\mu}~,
        \end{equation}
        we have a natural vector field $Q$ on $\hat\scM_{\text{3d}\,\caN=8}$,
        \begin{equation}\label{eq:Q-3d}
            Q\ \coloneqq\ \lambda^{\alpha i}D_{\alpha i}+\rmd \bar \lambda_{\alpha i}\parder{\bar \lambda_{\alpha i}}~,
        \end{equation}
        which descends to a differential on functions on $\scM_{\text{3d}\,\caN=8}$.
        
        Again, there is a family of operators $\sfb$ satisfying\footnote{albeit the covariant form has not been constructed so far}
        \begin{equation}\label{eq:b-relations-2}
            \sfb^2\ =\ 0
            \eand
            Q\sfb+\sfb Q\ =\ \wave~,
        \end{equation}
        and we choose to work again with the evident operator arising in the $Y$-formalism,
        \begin{equation}\label{eq:def_b_3}
            \sfb\ \coloneqq\ -\frac{v_{\alpha i}\gamma^{\mu\,\alpha\beta}\delta^{ij}D_{\beta j}}{2\lambda^{\alpha i} v_{\alpha i}}\parder{x^\mu}~,
        \end{equation}
        where $v$ is a reference pure spinor $v$ with $v_{\alpha i} \gamma^{\mu\,\alpha\beta}\delta^{ij}v_{\beta j}=0$. A short computation verifies~\eqref{eq:b-relations-2}. 
        
        We have summarised the properties of the above objects in \cref{tab:coordinatesOperators2}.
        
        \begin{table}[ht]
            \vspace{15pt}
            \begin{center}
                \begin{tabular}{ccccc}
                    \toprule
                    & \multirow{2}{*}{$\sfSL(2,\IR)\times\sfSpin(8)$} & mass & Gra{\ss}mann & ghost
                    \\[-5pt]
                    & & dimension & degree & number
                    \\
                    \midrule
                    $x$ & $\mathbf{(3,8_v)}$ & $-1\phantom+$ & $0$ & $\phantom{+}0\phantom+$
                    \\
                    $\theta$ & $\mathbf{(2,8_s)}$ & $-\frac12\phantom+$ & $1$ & $\phantom+0\phantom+$
                    \\
                    $\lambda$ & $\mathbf{(2,8_s)}$ & $-\frac12\phantom+$ & $0$ & $\phantom+1\phantom+$
                    \\
                    $\bar\lambda$ & $\mathbf{(2,8_c)}$ & $\phantom{+}\frac12\phantom+$ & $0$ & $-1\phantom+$
                    \\
                    $\rmd\bar\lambda$ & $\mathbf{(2,8_c)}$ & $\phantom{+}\frac12\phantom+$ & $1$ & $\phantom{+}0\phantom+$
                    \\
                    \midrule
                    $D$ & $\mathbf{(2,8_s)}$ & $\phantom{+}\frac12\phantom+$ & $1$ & $\phantom{+}0\phantom+$
                    \\
                    $Q$ & $\mathbf{(1,1)}$ & $\phantom{+}0\phantom+$ & $1$ & $\phantom{+}1\phantom+$\\
                    $\sfb$ & $\mathbf{(1,1)}$ & $\phantom{+}2\phantom+$ & $1$ & $-1\phantom+$
                    \\
                    \bottomrule
                \end{tabular}
                \caption{Properties of three-dimensional coordinates and operators.}
                \label{tab:coordinatesOperators2}
            \end{center}
        \end{table}        
        
        \paragraph{Gauge algebra.}
        Recall that the BLG model has an underlying metric 3-Lie algebra in the sense of~\cite{Filippov:1985aa}. Such a 3-Lie algebra can be seen as a Lie algebra with an orthogonal representation~\cite{deMedeiros:2009hf}. In the case of the BLG model, the Lie algebra is $\frg=\frsu(2)\oplus\frsu(2)$ and the orthogonal representation is Euclidean $\IR^4$. Concretely, we can identify $\frg\cong V\wedge V$ with $V\coloneqq\IR^4$, and with respect to the standard basis $\sfe_k$, $k=1,\ldots,4$ on $\IR^4$, we have a ternary bracket
        \begin{equation}
            [\sfe_{k_1},\sfe_{k_2},\sfe_{k_3}]_V\ \coloneqq\ \eps_{k_1k_2k_3k_4}\sfe_{k_4}
        \end{equation}
        with $\eps_{k_1k_2k_3k_4}$ the Levi-Civita symbol,
        and the metric
        \begin{equation}
            \inner{\sfe_{k_1}}{\sfe_{k_2}}_{V}\ \coloneqq\ \delta_{k_1k_2} 
        \end{equation}
        with $\delta_{k_1k_2}$ the Kronecker symbol. These define a metric Lie algebra $\frg$ by the relations
        \begin{equation}
            \begin{aligned}
                (\sfe_{k_1}\wedge\sfe_{k_2})\acton\sfe_{k_3}\ &=\ [\sfe_{k_1},\sfe_{k_2},\sfe_{k_3}]_V~,
                \\
                \inner{\sfe_{k_1}\wedge\sfe_{k_2}}{\sfe_{k_3}\wedge\sfe_{k_4}}_\frg\ &=\ \inner{\sfe_{k_3}}{[\sfe_{k_1},\sfe_{k_2},\sfe_{k_4}]_V}_V~,
            \end{aligned}
        \end{equation}
        and we find $\frg\cong\frsu(2)\oplus\frsu(2)$ as a Lie algebra with an indefinite metric of signature $(+,+,+,-,-,-)$.
        
        \paragraph{Field content and action.}
        For the BLG model, the formalism presented in~\cite{Cederwall:2008xu} uses two fields. Firstly, there is a scalar superfield $\Psi$ on $\scM_{\text{3d}\,\caN=8}$ of mass dimension~$0$, Gra{\ss}mann degree~$1$, and ghost number $1$ taking values in the metric Lie algebra $\frg$, which encodes the gauge sector.
        
        The matter sector is a bit more subtle. There is a (trivial) $\sfSpin(8)$-bundle over $\scM_{\text{3d}\,\caN=8}$, and we can consider the associated vector bundle $E$ for the vector representation $\mathbf{8}_v$. From its sheaf of sections, we construct the quotient sheaf\footnote{Note that the sections can have singularities in $\IR^{2|1}\otimes\caS_{\text{10d}\,\text{MW}}$.}
        \begin{equation}
            \scE_{\scM_{\text{3d}\,\caN=8}}\ \coloneqq \ \Gamma(E)/\caI_E~,
        \end{equation}
        where $\caI_E$ is the ideal generated by $\lambda^{\alpha i}\gamma^m_{ij}\vartheta^j_\alpha$ where $\vartheta^j_\alpha$ is an arbitrary function of ghost degree~$-1$ and $\gamma^m_{ij}$ are the $\sfSpin(8)$-factor of the Clifford algebra generators for $\sfSpin(1,2)\times\sfSpin(8)$. The matter fields $\Phi^m$ are now elements of $\scE_{\scM_{\text{3d}\,\caN=8}}$ with values in $V$. Operationally, we can regard them as sections of $E$ (with values in $V$) subject to the identification
        \begin{equation}
            \Phi^m\ \sim\ \Phi^m+\lambda^{\alpha i}\gamma^m_{ij}\vartheta^j_\alpha~.
        \end{equation}
        
        We note that there is a natural pairing on $\scE_{\scM_{\text{3d}\,\caN=8}}$ given by
        \begin{equation}
            g_{mn}\Phi^m\Phi^n
            \ewith
            g_{mn}\ \coloneqq\ \lambda^{\alpha i}\gamma_{mn\,ij}\lambda^j_{\alpha}~.
        \end{equation}
        
        The pure spinor superspace $\scM_{\text{3d}\,\caN=8}$ comes with a natural dimensionless volume form $\Omega_{\text{3d}\,\caN=8}$~\cite{Cederwall:2008xu}, and we can formulate the action
        \begin{equation}\label{eq:BLGAction}
            S^{\text{3d}\,\caN=8}\ \coloneqq\ \int\Omega_{\text{3d}\,\caN=8}\Big\{\inner{\Psi}{Q\Psi+\tfrac13[\Psi,\Psi]}_\frg+g_{mn}\inner{\Phi^m}{Q\Phi^n+\Psi\Phi^n}_V\Big\}\,.
        \end{equation}
        
        \paragraph{$\BVbox$-algebra and -module structure.}
        Our remaining constructions can now follow fully analogous to the case of supersymmetric Yang--Mills theory, except for the fact that we are dealing with a $\BVbox$-algebra module. The $\BVbox$-algebra itself is given by 
        \begin{equation}
            \frB^\text{psM2}\ \coloneqq\ (\scC^\infty(\scM_{\text{3d}\,\caN=8}),Q,-\cdot-,\sfb)
        \end{equation}
        with $-\cdot-$ the pointwise product and the $Y$-formalism $\sfb$-operator~\eqref{eq:def_b_3}, which is evidently of second order. The relevant module $V^\text{psM2}$ is given by 
        \begin{equation}
            V^\text{psM2}\ \coloneqq\ (\scE_{\text{3d}\,\caN=8},Q,-\cdot-,\sfb)~,
        \end{equation}
        where the actions of $Q$ and $\sfb$ are the evident ones, induced by the operators~\eqref{eq:Q-3d} and~\eqref{eq:def_b_3} on $\scE_{\text{3d}\,\caN=8}$, respectively, and $-\cdot-$ is again the pointwise product. The fact that $V^\text{psM2}$ is a module over $\frB^\text{psM2}$ is self-evident. 
        
        The $\BVbox$-algebra and -module structure $(\frB^\text{psM2},V^\text{psM2})$ guarantees CK duality on the field theories currents~\cite{Borsten:2023reb}. Moreover, the same arguments as for supersymmetric Yang--Mills theories lift this CK duality to the tree-level amplitudes. Singularities in the integrand are either IR-type singularities, which can be regulated in an evident form, or they are of the form $\frac{1}{\lambda^{\alpha i}v_{\alpha i}}$, and then, because of our use of the $Y$-formalism $\sfb$-operator, there is a minimal subtraction scheme allowing us to extract finite CK-dual numerators for the tree-level amplitudes of the M2-brane model. 
        
        While the pure-spinor-based proof of CK duality of the tree-level amplitudes of supersymmetric Yang--Mills theory was an alternative proof, this proof for tree-level CK duality in BLG models is the first; only partial results were available in the literature previously, cf.~\cite{Huang:2012wr,Huang:2013kca,Sivaramakrishnan:2014bpa,Ben-Shahar:2021zww}. The relation of our notion of CK duality, the conventional one for gauge--matter theory, and the quartic CK duality of~\cite{Huang:2012wr,Huang:2013kca,Sivaramakrishnan:2014bpa} is explained in~\cite{Borsten:2023reb}. 
        
        \paragraph{Double copy.}
        The $\BVbox$-algebra and -module structure $(\frB^\text{psM2},V^\text{psM2})$ can now be straightforwardly double-copied, following our general formalism specialised to the evident cocommutative Hopf algebra $\frH_{\IM^{3}}$. The restricted tensor product leads again to a $\BVbox$-algebra and -module with
        \begin{equation}
            \hat\frB\ \coloneqq\ \frB^\text{psSYM}\otimes^{\frH_{\IM^{3}}}\frB^\text{psSYM}
            \eand
            \hat V\ \coloneqq\ V^\text{psSYM}\otimes V^\text{psSYM}~,
        \end{equation}
        and using the factorisations
        \begin{equation}
            \scM_{\text{3d}\,\caN=8}\ \coloneqq\ \IR^{3|16}\times \scM^\text{ps}_{\text{3d}\,\caN=8}
            \eand
            \scE_{\text{3d}\,\caN=8}\ \coloneqq\ \scE_\text{3d}\otimes\scE^\text{ps}_{\text{3d}\,\caN=8}~,
        \end{equation}
        we find that the graded vector spaces underlying $\hat \frB$ and $\hat V$ read as
        \begin{equation}
            \scC^\infty\left(\IR^{3|16}\times\scM^\text{ps}_{\text{3d}\,\caN=8}\times \scM^\text{ps}_{\text{3d}\,\caN=8}\right)
            \eand
            \scE_\text{3d}\otimes\scE^\text{ps}_{\text{3d}\,\caN=8}\otimes\scE^\text{ps}_{\text{3d}\,\caN=8}~.
        \end{equation}
        As expected both the odd superspace coordinates $\theta$ as well as all the pure spinor auxiliary coordinates $\lambda^A$, $\bar\lambda_A$, and $\rmd\bar\lambda_A$ get doubled. This larger space carries an action of the operator $\hat\sfb_-$, 
        \begin{equation}
            \ker(\hat\sfb_-)\ =\ \left\{\,f\in\scC^\infty\left(\IR^{3|16}\times\scM^\text{ps}_{\text{3d}\,\caN=8}\times \scM^\text{ps}_{\text{3d}\,\caN=8}\right)\oplus \scE_\text{3d}~\middle|~\,(\sfb\otimes \sfid) f=(\sfid\otimes \sfb)\,\right\},
        \end{equation}
        to which the kinematic Lie algebra of $\hat\frB$ can be truncated. The result is another cubic action of the form~\eqref{eq:ps_double_copy_SYM}, which we would expect to describe $\caN=16$ supergravity in three dimensions, cf.~\cite{Bargheer:2012gv,Huang:2012wr}. Studying the resulting action in detail is, however, beyond the scope of this paper, and we leave it to future work.
        
        \paragraph{Comment on the ABJ(M) models.}
        Both the Aharony--Bergman--Jafferis--Maldacena (ABJM) model~\cite{Aharony:2008ug} and the Aharony--Bergman--Jafferis (ABJ) model~\cite{Aharony:2008gk} can also be formulated in the pure spinor formalism of~\cite{Cederwall:2008xu}. The pure spinor superspace with $m,n=1,\ldots,4$ for these theories is obtained from the pure spinor space of the BLG model, $\scM_{\text{3d}\,\caN=8}$, by truncating the R-symmetry $\sfSpin(8)$ to $\sfSpin(6)$. It not difficult to adjust the action principal for the BLG model to this situation.
        
        There is, however, a technical complication compared to the BLG model: the representation space $V$ in the underlying $\BVbox$-module is complex, as explained in~\cite{Borsten:2023reb}, and therefore there is no suitable symplectic metric on the underlying vector space. While this is not a fundamental issue for discussing CK duality, it significantly complicates all constructions. We therefore refrain from giving the details here; the $\BVbox$-algebra and -module structure can be found in our paper~\cite{Borsten:2023reb}.
        
        \appendix
        \addappheadtotoc
        \appendixpage 
        
        \appendices
        
        \section{Restricted tensor product of modules over bialgebras }\label{app:restricted_tensor_product}
        
        Throughout this section, we use the Sweedler notation~\eqref{eq:Sweedler_notation}, and we fix a bialgebra $\frH$ over a field $\IK$ of arbitrary characteristic; see \cref{def:HopfAlgebra}. Furthermore, we view $\IK$ as the canonical $\frH$-module in which $\frH$ acts via the counit $\epsilon\,:\,\frH\rightarrow\IK$. 
        
        \begin{definition}
            Let $V$ and $W$ be $\frH$-modules. We call the subset
            \begin{equation}\label{eq:reducedTensorProduct}
                V\otimes^\frH W\ \coloneqq\ \bigcap_{\chi\in\frH}\ker\big((\chi\otimes\unit-\unit\otimes\chi)\acton\big)
            \end{equation}
            of $V\otimes W$ the \uline{restricted tensor product} of $V$ and $W$.
        \end{definition}
        
        The restricted tensor product forms an $\frH$-module under the following condition.
        \begin{definition}\label{def:restrictedly_tensorable}
            A bialgebra $\frH$ is \uline{restrictedly tensorable} if the left ideal of the unital associative algebra $\frH\otimes\frH$ generated by the subset
            \begin{equation}
                \Gamma\ \coloneqq\ \{\chi\otimes\unit-\unit\otimes\chi \mid \chi\in\frH\}
            \end{equation}
            is also a two-sided ideal.
        \end{definition}

        \begin{lemma}
            Let $\frH$ be a restrictedly tensorable bialgebra, and let $V$ and $W$ be $\frH$-modules. The restricted tensor product $V\otimes^\frH W$ is an $\frH$-submodule of $V\otimes W$.
        \end{lemma}
        
        \begin{proof}
            It suffices to see that, for arbitrary $u\in V\otimes^\frH W$ and $\chi_1,\chi_2\in\frH$, we have
            \begin{equation}
                (\chi_1\otimes\unit-\unit\otimes\chi_1)\Delta(\chi_2)u\ =\ 0~.
            \end{equation}
            Restricted tensorability implies that $(\chi_1\otimes\unit-\unit\otimes\chi_1)\Delta(\chi_2)$ is an element in the two-sided ideal left- and right-generated by $\Gamma$, and we can write this element as
            \begin{equation}
                (\chi_1\otimes\unit-\unit\otimes\chi_1)\Delta(\chi_2)
                \ =\ \sum_{i=1}^NX_i(\chi_{1,i}\otimes\unit-\unit\otimes\chi_{1,i})
            \end{equation}
            for some finite $N$ and $X^{(1)},\dotsc,X^{(N)}\in\frH\otimes\frH$ and $\chi_{1,1},\dotsc,\chi_{1,N}\in\frH$. It is now clear that the latter element of $\frH\otimes \frH$ annihilates all $u\in V\otimes^\frH W$.
        \end{proof}
        Examples of restrictedly tensorable bialgebras important to the discussion of CK duality and the double copy are primitively generated bialgebras. These are all bialgebras generated by a set of differential operators labelling momenta, together with their coproducts, a typical situation in a physical theory.
        
        Recall that an element $\chi$ in a bialgebra is \uline{primitive} if $\Delta(\chi)=\chi\otimes\unit+\unit\otimes\chi$, and a bialgebra is \uline{primitively generated} if it is generated as a unital associative algebra by its set of primitive elements; over a field of characteristic zero, it is a standard fact that a primitively generated Hopf algebra is isomorphic to the universal enveloping algebra of the Lie algebra of its primitive elements.
        \begin{lemma}
            Every primitively generated bialgebra $\frH$ is restrictedly tensorable.
        \end{lemma}\begin{proof}
            It suffices to show that, for every $\chi,\phi\in\frH$, we have
            \begin{equation}
                (\chi\otimes\unit-\unit\otimes\chi)\Delta(\phi)\ \in\ \frI~,
            \end{equation}
            where $\frI$ is the left ideal generated by the subset $\{\psi\otimes\unit-\unit\otimes\psi|\psi\in\frH\}$. This is evidently equivalent to showing that
            \begin{equation}\label{eq:induction_hypothesis}
                [\chi\otimes\unit-\unit\otimes\chi,\Delta(\phi)]\ \in\ \frI~,
            \end{equation}
            since $\Delta(\phi)(\chi\otimes\unit-\unit\otimes\chi)\in\frI$.
            
            By the assumption of primitive-generatedness, we may assume that $\phi$ is a linear combination of products of primitive elements.
            We proceed by induction. First, the base case: suppose that $\phi$ is primitive. Then it is immediate that
            \begin{equation}
                [\chi\otimes\unit-\unit\otimes\chi,\Delta(\phi)]
                \ =\ [\chi,\phi]\otimes\unit-\unit\otimes[\chi,\phi]\ \in\ \frI~.
            \end{equation}
            Next, suppose that we have shown~\eqref{eq:induction_hypothesis} in the case where $\phi$ is a linear combination of products of at most $n-1$ primitive elements. Now, suppose that $\phi=\phi_1\phi_2$ with $\phi_1$ primitive and $\phi_2$ a linear combination of at most $n-1$ primitive elements. Let
            \begin{equation}
                [\chi\otimes\unit-\unit\otimes\chi,\Delta(\phi_1)]
                \ =\ 
                \chi'\otimes\unit-\unit\otimes\chi'~.
			\end{equation}
            Then
            \begin{equation}
            \begin{aligned}
                &[\chi\otimes\unit-\unit\otimes\chi,\Delta(\phi)]\\
                &\kern.5cm\ =\
                [\chi\otimes\unit-\unit\otimes\chi,\Delta(\phi_1)\Delta(\phi_2)]\\
                &\kern.5cm\ =\
                [\chi\otimes\unit-\unit\otimes\chi,\Delta(\phi_1)]\Delta(\phi_2)+\Delta(\phi_1)[\chi\otimes\unit-\unit\otimes\chi,\Delta(\phi_2)]\\
                &\kern.5cm\ =\ (\chi'\otimes\unit-\unit\otimes\chi')\Delta(\phi_2)+\Delta(\phi_1)[\chi\otimes\unit-\unit\otimes\chi,\Delta(\phi_2)]\\
                &\kern.5cm\ =\ [\chi'\otimes\unit-\unit\otimes\chi',\Delta(\phi_2)]+\Delta(\phi_2)(\chi'\otimes\unit-\unit\otimes\chi')+\Delta(\phi_1)[\chi\otimes\unit-\unit\otimes\chi,\Delta(\phi_2)]~,
            \end{aligned}
            \end{equation}
            and now all three terms manifestly belong to the left ideal $\frI$.
        \end{proof}        
        \noindent
        The restricted tensor product is not, in general, symmetric nor associative up to isomorphism. Our construction generalises the familiar concept of the \uline{module of invariants}.
        
        \begin{proposition}
            For any $\frH$-module $V$, the restricted tensor product $V\otimes^\frH\IK$ is canonically isomorphic to $V^\frH\coloneqq\bigcap_{\chi\in\frH}\ker(\chi-\epsilon(\chi))$, which is called the \uline{module of invariants}.\footnote{Compare this to the well-known result that the module of coinvariants $V_\frH$ is given by $V_\frH\cong V\otimes_\frH\IK$.}
        \end{proposition}
        
        \begin{proof}
            It is simply a matter of unwinding the definition~\eqref{eq:reducedTensorProduct} to see that $V\otimes^\frH\IK\subseteq V\otimes\IK\cong V$ is given by $V^\frH$.
        \end{proof}
        
        \noindent 
        Thus, the restricted tensor product $\otimes^\frH$ is, in some sense, the dual of the tensor product $\otimes_\frH$ over the Hopf algebra $\frH$: $V\otimes^\frH W$ is a submodule of $V\otimes W$, whereas $V\otimes_\frH W$ is a quotient of $V\otimes W$.
        
        \begin{proposition}
            Suppose that $V$ and $W$ are $\frH$-modules equipped with $\frH$-linear maps $f\colon V^{\otimes n}\rightarrow V$ and $g\colon W^{\otimes n}\rightarrow W$. Then, the $\frH$-linear map $f\otimes g\colon(V\otimes W)^{\otimes n}\rightarrow V\otimes W$ restricts to an $\frH$-linear map $f\otimes^\frH g\colon(V\otimes^\frH W)^{\otimes n}\rightarrow V\otimes^\frH W$.
        \end{proposition}
        
        \begin{proof}
            For clarity of exposition, we spell out the proof only for $n=2$; the other cases generalise straightforwardly. Given $u^{(1)}_{1,2}\otimes u^{(2)}_{1,2}\in V\otimes W$, then
            \begin{equation}\label{eq:fTensorg}
                (f\otimes g)\big(u^{(1)}_1\otimes u^{(2)}_1,u^{(1)}_2\otimes u^{(2)}_2\big)\ =\ f\big(u^{(1)}_1,u^{(1)}_2\big)\otimes g\big(u^{(2)}_1,u^{(2)}_2\big)\,.
            \end{equation}
            Suppose now that $u^{(1)}_i\otimes u^{(2)}_i\in V\otimes^\frH W\subseteq V\otimes W$ for $i=1,2$ and let $\chi\in\frH$. Then, 
            \begin{equation}
                \begin{aligned}
                    \big(\chi f\big(u^{(1)}_1,u^{(1)}_2\big)\big)\otimes g\big(u^{(2)}_1,u^{(2)}_2\big)\ &=\ f\big(\chi^{(1)}u^{(1)}_1,\chi^{(2)}u^{(1)}_2\big)\otimes g\big(u^{(2)}_1,u^{(2)}_2\big)
                    \\
                    &=\ (f\otimes g)\big(\chi^{(1)}u^{(1)}_1\otimes u^{(2)}_1,\chi^{(2)}u^{(1)}_2\otimes u^{(2)}_2\big)
                    \\
                    &=\ (f\otimes g)\big(u^{(1)}_1\otimes\chi^{(1)}u^{(2)}_1,u^{(1)}_2\otimes\chi^{(2)}u^{(2)}_2\big)
                    \\
                    &=\ f\big(u^{(1)}_1,u^{(1)}_2\big)\otimes g\big(\chi^{(1)}u^{(2)}_1,\chi^{(2)}u^{(2)}_2\big)
                    \\
                    &=\ f\big(u^{(1)}_1,u^{(1)}_2\big)\otimes\big(\chi g\big(u^{(2)}_1,u^{(2)}_2\big)\big)\,,
                \end{aligned}
            \end{equation}
            where in the second and fourth steps we have used~\eqref{eq:fTensorg}, and in third step the assumption $\chi^{(i)}u^{(1)}_i\otimes u^{(2)}_i=u^{(1)}_i\otimes\chi^{(i)}u^{(2)}_i$ for $i=1,2$. Hence, $f\otimes g\colon(V\otimes W)^2\rightarrow V\otimes W$ in~\eqref{eq:fTensorg} restricts to a map $f\otimes^\frH g\colon(V\otimes^\frH W)^2\rightarrow V\otimes^\frH W$.
        \end{proof}
        
        \noindent 
        This proposition now implies that given two $\frH$-algebras, that is, $\frH$-modules $V$ equipped with $\frH$-linear $n$-ary algebraic operations $V^{\otimes n}\rightarrow V$, their restricted tensor naturally inherits a corresponding algebra structure.
        
        \section{Analytical settings via convolutions}\label{app:compactification-setting}

        As briefly remarked in the introduction, in the case where the Hopf algebra is commutative, to construct the double copy, instead of working with the restricted tensor product $\otimes^\frH$ of \cref{app:restricted_tensor_product}, we can instead work with a tensor product $\otimes_\frH$ over the commutative Hopf algebra, which corresponds to the convolutional double copy of~\cite{Anastasiou:2014qba,Anastasiou:2018rdx, Monteiro:2021ztt,Luna:2022dxo}. This approach runs into analytical difficulties because plane wave states cannot be convolved (or, equivalently, delta functions in momentum space cannot be squared). One can circumvent this either by compactifying space--time as in \cref{ssec:compactification-setting} or by complicating the notion of Hopf algebras as in \cref{ssec:generalised-hopf-setting}. 
        
        \subsection{Analytical setting via compactification}\label{ssec:compactification-setting}

        In this section, we provide a proof of the statement that compactifying space-time provides an analytical setting using the tensor product over the Hopf algebra. In the following, the metric signature is irrelevant. We compactify $\IM^d$ to $\IM^d/\Lambda\IZ^d$; without loss of generality, we may work with units where $\Lambda=1$.
        
        Let $\scS$ be the subspace of $\scC^\infty(\IM^d/\IZ^d,\IC)$ consisting of finite linear combinations of plane waves, i.e.~smooth functions whose Fourier series' supports are finite sets. This is dense inside $L^2(\IM^d/\IZ^d,\IC)$ in the $L^2$-norm topology as well as inside $\scC^\infty(\IM^d/\IZ^d,\IC)$ in the Fr{\'e}chet topology, since Fourier series of smooth functions converge pointwise and hence, uniformly.
        
        \begin{proposition}\label{thm:compactification-setting-justification}
            As modules over $\frH_{\IM^d}$ of differential operators with constant coefficients discussed in \cref{ex:H-Box-Minkowski}, we have
            \begin{subequations}
                \begin{equation}
                    \scS\otimes_{\IC[\partial_\mu]}\scS\ \cong\ \scS
                \end{equation}
                by means of the convolution
                \begin{equation}\label{eq:convolution}
                    f\ =\ f^{(1)}\otimes f^{(2)}\ \mapsto\ f^{(1)}\star f^{(2)}
                \end{equation}
                for all $f\in\scS\otimes_{\IC[\partial_\mu]}\scS$.
            \end{subequations}
        \end{proposition}
        
        \begin{proof}
            Let us first show injectivity of~\eqref{eq:convolution}. Suppose that $f,g,h\in\scS$. We wish to show that
            \begin{equation}\label{eq:tensorConvolution}
                f\otimes(g\star h)\ =\ (f\star g)\otimes h
            \end{equation}
            in the tensor product $\scS\otimes_{\IC[\partial_\mu]}\scS$. If this holds, then from $f_1\star g_1=f_2\star g_2$, we get $f_1\otimes g_1=f_1\otimes(g_1\star\id_\scC)=(f_1\star g_1)\otimes\id_\scC=(f_2\star g_2)\otimes\id_\scC=f_2\otimes g_2$, that is injectivity of~\eqref{eq:convolution}. To verify~\eqref{eq:tensorConvolution}, let $K\coloneqq\operatorname{supp}(\hat f)\cup\operatorname{supp}(\hat g)\cup\operatorname{supp}(\hat h)\subseteq\IZ^d$ be the union of the supports of the Fourier transforms $\hat f$, $\hat g$, and $\hat h$ of all three functions $f$, $g$, and $h$; let $\delta_K\in\scS$ be an approximation of the Dirac comb on $K$, namely
            \begin{equation}
                \delta_K(x)\ \coloneqq\ \sum_{k\in K}\rme^{\rmi k\cdot x}~.
            \end{equation}
            It is a convolutional idempotent, that is, $\delta_K\star\delta_K=\delta_K$.
            By multivariate polynomial interpolation in Fourier space, we can find (not necessarily unique) differential operators $D_f,D_g,D_h\in\IC[\partial_1,\dotsc,\partial_d]$ such that $f=D_f\delta_K$, $g=D_g\delta_K$, and $h=D_h\delta_K$. Then it is clear that, inside $\scS\otimes_{\IC[\partial_\mu]}\scS$, we have
            \begin{equation}
                \begin{aligned}
                    f\otimes(g\star h)\ &=\ D_f\delta_K\otimes(D_g\delta_K\star D_h\delta_K)
                    \\
                    &=\ D_f\delta_K\otimes D_gD_h\delta_K
                    \\
                    &=\ D_fD_g\delta_K\otimes D_h\delta_K
                    \\
                    &=\ (D_f\delta_K\star D_g\delta_K)\otimes D_h\delta_K
                    \\
                    &=\ (f\star g)\otimes h~.
                \end{aligned}
            \end{equation}
            
            Having shown injectivity of~\eqref{eq:convolution}, surjectivity is now straightforward: for any $f\in\scS$ we have $f=f\star\delta_{\operatorname{supp}(\hat f)}$.
        \end{proof}
        
        \subsection{Analytical setting via generalisations of Hopf algebras}\label{ssec:generalised-hopf-setting}

        In this section, we provide an analytical setting for the double copy using the tensor product over the Hopf algebra where compactification of space-time is not needed, at the cost of having to work with an algebra of pseudo-differential operators that does not form a Hopf algebra any more. \footnote{We thank an anonymous user on MathOverflow for their help.}
        
        The physical metric signature is irrelevant for the following analytical considerations, but for analytical considerations it is convenient to use an auxiliary positive-definite metric on space-time $\IR^d$. We emphasise that this does not really pertain to the physics but is only used in the course of the mathematical proofs.
        
        Since we are going beyond the usual setting of Hopf algebras, we spell out precisely what we achieve.

        \begin{definition}\label{def:convenientAnalyticalSetting}
            A \uline{convenient analytical setting} consists of
            \begin{enumerate}[(i)]\itemsep-4pt
                \item a function space $\scC\subseteq\scC^\infty(\IR^d,\IC)$ and
                \item a space of (pseudo-)differential operators $\scD\subseteq\IC[\![\partial_1,\dotsc,\partial_d]\!]$
            \end{enumerate}
            such that
            \begin{enumerate}[(i)]\itemsep-4pt
                \item $\scC$ is closed under pointwise products; thus, it forms a non-unital commutative associative algebra;
                \item $\scD$ is closed under composition and contains $1$; thus it forms a unital commutative ring;
                \item $\scD$ contains $\IR[\partial_\mu]=\frH_{\IR^d}$ as a subring;
                \item $\scD\scC\subseteq\scC$; thus, $\scC$ is a module over $\scD$;
                \item an analytic Leibniz rule holds in the sense that, for $D=\sum_{I\in\IN^d}c_I\partial_I\in\scD$ and $f,g\in\scC$, we have
                \begin{equation}
                    D(f\cdot g)\ =\ \sum_{I,I'\in\IN^d}\binom{I+I'}Ic_{I+I'}(\partial_If)\cdot(\partial_{I'}g)
                \end{equation}
                in the topology of pointwise convergence on some neighbourhood of the origin in Fourier space;
                \item $\scC\otimes_\scD\scC=\scC$;
                \item $\scC$ is dense inside $\scC^\infty(\IR^d,\IC)$ with respect to the Fr\'{e}chet space topology (i.e.\ topology of uniform convergence on compact sets).
            \end{enumerate}
        \end{definition}
        
        \noindent
        We further define the following tube domain of the real hyperplane:
        \begin{equation}
            \IR^d_\epsilon\ \coloneqq\ \{x+\rmi y\,|\,x,y\in\IR^d\text{ and }\|y\|<\epsilon\}\ \subseteq\ \IC^d~.
        \end{equation}
        Define $\scC_0$ to be the space of functions $f\colon\IR^d\to\IC$ such that there exist $\epsilon,\delta>0$ such that $f$ extends analytically to $\IR^d_\epsilon$ with
        \begin{equation}
            |f(x+\rmi y)|\ =\ \scO(\rme^{-\delta\|x\|})~.
        \end{equation}
        Define $\scC$ as
        \begin{equation}\label{eq:generalised-hopf-function-space}
            \scC\ \coloneqq\ \{f\in\scC^\infty(\IR^d,\IC)\,|\,\partial_If\in\scC_0\}~,
        \end{equation}
        i.e.~the space of functions whose arbitrary-order derivatives lie in $\scC_0$.
        
        \begin{lemma}
            If $f\in\scC_0$ is holomorphic on $\IR^d_\epsilon$ and $|f(x+\rmi y)|\leq C\rme^{-\delta\|x\|}$, then $\hat f$ is holomorphic on $\IR^d_\delta$ and $|\hat f(\xi+\rmi\eta)|\le C'\rme^{-\epsilon'\|\xi\|}/(\delta-\|\eta\|)^d$ for some constant $C'$. In particular, $\hat f\in\scC_0$ also; thus, the Fourier transform is an involution of $\scC_0$.
        \end{lemma}

        \begin{proof}
            To check holomorphicity of $\hat f$ on $\IR^d_\epsilon$, it suffices to check that the integral
            \begin{equation}
                \hat f(\xi+\rmi\eta)\ =\ \int_{\IR^d}\rmd^dx\,f(x)\,\rme^{-\rmi(\xi+\rmi\eta)\cdot x}
            \end{equation}
            converges as long as $\|\eta\|<\delta$ so that we can take derivatives under the integral sign. But this is clear since $|f(x)|=\scO(\rme^{-\delta\|x\|})$.
            
            Furthermore, for arbitrary $y\in\IR^d$ with $\|y\|<\epsilon$, we can use Cauchy's integral theorem axis-by-axis to obtain the estimate
            \begin{equation}
                \begin{aligned}
                    |\hat f(\xi+\rmi\eta)|\ &=\ \left|\int_{\IR^d}\rmd^dx\,f(x)\,\rme^{-\rmi(\xi+\rmi\eta)\cdot x}\right|
                    \\
                    &=\ \left|\int_{\IR^d}\rmd^dx\,f(x+\rmi y)\,\rme^{-\rmi(\xi+\rmi\eta)\cdot(x+\rmi y)}\right|
                    \\
                    &\leq\ \int_{\IR^d}\rmd^dx\,C\rme^{-(\delta-\|\eta\|)\|x\|}\rme^{\xi\cdot y}
                    \\
                    &\leq\ \frac{C'}{(\delta-\|\eta\|)^d}\,\rme^{\xi\cdot y}~,
                \end{aligned}
            \end{equation}
            where $C'$ is a constant depending on $d$ only. By choosing $y=-\epsilon'\xi/\|\xi\|$ for arbitrary $0<\epsilon'<\epsilon$, we obtain $|\hat f(\xi+\rmi\eta)|\leq C'\rme^{-\epsilon'\|\xi\|}/(\delta-\|\eta\|)^d$. Taking the limit $\epsilon'\to\epsilon$, we obtain $|\hat f(\xi+\rmi\eta)\leq C'\rme^{-\epsilon'\|\xi\|}/(\delta-\|\eta\|)^d$.
        \end{proof}
        
        \begin{lemma}
            $\scC$ forms a non-unital subalgebra of $\scC^\infty(\IR^d,\IC)$.
        \end{lemma}

        \begin{proof}
            It is clear that $\scC$ is closed under sums and scalar multiplication. The only non-trivial thing to prove is closure under pointwise product. Let $f,g\in\scC$. Then
            \begin{equation}
                \partial_I(f\cdot g)\ =\ \sum_{\substack{I',I''\in\IN^d\\I'+I''=I}}\binom I{I'}(\partial_{I'}f)(\partial_{I''}g)\ \in\ \scC_0~.
            \end{equation}
            Hence $f\cdot g\in\scC_0$.
        \end{proof}
        
        Now, define $\scD_0$ to be the space of pseudo-differential operators of the form $p(\partial)$ where $p\in\scP_0$, where $\scP_0$ is the class of functions $p\colon\IR^d\to\IC$ such that there exists an $\epsilon>0$ such that $p$ extends analytically to $\IR^d_\epsilon$ and that, on this tube domain, for every $\delta>0$, there exists a $C_\delta>0$ such that
        \begin{equation}
            |p(x+\rmi y)|\ \leq\ C_\delta\rme^{\delta\|x\|}~.
        \end{equation}
        Define the ring $\scD$ as
        \begin{equation}\label{eq:generalised-hopf-definition}
            \scD\ \coloneqq\ \Bigg\{\sum_{i=1}^np_iq_i\,\Bigg|\,n\in\IN,\;p_1,\dotsc,p_n\in\scD_0,\;q_1,\dotsc,q_n\in\frH_{\IR^d}\Bigg\}\,,
        \end{equation}
        that is, the ring of pseudo-differential operators generated by $\scD_0$ and $\frH_{\IR^d}=\IR[\partial_1,\dotsc,\partial_d]$.
        
        \begin{lemma}\label{lem:C0-module}
            $\scC_0$ is a $\scD_0$-module.
        \end{lemma}\begin{proof}
            Let $f\in\scC_0$ and $D=p(\partial)$ with $p\in\scP_0$. Then $\hat f$ is analytic on $\IR^d_\epsilon$ and $|\hat f(\xi+\rmi\eta)|\leq C\rme^{-\delta\|\xi\|}$ for some $C,\epsilon,\delta>0$. Similarly, $p$ is analytic on $\IR^d_{\epsilon'}$. 
            
            Then, in Fourier space, the pointwise product $\hat fp$ is analytic on $\IR^d_{\min\{\epsilon,\epsilon'\}}$ and $|\hat fp|=\scO(\rme^{-\delta'\|\xi\|})$ for any $\delta'<\delta$. So $\hat fp\in\scC_0$, and hence $Df\in\scC_0$.
        \end{proof}
        
        \begin{lemma}\label{lem:C-module}
            $\scC$ is a $\scD$-module.
        \end{lemma}

        \begin{proof}
            It is clear that $\scC$ is closed under the action of $\IR[\partial_1,\dotsc,\partial_d]$ by construction. It remains to show that $\scC$ is a $\scD_0$-module.
            
            Let $f\in\scC$ and $D\in\scD_0$ and $I\in\IN^d$. It suffices to show that $\partial_IDf\in\scC_0$. But since $\partial_If\in\scC_0$, so $\partial_IDf=D(\partial_If)\in\scC_0$ (using \cref{lem:C0-module}).
        \end{proof}
        
        \begin{lemma}\label{lem:C-closed-under-pointwise-product}
            $\scC\star\scC=\scC$, where $\star$ denotes convolution.
        \end{lemma}

        \begin{proof}
            It is clear that $\scC_0\cdot\scC_0\subseteq\scC_0$. Since the Fourier transform is bijective on $\scC_0$, thus $\scC_0\star\scC_0\subseteq\scC_0$.
            
            Now, suppose that $f,g\in\scC$. Then $f\star g\in\scC_0$, and for any multi-index $I\in\IN^d$, we have $\partial_I(f\star g)=(\partial_If)\star g\subset\scC_0$. Hence $f\star g\in\scC$. Thus $\scC\star\scC\subseteq\scC$.
            
            It remains to show that $\scC\star\scC\supseteq\scC$. Given any $f\in\scC$, then we have
            \begin{subequations}
                \begin{equation}
                    f/u_\epsilon\cdot u_\epsilon\ =\ f~,
                \end{equation}
                where
                \begin{equation}
                    u_\epsilon(x+\rmi y)\ =\ \frac1{\prod_{i=1}^d\cosh(\epsilon(x_i+\rmi y_i))+\rmi}~.
                \end{equation}
            \end{subequations}
            Now, clearly $u_\epsilon\in\scC_0$, so the same holds for the Fourier transform $\hat u_\epsilon\in\scC_0$. Furthermore, for any polynomial $q$, clearly $pu_\epsilon\in\scC_0$ as well. Hence $\hat u_\epsilon\in\scC$. Similarly, if $|f|\leq C\rme^{-\delta\|x\|}$, then $\epsilon<\delta$ ensures that $pf/u_\epsilon\in\scC_0$ for any polynomial $p$; hence $\widehat{f/u_\epsilon}\in\scC_0$. Thus, $\hat f=\hat u_\epsilon\star\widehat{f/u_\epsilon}$, so that $\scC\star\scC\supseteq\scC$.
        \end{proof}
        
        \begin{theorem}\label{thm:generalised-hopf-setting-justification}
            $(\scC,\scD)$ is a convenient analytical setting.
        \end{theorem}

        \begin{proof}
            The numbering follows \cref{def:convenientAnalyticalSetting}.

            (i) is clear by construction. (ii) is also clear by construction, since composition amounts to pointwise products in Fourier space. (iii) is also clear by construction.
            
            (iv) was shown in \cref{lem:C-module}. (v) is clear by analyticity in Fourier space. As for (vi): it is clear that $\scC$ is a submodule of $\scD$ considered as a module over itself (since $\scC\subseteq\scC_0\subseteq\scP_0$). So $\scC\otimes_\scD\scC\subseteq\scC$. \Cref{lem:C-closed-under-pointwise-product} then implies that $\scC\otimes_\scD\scC=\scC$.
            
            (vii) It is clear that smooth functions with compact support are dense inside $\scC^\infty(\IR^d,\IC)$ (e.g.\ multiply by bump functions $\psi_m$ supported at $[-m-1,m+1])$ that are $1$ on $[-m,m]$). Suppose that $f$ is smooth with compact support. Then $f$ is the limit of convolutions $f\star\psi_m$ where $\psi_m=\prod_{i=1}^dm\rme^{-\pi m^2z_i^2}$ is a family of analytic functions with exponential falloff approximating the Dirac delta.
        \end{proof}
        
        \paragraph{Symmetric monoidal category.}
        Since $\scD$ is no longer a Hopf algebra, the category of arbitrary modules over $\scD$ no longer has a well defined tensor product (i.e.~does not form a symmetric monoidal category); in particular, double copy of arbitrary $\scD$-modules is not guaranteed to work. Instead, we single out a particular subcategory of the category of all $\scD$-modules that is closed under the tensor product.
        
        Consider the category $\operatorname{Mod}_{\scD,\mathrm{nice}}$ whose objects are $\sfD$-modules of the form
        \begin{equation}
            \bigoplus_{i=1}^K\overbrace{\scC\otimes_\IR\scC\otimes_\IR\dotsb\otimes_\IR\scC}^{n_i}~,
        \end{equation}
        where $n_i\in\IN$ and $K$ is a non-negative integer or $\infty$. These all have a canonical action of $\scD$ on them by virtue of the `infinitary Leibniz rule' defining an `infinitary coproduct'.
        
        Consider the full subcategory of the category of chain complexes of $\sfD$-modules consisting of those whose degreewise components all belong to $\operatorname{Mod}_{\sfD,\mathrm{nice}}$. This forms a symmetric monoidal category equipped with $\otimes_\IR$. In particular, we can define operads over this category.

        \section{Proofs by direct computation}\label{app:postponed_proofs}
        
        In this section, we collect mostly straightforward computational proofs omitted from the body of the paper.
        
        \paragraph{\cref{prop:colour-flavour-stripping}.} 
        It is clear, cf.~e.g.~the review in~\cite[Section 6]{Borsten:2021hua}, that $\frg\otimes \frC$ is a dg Lie algebra, and that $R\otimes V$ is a dg vector space with an action $\frg\otimes\frC\curvearrowright R\otimes V$. It is also well-known that a Lie algebra and a representation can be packaged into a Lie algebra with Lie bracket the semi-direct product. This extends to the differential graded setting. It remains to show that the given inner product is indeed cyclic, i.e.
        \begin{equation}
            \inner{\ell_1}{\mu_2(\ell_2,\ell_3)}_\frL\ =\ 
            (-1)^{|\ell_1|\,|\ell_2|+|\ell_1|\,|\ell_3|+|\ell_2|\,|\ell_3|}\inner{\ell_3}{\mu_2(\ell_1,\ell_2)}_\frL~.
        \end{equation}
        This is well-known to be the case for $\ell_1,\ell_2,\ell_3\in\frg\otimes\frC$. For $\ell_1,\ell_2,\ell_3\in R\otimes V$, both sides of the relation are trivial, and for $\ell_1\in\frg\otimes\frC$, $\ell_2,\ell_3\in R\otimes V$ (as well as cyclic permutations), cyclicity is ensured by~\eqref{eq:special_cyclicity}. Because of the lack of pairing between $R\otimes V$ and $\frg\otimes \frC$, both sides of the identity also vanish for $\ell_1\in R\otimes V$ and $\ell_2,\ell_3\in \frg\otimes \frC$. 
        
        \paragraph{\cref{prop:kinematic_module}.} 
        By direct computation, from~\cref{def:kinematicLieAlgebraPreBVBox} and Equations~\eqref{eq:kinematic_module_action} and~\eqref{eq:derivedBracketModuleJacobi}, we have
        \begin{equation}
            \begin{aligned}
                &[\phi_1[1],\phi_2[1]]_\frK\acton_\frV v[1]
                \\
                &\kern.5cm=\ (-1)^{|\phi_1|}\{\phi_1,\phi_2\}_\frB[1]\acton_\frV v[1]
                \\
                &\kern.5cm=\ (-1)^{|\phi_2|+1}\{\{\phi_1,\phi_2\}_\frB,v\}_{V}[1]
                \\
                &\kern.5cm=\ (-1)^{|\phi_1|+|\phi_2|}\left(\{\phi_1,\{\phi_2,v\}_{V}\}_{V}-(-1)^{(|\phi_1|+1)(|\phi_2|+1)}\{\phi_2,\{\phi_1,v\}_{V}\}_{V}\right)[1]
                \\
                &\kern.5cm=\ \phi_1[1]\acton_\frV(\phi_2[1]\acton_\frV v[1])-(-1)^{(|\phi_1|+1)(|\phi_2|+1)}\phi_2[1]\acton_\frV(\phi_1[1]\acton_\frV v[1])
            \end{aligned}
        \end{equation}
        for all $\phi_1,\phi_2\in\frB$ and $v\in V$, hence $(\frV,\acton_\frV)$ is a graded (left) module over the kinematic Lie algebra $(\frK,[-,-]_\frK)$.

        \paragraph{\cref{prop:BVbox_is_pre_BVbox}.} 
        Using the definition~\eqref{eq:derived_bracket} of the derived bracket and the associativity $\sfm_2(\sfm_2(\phi_1,\phi_2),\phi_3)=\sfm_2(\phi_1,\sfm_2(\phi_2,\phi_3))$ for all $\phi_{1,2,3}\in\frB$ of $\sfm_2$, it is easy to see that~\eqref{eq:BV_GB_Poisson} is, in fact, equivalent to~\eqref{eq:b_second_order}.
            
        To establish the shifted Jacobi identity~\eqref{eq:BV_GB_Jacobi}, we follow~\cite[Proposition 1.2]{Getzler:1994yd}. In particular, set
        \begin{equation}\label{eq:PoissonatorJacobiator}
            \begin{aligned}
                \sfPoiss(\phi_1,\phi_2,\phi_3)\ &\coloneqq\ \{\phi_1,\sfm_2(\phi_2,\phi_3)\}-\sfm_2(\{\phi_1,\phi_2\},\phi_3)
                \\
                &\kern1.5cm-(-1)^{(|\phi_1|+1)|\phi_2|}\sfm_2(\phi_2,\{\phi_1,\phi_3\})~,
                \\
                \sfJac(\phi_1,\phi_2,\phi_3)\ &\coloneqq\ \{\phi_1,\{\phi_2,\phi_3\}\}-(-1)^{|\phi_1|+1}\{\{\phi_1,\phi_2\},\phi_3\}
                \\
                &\kern1.5cm-(-1)^{(|\phi_1|+1)(|\phi_2|+1)}\{\phi_2,\{\phi_1,\phi_3\}\}~,
            \end{aligned}
        \end{equation}
        which we call the \uline{Poissonator} and the \uline{Jacobiator}, respectively. Then,
        \begin{equation}\label{eq:derivationPoissonImpliesJacobi}
            \begin{aligned} 
                &\sfJac(\phi_1,\phi_2,\phi_3)-\{\phi_1,\{\phi_2,\phi_3\}\}
                \\
                &\kern1cm=\ -(-1)^{|\phi_1|+1}\big[\sfb(\sfm_2(\{\phi_1,\phi_2\},\phi_3))-\sfm_2(\sfb(\{\phi_1,\phi_2\}),\phi_3)
                \\
                &\kern4cm-(-1)^{|\phi_1|+|\phi_2|+1}\sfm_2(\{\phi_1,\phi_2\},\sfb\phi_3)\big]
                \\
                &\kern2cm-(-1)^{(|\phi_1|+1)(|\phi_2|+1)}\big[\sfb(\sfm_2(\phi_2,\{\phi_1,\phi_3\}))-\sfm_2(\sfb\phi_2,\{\phi_1,\phi_3\}))
                \\
                &\kern4cm-(-1)^{|\phi_2|}\sfm_2(\phi_2,\sfb(\{\phi_1,\phi_3\}))\big]
                \\
                &\kern1cm=\ (-1)^{|\phi_1|+1}\sfb\big(\sfPoiss(\phi_1,\phi_2,\phi_3)-\{\phi_1,\sfm_2(\phi_2,\phi_3)\}\big)
                \\
                &\kern2cm-(-1)^{|\phi_1|+1}\big[\sfm_2(\{\sfb\phi_1,\phi_2\},\phi_3)+(-1)^{|\phi_1|}\sfm_2(\{\phi_1,\sfb\phi_2\},\phi_3)
                \\
                &\kern4cm-(-1)^{|\phi_1|+|\phi_2|+1}\sfm_2(\{\phi_1,\phi_2\},\sfb\phi_3)\big]
                \\
                &\kern2cm+(-1)^{(|\phi_1|+1)(|\phi_2|+1)}\big[\sfm_2(\sfb\phi_2,\{\phi_1,\phi_3\}))-(-1)^{|\phi_2|}\sfm_2(\phi_2,\{\sfb\phi_1,\phi_3\})
                \\
                &\kern4cm-(-1)^{|\phi_1|+|\phi_2|}\sfm_2(\phi_2,\{\phi_1,\sfb\phi_3\})\big]
                \\
                &\kern1cm=\ (-1)^{|\phi_1|+1}\sfb\big(\sfPoiss(\phi_1,\phi_2,\phi_3)-\{\phi_1,\sfm_2(\phi_2,\phi_3)\}\big)
                \\
                &\kern2cm+(-1)^{|\phi_1|+1}\big[\sfPoiss(\sfb\phi_1,\phi_2,\phi_3)-\{\sfb\phi_1,\sfm_2(\phi_2,\phi_3)\}\big]
                \\
                &\kern2cm-\big[\sfPoiss(\phi_1,\sfb\phi_2,\phi_3)-\{\phi_1,\sfm_2(\sfb\phi_2,\phi_3)\}\big]
                \\
                &\kern2cm-(-1)^{|\phi_2|}\big[\sfPoiss(\phi_1,\phi_2,\sfb\phi_3)-\{\phi_1,\sfm_2(\phi_2,\sfb\phi_3)\}\big]
                \\
                &\kern1cm=\ (-1)^{|\phi_1|+1}\big[\sfb(\sfPoiss(\phi_1,\phi_2,\phi_3))+\sfPoiss(\sfb\phi_1,\phi_2,\phi_3)
                \\
                &\kern2cm+(-1)^{|\phi_1|}\sfPoiss(\phi_1,\sfb\phi_2,\phi_3)+(-1)^{|\phi_1|+|\phi_2|}\sfPoiss(\phi_1,\phi_2,\sfb\phi_3)\big]
                \\
                &\kern2cm-\{\phi_1,\{\phi_2,\phi_3\}\}~,
            \end{aligned}
        \end{equation}
        where we have repeatedly made use of the definition~\eqref{eq:derived_bracket} of the derived bracket and the fact that $\sfb$ is a derivation for the derived bracket as shown in \cref{prop:algebraRelationsDB}. Hence,
        \begin{equation}\label{eq:JacobiPoissonRelation}
            \begin{aligned}
                \sfJac(\phi_1,\phi_2,\phi_3)\ &=\ (-1)^{|\phi_1|+1}\big[\sfb(\sfPoiss(\phi_1,\phi_2,\phi_3))+\sfPoiss(\sfb\phi_1,\phi_2,\phi_3)
                \\
                &\kern1cm+(-1)^{|\phi_1|}\sfPoiss(\phi_1,\sfb\phi_2,\phi_3)+(-1)^{|\phi_1|+|\phi_2|}\sfPoiss(\phi_1,\phi_2,\sfb\phi_3)\big]~.
            \end{aligned}
        \end{equation}
        So the shifted Poisson identity~\eqref{eq:BV_GB_Poisson} implies the shifted Jacobi identity~\eqref{eq:BV_GB_Jacobi}.
        
        \paragraph{\cref{prop:restrictedKinematicLieModule}.} 
        To show that $\frMod^0(V)\coloneqq(\ker\sfb_V)[1]$ is a module over the dg Lie algebra $\frKin^0(\frB)\coloneqq (\ker\sfb_\frB)[1]$, it suffices to show for every $\phi\in\ker\sfb_\frB$ and $v\in\ker\sfb_V$ that $\phi[1]\acton_\frV v[1]=(-1)^{|\phi|}\{\phi,v\}_V[1]$ is an element of $(\ker\sfb_V)[1]$, i.e.~$\sfb_V\{\phi,v\}_V=0$:
        \begin{equation}
			\begin{aligned}
                \sfb_V\{\phi,v\}_V\ &=\ \sfb_V\left(\sfb_V(\phi\acton_\frV v)-(\sfb_\frB \phi)\acton_\frV v - (-1)^{|\phi|}\phi\acton_\frV (\sfb_V v)\right)\ =\ \sfb_V^2(\phi\acton_\frV v)
                \\
                \ &=\ 0~.
			\end{aligned}
		\end{equation}
        
        \paragraph{Cyclicity in the tensor product of $\BVbox$-algebras.}
        Consider the tensor product of two $\BVbox$-algebras $\frB_\rmL$ and $\frB_\rmR$ as defined in~\eqref{eq:ordinary_tensor_product}. We now verify the properties of the metric. Firstly, we have 
        \begin{equation}
            \begin{aligned}
                &\inner{\phi_{2\rmL}\otimes\phi_{2\rmR}}{\phi_{1\rmL}\otimes\phi_{1\rmR}}
                \\
                &\kern.5cm=\ (-1)^{|\phi_{2\rmR}||\phi_{1\rmL}|+n_\rmR(|\phi_{1\rmL}|+|\phi_{2\rmL}|)}\inner{\phi_{2\rmL}}{\phi_{1\rmL}}_\rmL\inner{\phi_{2\rmR}}{\phi_{1\rmR}}_\rmR
                \\
                &\kern.5cm=\ (-1)^{(|\phi_{1\rmL}|+|\phi_{1\rmR}|)(|\phi_{2\rmL}|+|\phi_{2\rmR}|)+|\phi_{1\rmR}||\phi_{2\rmL}|+n_\rmR(|\phi_{1\rmL}|+|\phi_{2\rmL}|)}\inner{\phi_{1\rmL}}{\phi_{2\rmL}}_\rmL\inner{\phi_{1\rmR}}{\phi_{2\rmR}}_\rmR
                \\
                &\kern.5cm=\ (-1)^{(|\phi_{1\rmL}|+|\phi_{1\rmR}|)(|\phi_{2\rmL}|+|\phi_{2\rmR}|)}\inner{\phi_{1\rmL}\otimes\phi_{1\rmR}}{\phi_{2\rmL}\otimes\phi_{2\rmR}}
            \end{aligned}
        \end{equation}
        for all $\phi_{1\rmL},\phi_{2\rmL}\in\frB_\rmL$ and $\phi_{1\rmR},\phi_{2\rmR}\in \frB_\rmR$, establishing graded symmetry. Next, we verify the axioms~\eqref{eq:axiomsMetric}. In particular, using the definition of $\hat\sfd$ from~\eqref{eq:ordinary_tensor_product}, we find
        \begin{equation}
            \begin{aligned}
                &\inner{\hat\sfd(\phi_{1\rmL}\otimes\phi_{1\rmR})}{\phi_{2\rmL}\otimes\phi_{2\rmR}}
                \\
                &\kern.5cm=\ \inner{\sfd_\rmL\phi_{1\rmL}\otimes\phi_{1\rmR}}{\phi_{2\rmL}\otimes\phi_{2\rmR}}+(-1)^{|\phi_{1\rmL}|}\inner{\phi_{1\rmL}\otimes\sfd_\rmR\phi_{1\rmR}}{\phi_{2\rmL}\otimes\phi_{2\rmR}}
                \\
                &\kern.5cm=\ (-1)^{|\phi_{1\rmR}||\phi_{2\rmL}|+n_\rmR(|\phi_{1\rmL}|+|\phi_{2\rmL}|+1)}\inner{\sfd_\rmL\phi_{1\rmL}}{\phi_{2\rmL}}_\rmL\inner{\phi_{1\rmR}}{\phi_{2\rmR}}_\rmR
                \\
                &\kern1.5cm+(-1)^{|\phi_{1\rmL}|+(|\phi_{1\rmR}|+1)|\phi_{2\rmL}|+n_\rmR(|\phi_{1\rmL}|+|\phi_{2\rmL}|)}\inner{\phi_{1\rmL}}{\phi_{2\rmL}}_\rmL\inner{\sfd_\rmR\phi_{1\rmR}}{\phi_{2\rmR}}_\rmR
                \\
                &\kern.5cm=\ -(-1)^{|\phi_{1\rmL}|+|\phi_{1\rmR}||\phi_{2\rmL}|+n_\rmR(|\phi_{1\rmL}|+|\phi_{2\rmL}|+1)}\inner{\phi_{1\rmL}}{\sfd_\rmL\phi_{2\rmL}}_\rmL\inner{\phi_{1\rmR}}{\phi_{2\rmR}}_\rmR
                \\
                &\kern1.5cm-(-1)^{|\phi_{1\rmL}|+|\phi_{1\rmR}|+(|\phi_{1\rmR}|+1)|\phi_{2\rmL}|+n_\rmR(|\phi_{1\rmL}|+|\phi_{2\rmL}|)}\inner{\phi_{1\rmL}}{\phi_{2\rmL}}_\rmL\inner{\phi_{1\rmR}}{\sfd_\rmR\phi_{2\rmR}}_\rmR
                \\
                &\kern.5cm=\ -(-1)^{|\phi_{1\rmL}|+|\phi_{1\rmR}|}\inner{\phi_{1\rmL}\otimes\phi_{1\rmR}}{\sfd_\rmL\phi_{2\rmL}\otimes\phi_{2\rmR}}
                \\
                &\kern1.5cm-(-1)^{|\phi_{1\rmL}|+|\phi_{1\rmR}|+|\phi_{2\rmL}|}\inner{\phi_{1\rmL}\otimes\phi_{1\rmR}}{\phi_{2\rmL}\otimes\sfd_\rmR\phi_{2\rmR}}
                \\
                &\kern.5cm=\ -(-1)^{|\phi_{1\rmL}|+|\phi_{1\rmR}|}\inner{\phi_{1\rmL}\otimes\phi_{1\rmR}}{\hat \sfd(\phi_{2\rmL}\otimes\phi_{2\rmR})}
            \end{aligned}
        \end{equation}
        again for all $\phi_{1\rmL},\phi_{2\rmL}\in \frB_\rmL$ and $\phi_{1\rmR},\phi_{2\rmR}\in \frB_\rmR$, which verifies the first relation in~\eqref{eq:axiomsMetric}. A similar calculation for $\hat\sfb$ establishes the last relation in~\eqref{eq:axiomsMetric}. It remains to verify the second relation in~\eqref{eq:axiomsMetric}. Using the definition of $\sfm_2$ from~\eqref{eq:ordinary_tensor_product}, we find 
        \begin{equation}
            \begin{aligned}
                &\inner{\hat \sfm_2(\phi_{1\rmL}\otimes\phi_{1\rmR},\phi_{2\rmL}\otimes\phi_{2\rmR})}{\phi_{3\rmL}\otimes\phi_{3\rmR}}
                \\
                &\kern.5cm=\ (-1)^{|\phi_{1\rmR}||\phi_{2\rmL}|}\inner{\sfm_{2\rmL}(\phi_{1\rmL},\phi_{2\rmL})\otimes\sfm_{2\rmR}(\phi_{1\rmR},\phi_{2\rmR})}{\phi_{3\rmL}\otimes\phi_{3\rmR}}
                \\
                &\kern.5cm=\ (-1)^{|\phi_{1\rmR}||\phi_{2\rmL}|+(|\phi_{1\rmR}|+|\phi_{2\rmR}|)|\phi_{3\rmL}|+n_\rmR(|\phi_{1\rmL}|+|\phi_{2\rmL}|+|\phi_{3\rmL}|)}
                \\
                &\kern1.5cm\times\inner{\sfm_{2\rmL}(\phi_{1\rmL},\phi_{2\rmL})}{\phi_{3\rmL}}_\rmL\inner{\sfm_{2\rmR}(\phi_{1\rmR},\phi_{2\rmR})}{\phi_{3\rmR}}_\rmR
                \\
                &\kern.5cm=\ (-1)^{|\phi_{1\rmR}||\phi_{2\rmL}|+(|\phi_{1\rmR}|+|\phi_{2\rmR}|)|\phi_{3\rmL}|+|\phi_{1\rmL}||\phi_{2\rmL}|+|\phi_{1\rmR}||\phi_{2\rmR}|+n_\rmR(|\phi_{1\rmL}|+|\phi_{2\rmL}|+|\phi_{3\rmL}|)}
                \\
                &\kern1.5cm\times\inner{\phi_{2\rmL}}{\sfm_{2\rmL}(\phi_{1\rmL},\phi_{3\rmL})}_\rmL\inner{\phi_{2\rmR}}{\sfm_{2\rmR}(\phi_{1\rmR},\phi_{3\rmR})}_\rmR
                \\
                &\kern.5cm=\ (-1)^{(|\phi_{1\rmL}|+|\phi_{1\rmR}|)(|\phi_{2\rmL}|+|\phi_{2\rmR}|)+|\phi_{1\rmR}||\phi_{3\rmL}|}
                \\
                &\kern1.5cm\times\inner{\phi_{2\rmL}\otimes\phi_{2\rmR}}{\sfm_{2\rmL}(\phi_{1\rmL},\phi_{3\rmL})\otimes\sfm_{2\rmR}(\phi_{1\rmR},\phi_{3\rmR})}
                \\
                &\kern.5cm=\ (-1)^{(|\phi_{1\rmL}|+|\phi_{1\rmR}|)(|\phi_{2\rmL}|+|\phi_{2\rmR}|)}\inner{\phi_{2\rmL}\otimes\phi_{2\rmR}}{\hat\sfm_2(\phi_{1\rmL}\otimes\phi_{1\rmR},\phi_{3\rmL}\otimes\phi_{3\rmR})}~.
            \end{aligned}
        \end{equation}
        
    \end{body}
    
\end{document}